\documentclass[lettersize,journal]{IEEEtran}
\usepackage{amsmath,amsfonts}
\usepackage{algorithmic}
\usepackage{algorithm}
\usepackage{array}
\usepackage[caption=false,font=normalsize,labelfont=sf,textfont=sf]{subfig}
\usepackage{textcomp}
\usepackage{stfloats}
\usepackage{url}
\usepackage{verbatim}
\usepackage{graphicx}
\usepackage{cite}
\usepackage{amsthm}
\usepackage[table]{xcolor}
\usepackage{multirow} 
\usepackage{booktabs} 
\usepackage{footnote}
\makesavenoteenv{algorithm}
\allowdisplaybreaks[4]
\begin{document}

\title{Decentralized Federated Averaging \\via Random Walk}

\author{Changheng~Wang,~\IEEEmembership{Student Member,~IEEE,}  Zhiqing~Wei,~\IEEEmembership{Member,~IEEE,} Lizhe~Liu, Qiao~Deng, Yingda Wu,\\ Yangyang Niu,~\IEEEmembership{Student Member,~IEEE,} Yashan~Pang, and Zhiyong~Feng,~\IEEEmembership{Senior Member,~IEEE}
	
\thanks{
	\textit{Corresponding author: Zhiqing Wei.}
	
	Changheng Wang is with Beijing University of Posts and Telecommunications, Beijing 100876, China, and also with the National Key Laboratory of Advanced Communication Networks, Shijiazhuang 050081, China (email: ch\_wang@bupt.edu.cn).
	
	Zhiqing Wei, Qiao Deng, Yingda Wu, Yangyang Niu, and Zhiyong Feng are with Beijing University of Posts and Telecommunications, Beijing 100876, China (e-mail: weizhiqing@bupt.edu.cn; qxqx233@bupt.edu.cn; wydwydwyd@bupt.edu.cn; niuyy@bupt.edu.cn; fengzy@bupt.edu.cn).

    Lizhe Liu, Yashan Pang are with the National Key Laboratory of Advanced Communication Networks, Shijiazhuang 050081, China (e-mail: liu\_lizhe@sina.com; yashanpang@163.com).
}}

\markboth{}%
{Shell \MakeLowercase{\textit{et al.}}: A Sample Article Using IEEEtran.cls for IEEE Journals}

\maketitle

\begin{abstract}
Federated Learning (FL) is a communication-efficient distributed machine learning method that allows multiple devices to collaboratively train models without sharing raw data. FL can be categorized into centralized and decentralized paradigms. The centralized paradigm relies on a central server to aggregate local models, potentially resulting in single points of failure, communication bottlenecks, and exposure of model parameters. In contrast, the decentralized paradigm, which does not require a central server, provides improved robustness and privacy. The essence of federated learning lies in leveraging multiple local updates for efficient communication. However, this approach may result in slower convergence or even convergence to suboptimal models in the presence of heterogeneous and imbalanced data. To address this challenge, we study decentralized federated averaging via random walk (DFedRW), which replaces multiple local update steps on a single device with random walk updates. Traditional Federated Averaging (FedAvg) and its decentralized versions commonly ignore stragglers, which reduces the amount of training data and introduces sampling bias. Therefore, we allow DFedRW to aggregate partial random walk updates, ensuring that each computation contributes to the model update. To further improve communication efficiency, we also propose a quantized version of DFedRW. We demonstrate that (quantized) DFedRW achieves convergence upper bound of order $\mathcal{O}(\frac{1}{k^{1-q}})$ under convex conditions. Furthermore, we propose a sufficient condition that reveals when quantization balances communication and convergence. Numerical analysis indicates that our proposed algorithms outperform (decentralized) FedAvg in both convergence rate and accuracy, achieving a 38.3\% and 37.5\% increase in test accuracy under high levels of heterogeneities, without increasing communication costs for the busiest device.
\end{abstract}


\section{Introduction}
\subsection{Motivations}
\IEEEPARstart{D}{istributed} machine learning disperses data and computation across multiple devices, which can both improve computational efficiency and efficiently utilize the distributed resources of the system. This enables the processing of large-scale datasets and the training of complex models. However, traditional distributed machine learning methods face major challenges, including the need to transfer private data between devices and low communication efficiency, both of which restrict system performance. Federated Learning (FL), as one of the primary frameworks of distributed optimization, allows the training of local models on multiple distributed devices. It updates the global model by sharing local model parameters or gradients with a central server or neighboring devices \cite{1Peter}, achieving collaborative learning while ensuring data privacy and demonstrating excellent scalability.

The primary advantages of FL lie in its high protection of data privacy and security while still leveraging distributed computing resources for efficient model training. Additionally, FL significantly reduces the bandwidth required for data transmission by transferring model updates instead of raw data. By performing multiple epochs of local updates, it makes full use of distributed computing resources and improves communication efficiency. FL has been widely researched and applied in fields such as healthcare \cite{2Rieke}, \cite{3Wu}, smart cities \cite{4Zeng}, and financial services \cite{5Zheng}.

FL includes centralized and decentralized paradigms, differentiated by the presence or absence of a central server to coordinate and aggregate local model updates distributed across multiple devices \cite{6Li}. As the number of participating devices increases, the communication load between the central server and local devices significantly rises, creating a communication bottleneck that leads to a decline in overall system performance. The central server also faces high demands on computational and storage resources due to the frequent reception and aggregation of numerous model updates from local devices and the need to maintain system consistency. Moreover, single point of failure (SPOF) is a critical issue faced by centralized FL (CFL) \cite{7Korkmaz}. To overcome these challenges, decentralized FL (DFL) eliminates the reliance on a central server, utilizing the Peer-to-Peer (P2P) network structure to distribute the processing and aggregation of model updates. The decentralized paradigm not only enhances system robustness and scalability, better adapting to dynamic environments, but also further improves privacy protection, reduces communication load, and mitigates the imbalance of computational resources \cite{7Imteaj}.

\begin{figure}[!t]
	\centering
	\includegraphics[width=3in]{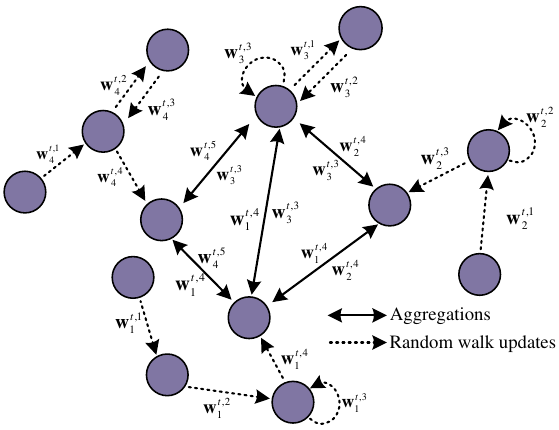}
	\caption{Decentralized federated averaging via random walk. The model of the $i$-th update in the $p$-th random walk chain during the $t$-th global round is denoted as $\mathbf{w}_p^{t,i}$. The self-loops at $p=1$, $i=3$, $p=i=2$ and $p=i=3$ indicate that after this update, the next update remains on the same device. The model undergoes decentralized aggregation periodically.}
	\label{fig_1}
\end{figure}

Due to factors such as user behavior, environment, hardware, and software differences, the local data on distributed devices in a federated network is not always independently and identically distributed (IID) \cite{7Jianyu}. Multiple local stochastic gradient descent (SGD) iterations may cause the local models on devices to update in a more personalized direction, deviating from the global optimal objective, which may lead to slower convergence of the global model and potentially result in convergence to a suboptimal model \cite{8Mao}. To address the poor performance of DFL on data with statistical heterogeneity and distribution uncertainty, multiple local model updates on devices should not be restricted to local devices. Traversing diverse local data helps improve the consistency between local and global models.

The random walk SGD sends the model to a random neighbor for sequential updates, effectively utilizing more heterogeneous data while ensuring that every computation and communication cost contributes to the improvement of global model \cite{9Sun}. Therefore, we employ a DFL algorithm by utilizing the random walk method. In each iteration, the model is updated at a particular device and then sent to a randomly selected neighbor for the next update using its local data. Periodically, the local model is sent to a set of random neighbors for aggregation. Fig. \ref{fig_1} illustrates how our proposed framework operates. This approach addresses the statistical heterogeneity of data across local devices while balancing time efficiency with communication and computation costs. Furthermore, performing multiple iterations of local model updates using random walks results in lower communication overhead compared to traditional decentralized SGD (DSGD) methods, which aggregate immediately after a single local update, as depicted in Fig. \ref{fig_2} (a).

The imbalance in data across local devices, coupled with the variations in hardware, battery life, and network connectivity within a federated network, causes significant heterogeneity among devices \cite{30Liao}. This system heterogeneity manifests in differences in communication, computation, and storage capabilities, which intensify the straggler effect. When local data is statistically heterogeneous, stragglers may possess unique data characteristics. Ignoring these stragglers can result in sampling bias and reduce the number of devices available for training, negatively impacting convergence. To address this issue, we propose to allow each set of devices executing the random walk to perform a variable number of random walks according to their computational capabilities. By dynamically adjusting the workload, we can balance the load among devices, thereby mitigating the straggler effect.

We focus on two critical issues in DFL. First, while multiple local updates can effectively reduce communication costs and enhance training efficiency in homogeneous systems, their benefits may diminish or even become counterproductive as heterogeneities increase. Therefore, it is essential to develop DFL methods that can adapt to these heterogeneities. Second, as the number of participating devices or the amount of model parameters increases, the communication cost among local devices rises sharply. Designing effective communication strategies is another key objective of this paper.

The proposed approach is applicable to various distributed learning environments. For instance, in scenarios such as personalized recommendation systems, collaborative unmanned aerial vehicle (UAV) sensing, and vehicular networks, devices collect highly heterogeneous data. The local update strategy based on random walks can mitigate data drift. Similarly, in computationally heterogeneous environments such as industrial internet of things (IIoT) and smart cities, the adaptive adjustment of random walk lengths enables efficient utilization of computational resources. Moreover, in decentralized settings like post-disaster emergency networks, decentralized finance (DeFi), and distributed medical systems, the absence of a reliable central server or privacy concerns may prevent uploading private models. In such cases, decentralized aggregation facilitates efficient model aggregation, prevents data leakage, and enhances system robustness. Finally, in latency-sensitive applications such as online model updates for autonomous driving and network intrusion detection, communication overhead often becomes a bottleneck. The quantization technique based on parameter differences can significantly reduce transmission costs, while the periodic aggregation of multiple random walk chains accelerates global model convergence.

\begin{figure*}[!t]
	\centering
	\includegraphics[width=6.7in]{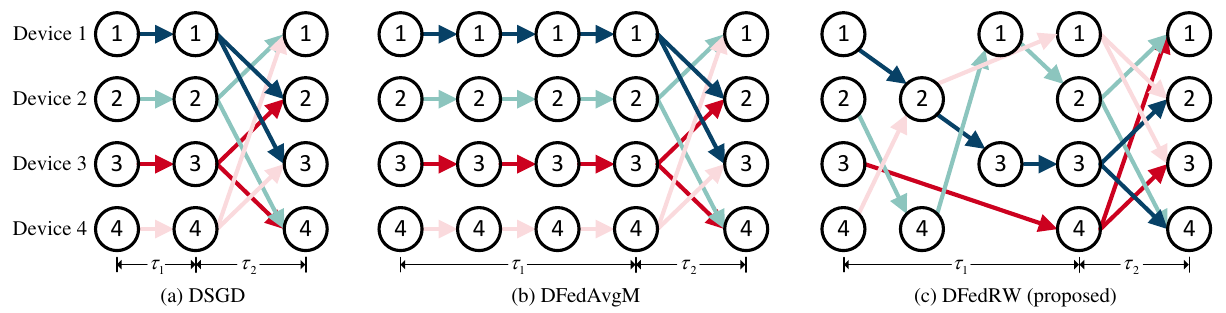}
	\caption{Comparison of communication and training styles of DSGD, DFedAvgM and the proposed DFedRW. Different colored chains represent different model update trajectories. In process $\tau_1$, DSGD and DFedAvgM perform one or multiple local updates on a single device. DFedRW employs a random walk approach to perform a variable number of updates to address statistical and system heterogeneity. In process $\tau_2$, inter-device communications take place to aggregate the global model.}
	\label{fig_2}
\end{figure*}

\subsection{Our Contributions}
We are interested in decentralized settings where devices have limited computation and communication capabilities. To restrict computational costs, we focus on first-order methods, where model updates involve only the gradients of the local loss functions and do not depend on higher-order derivatives. We outline our contributions in the following three aspects.

\begin{itemize}
	\item Algorithmically, we propose a decentralized distributed learning algorithm, DFedRW, to learn models from distributed data on devices within a heterogeneous network. Instead of restricting local epochs to a single device, our algorithm completes multiple local updates using a random walk method. The number of updates is determined by the communication and computational resources of the devices. Additionally, we introduce a quantized version of DFedRW to further reduce communication costs between devices.
	\item Theoretically, we prove that DFedRW and its quantized version can achieve a convergence upper bound of order $\mathcal{O}(\frac{1}{k^{1-q}})$ in convex settings, approaching the $\mathcal{O}(\frac{1}{\sqrt{k}})$ level of SGD. We also demonstrate that their convergence bound depends on the level of statistical and system heterogeneity, with greater heterogeneities leading to a closer bound between the quantized and original versions. Furthermore, we provide sufficient conditions for balancing communication and convergence performance.
	\item Empirically, we conduct extensive numerical experiments by training deep neural networks (DNNs) under various statistical and system heterogeneity settings. The numerical results show that our proposed (quantized) DFedRW achieves faster convergence and higher accuracy compared to the baselines, and its advantages are amplified in scenarios with high levels of heterogeneities and sparse graphs. Moreover, these benefits are achieved without increasing the communication cost for the busiest nodes.
\end{itemize}

\subsection{Organizations}
The rest of the paper is organized as follows. Section II provides a brief summary of the related literature. Section III outlines the mathematical formulation of the problem. In Section IV, we introduce the DFedRW algorithm along with its quantized version. The convergence analysis of the proposed algorithms is presented in Section V. Section VI gives the numerical results and their detailed analysis. Finally, the paper concludes with Section VII. The complete technical proofs are provided in the appendix.

\subsection{Notation}
Let lowercase, lowercase bold, and uppercase bold letters represent scalars, vectors, and matrices, respectively. For a vector $\mathbf{w} \in \mathbb{R}^d$, the local copy at device $i$ is denoted as $\mathbf{w}_i$. For a matrix $\mathbf{X}=(x_{i, j})_{N \times M}$, $\mathbf{X}^\top$ represents its transpose, $\mathbf{X}(i,:)$ represents its $i$-th row, and $\|\mathbf{X}\|:=\sqrt{\sum_{i=1}^N \sum_{j=1}^M x_{i, j}^2}$ denotes its Frobenius norm. When $N$$=$$M$$=$$1$, $\|\cdot\|$ is the $l_2$ norm of a vector. We denote the $i$-th eigenvalue of a square matrix $\mathbf{H}$ as $\lambda_i(\mathbf{H})$. For a function $f(\mathbf{w}): \mathbb{R}^d \rightarrow \mathbb{R}$, $\nabla f(\mathbf{w})$ represents its gradient. $\mathbb{E}[\cdot]$ represents the expectation taken with respect to the given probability space. Given a scalar $x$, $\operatorname{sgn}(x)$ indicates its sign, $\lceil x \rceil$ and $\lfloor x \rfloor$ denote the ceiling and floor functions, respectively.

\section{Related Works}
We conducted a review of three closely related areas to this work: decentralized federated learning, statistical and system heterogeneity, and random walk learning. For clarity, we provide a comparative overview of related work in Table \ref{table_0}.

\begin{table}[!t]
	\centering
	\caption{Comparison of Related Work in DFL.}
	\label{table_0}
	\renewcommand{\arraystretch}{1}
	\begin{tabular}{lll}
		\toprule
		Taxonomy & References & Solution Categorization \\ 
		\midrule
		\multirow{4}{*}{Efficiency} & \cite{15Zhou}  & Lazy updates \\
		& \cite{16Tang} & Sparsification \\
		& \cite{14Sun}, \cite{17Chen} & Quantization \\
		& \cite{14Sun}, \cite{34Sun} & Momentum acceleration \\
		\midrule
		\multirow{2}{*}{Topology} & \cite{19Koloskova} & Adaptive network topology \\
		& \cite{20Bellet} & Grouped topology \\
		\midrule
		& \cite{8Mao} & Proximal regularization \\
		Statistical & \cite{25Ayache}, \cite{32Sun}, \cite{33Ayache} & Random walk learning \\
		heterogeneity & \cite{27Mendieta} & Standard regularization \\
		& \cite{30Liao} & Consensus-based connectivity  \\
		\midrule
		System & \cite{29Shi} & Local consistency optimization \\
		heterogeneity & \cite{30Liao} & Adaptive update frequency \\
		\bottomrule
	\end{tabular}
\end{table}

\textit{Decentralized federated learning}. DFL is a learning paradigm built upon federated learning and decentralized training. It aims to enhance communication efficiency and protect private data from being sent to untrusted parameter servers \cite{11Hashemi}. \cite{12Heged} and \cite{13Lian} provide a comprehensive empirical study comparing centralized and decentralized FL in terms of convergence rates, bandwidth consumption, and computational complexity. In \cite{14Sun}, the use of momentum in decentralized federated averaging (DFedAvgM) is proposed, as illustrated in Figure 2 (b). This method establishes convergence guarantees for decentralized federated learning with multiple local iterations under non-convex and Polyak-Łojasiewicz (PŁ) conditions. In addressing the challenges of optimizing the efficiency of DFL, lazy, sparse, and quantized techniques improve system performance by reducing the consumption of computational resources and lowering the demand on communication bandwidth. Zhou \textit{et al}. \cite{15Zhou} introduced lazily aggregated gradients (LAG) for non-convex FL to reduce uplink communication overhead. In \cite{16Tang}, a communication-efficient DFL framework is proposed, which requires only the exchange of information with peers in a highly sparsified model. Chen \textit{et al}. \cite{17Chen} proposed an adaptive non-uniform quantization level adjustment scheme to minimize quantization distortion in DFL, demonstrating convergence without assuming convex loss. The topology of DFL networks directly impacts model communication overhead, convergence, and generalization \cite{18Vogels}. It has been shown in \cite{19Koloskova} that under low-noise conditions, the similarity of functions and communication topology affects the convergence performance of DSGD. In \cite{20Bellet}, a novel sparse topology is introduced, which groups sparsely connected nodes to address label distribution skew by replacing global labels with intra-group labels, achieving convergence rate similar to fully connected network. More details on DFL can be found in \cite{21Edoardo}, while the latest advancements, additional applications, and outstanding issues are discussed in \cite{22Li}, \cite{23Martínez}, and \cite{24Qu}.

\begin{table*}[!t]
	\centering
	\small 
	\caption{\label{table_1}Advantages of DFedRW Compared to FedAvg and DFedAvg.}
	\renewcommand{\arraystretch}{1.4} 
	\setlength{\tabcolsep}{6pt} 
	\begin{tabular}{>{\raggedright\arraybackslash}p{3cm} >{\raggedright\arraybackslash}p{4.5cm} >{\raggedright\arraybackslash}p{4.5cm} >{\raggedright\arraybackslash}p{4.5cm}}
		\toprule
		& \textbf{\textsf{FedAvg}} \cite{28McMahan} & \textbf{\textsf{DFedAvgM}} \cite{14Sun} & \textbf{\textsf{DFedRW (proposed)}} \\ 
		\midrule
		\textbf{\textsf{1. Robustness}} & SPOF can disrupt the FL system. & Partial device failures do not impact the overall functionality. & Partial device failures do not impact the overall functionality. \\
		\rowcolor[HTML]{F7F7F7} 
		\textbf{\textsf{2. Data privacy}} & Attacking the server exposes all device parameters or gradients. & Each device only has parameters or gradients of some neighbors. & Each device only has parameters or gradients of some neighbors. \\
		\textbf{\textsf{3. Statistical heterogeneity}} & Multiple local updates drive the model towards local optima. & Multiple local updates drive the model towards local optima. & Random walk-based updates steer the model to a global optimum. \\
		\rowcolor[HTML]{F7F7F7} 
		\textbf{\textsf{4. System heterogeneity}} & Dropping stragglers affects data integrity and model convergence. & Dropping stragglers affects data integrity and model convergence. & Integrating partial contributions from stragglers. \\
		\bottomrule
	\end{tabular}
\end{table*}

\textit{Heterogeneities}. In the DFL system, the inconsistent data distribution across user devices results in statistical heterogeneity among local data. Multiple local iterations can cause local models to diverge from the global model and slow down the convergence rate \cite{25Ayache}. To address this issue, many works have introduced proximal terms \cite{26Li,8Mao} to combat data heterogeneity during model training. However, some added proximal terms may lead to a significant increase in computational or storage costs. Mendieta \textit{et al}. \cite{27Mendieta} analyzed the effectiveness of standard regularization methods in countering data heterogeneity from the perspective of machine learning training principles. Differences in communication, computation, and storage capabilities contribute to device heterogeneity, potentially causing communication delays or computational bottlenecks, which can impact training efficiency. Naively ignoring stragglers can increase data heterogeneity \cite{26Li}. Therefore, statistical and system heterogeneity interact with each other. Shi \textit{et al}. \cite{29Shi} proposed a scheme to enhance the consistency between local models and local flat models, achieving a trade-off between communication complexity and generalization performance. Liao \textit{et al}. \cite{30Liao} addressed system and statistical heterogeneity by integrating adaptive network topology and local update frequency, and established convergence upper bounds related to both aspects. 

\textit{Random walk learning}. As a decentralized learning algorithm, random walk learning achieves the seemingly opposing goals of extending the benefits of locality and mitigating statistical heterogeneity \cite{25Ayache}. Sun \textit{et al}. \cite{32Sun} established the convergence of irreversible finite-state Markov chains for non-convex problems. Ayache \textit{et al}. \cite{33Ayache} examined random walks under uniform, static, and adaptive weights, deriving non-asymptotic convergence bounds related to the Lipschitz constant and graph. In \cite{34Sun}, convergence results are established for random walk SGD with adaptive step sizes and momentum under mild assumptions. Sun \textit{et al}. \cite{9Sun} introduced a decentralized Markov chain that only communicates intermediate results with neighbors, elucidating the dependency between topology and mixing time. Shah \textit{et al}. \cite{35Shah} investigated the convergence of the alternating direction method of multipliers (ADMM) updated iteratively in a random walk manner. Since random walk is serially executed along a path, where model parameters or gradients are passed from one device to the next for updates, this method achieve efficient communication at the expense of increased algorithm runtime. To enhance communication efficiency and accelerate the convergence of the Markov chain, \cite{36Ye} proposes parallel random walk ADMM.

This work distinguishes itself by presenting a communication-efficient decentralized federated training framework that combats arbitrary system and statistical heterogeneity among local devices. By employing parallel random walk SGD to update the DFL model parameters, it achieves a balance between communication, computation, and convergence bound. In Fig. 2, we compare the proposed decentralized federated averaging via random walk (DFedRW) with traditional distributed learning strategies. Furthermore, we introduce a quantized DFedRW that effectively reduces communication costs during the update and aggregation processes. We detail the convergence results of the proposed algorithms under convex conditions, revealing the relationship between convergence bounds and heterogeneities. In Table \ref{table_1}, we summarize the advantages of DFedRW compared to FedAvg \cite{28McMahan} and DFedAvgM \cite{14Sun}. Moreover, we propose sufficient conditions for optimizing the trade-off between communication efficiency and convergence bound in quantized implementations.

\section{Problem Formulation}
\subsection{Communication Graph}
We model the network structure in the consensus optimization problem using an undirected graph $\mathcal{G}=(\mathcal{V}, \mathcal{E})$, where $\mathcal{V}=[n]:=\{1,2,\ldots,n\}$ represents the set of $n$ devices, and $\mathcal{E}=\{(i,j)\in\mathcal{V}\times\mathcal{V}\}$ denotes the set of edges connecting neighboring devices. The set of neighbors for device $i\in\mathcal{V}$ is denoted by $\mathcal{N}(i)$. The undirected nature of the graph implies that if $\forall(i,j)\in\mathcal{E}$, then $(j,i)\in\mathcal{E}$. Additionally, since the devices allow self-loops, we have $(i,i)\in\mathcal{E}$, $\forall i$.

\subsection{Learning Objective}
Our objective is to minimize the empirical average of local losses over distributed data on a graph with $n$ devices \cite{28McMahan}
\begin{equation}
	\label{eq_1}
    \min_{\mathbf{w} \in \mathcal{W}} f(\mathbf{w}):=\frac{1}{n} \sum_{i=1}^n \underbrace{\mathbb{E}_{\xi \sim \mathcal{D}_i} F_i(\mathbf{w} ; \xi)}_{=: f_i(\mathbf{w})},
\end{equation}
where $i$-th device maintains a local dataset $\mathcal{D}_i=\{\xi_{i, 1}, \xi_{i, 2}, \ldots, \xi_{i, n}\}$, and $\xi_{i, r}=(\mathbf{x}_{i, r}, y_{i, r}) \in \mathbb{R}^m \times \mathbb{R}$, for $\forall r \in[n]$, consists of $m$-dimensional feature vectors and one-dimensional real-value labels. $\mathcal{D}_i$ is sampled from the unknown local distribution $\Pi_i$ at $i$-th device. The global objective function is $f(\mathbf{w})$, $\mathbf{w} \in \mathcal{W} \subseteq \mathbb{R}^d$, with $\mathcal{W}$ being a closed and bounded feasible set. $f_i(\mathbf{w}):=\int_{\Pi_i} F_i(\mathbf{w}; \xi) d \Pi_i(\xi)$ denotes the local objective at $i$-th device. $F_i(\mathbf{w} ; \xi)$ represents the local loss function associated with the training data $\xi$.

The optimal model $\mathbf{w}^*$ of $f(\mathbf{w})$ satisfies \cite{28McMahan}
\begin{equation}
	\label{eq_2}
	\mathbf{w}^*=\arg \min _{\mathbf{w} \in \mathcal{W}} f(\mathbf{w}).
\end{equation}

\subsection{Heterogeneities}
In distributed networks, the statistical distribution of local data at each device is typically non-independent and identically distributed (Non-IID). Moreover, even if the local data is IID, sampling discrepancies can induce statistical heterogeneity. This heterogeneity in local data can impact the convergence bound of the global model and may cause it to diverge from the global optimum \cite{7Jianyu}. To quantify the extent of local dissimilarity, we use the variance of the expected local gradient relative to the global gradient for a given model.

\theoremstyle{definition}
\newtheorem{definition}{Definition}
\begin{definition}[$\delta^2$-local dissimilarity \cite{37Karimireddy}]
	\label{definition_1}
	The local objective function $F_i$ is $\delta^2$-locally dissimilar at $\mathbf{w}$ if
	\begin{equation}
	\delta(\mathbf{w})^2  \geq \frac{\mathbb{E}_i[\| \nabla F_i(\mathbf{w}) \|^2]}{\| \nabla f(\mathbf{w}) \|^2},
    \end{equation}
	for $\| \nabla f(\mathbf{w}) \| \neq 0$, and $\delta(\mathbf{w})\geq 1$.
\end{definition}

The dissimilarity among local functions increases as $\delta(\mathbf{w})$ becomes larger. $\mathbb{E}_i[\| \nabla F_i(\mathbf{w}) \|^2]=\| \nabla f(\mathbf{w}) \|^2$ is a uniform stationary solution for all $F_i$, and the distributed network can be regarded as a homogeneous system.

In practice, we can collect statistics on the squared local gradient norms of devices over multiple communication rounds, compute their expected value $\mathbb{E}_i [\|\nabla F_i(\mathbf{w}) \|^2]$, and then divide it by the squared global gradient norm $ \|\nabla f(\mathbf{w})\|^2 $ to estimate $\delta^2$. The global gradient norm can be estimated using the currently aggregated global model. When $\delta^2$ approaches 1, the data distribution is relatively uniform (approaching IID), which facilitates global convergence. Conversely, a big value of $\delta^2$ indicates significant statistical heterogeneity, leading to training oscillations and adversely affecting the convergence.

Due to the varying communication, computation, and storage capabilities of devices in the network, the efficiency in solving local problems differs across the devices. Therefore, we allow each random walk trajectory to approximate the local objective solution with a level of inexactness inversely proportional to the capability of device. In this paper, the random walk trajectory refers to the sequence of devices encountered during multiple local updates using the random walk SGD. The concept of a $\gamma_i^t$-inexact solution is defined as follows.

\begin{definition}[$\gamma_i^t$-inexactness \cite{26Li}]
	\label{definition_2}
	We define $\mathbf{w}^{\prime}$ as a $\gamma_i^t$-inexact solution of $\min _{\mathbf{w} \in \mathcal{W}} F_i(\mathbf{w})$ if 
		\begin{equation}
		\|\nabla F_i(\mathbf{w}^{\prime})\| \leq \gamma_i^t\|\nabla F_i(\mathbf{w}^t)\|,
		\end{equation}
		for $\gamma_i^t \in [0,1]$, and $\delta(\mathbf{w})\geq 1$.
\end{definition}

The parameter $\gamma_i^t$ represents the level of inexactness, quantifying the local computation capability of device $i$ during the $t$-th communication round. In practice, we can estimate $\gamma_i^t$ from the ratio of gradient norm reduction after local updates on device $i$, where $\mathbf{w}^{\prime}$ denotes an approximate optimal solution obtained during local optimization. A smaller $\gamma_i^t$ indicates stronger local computational capacity, allowing for more effective model updates and contributing to the stable convergence of the global model. In contrast, a larger $\gamma_i^t$ suggests that local updates are coarse, which may introduce oscillations in global optimization and hinder convergence. Following this, we will present a heterogeneity expression applicable to the analysis of DFedRW.

\newtheorem{lemma}{Lemma}
\begin{lemma}[Random walk inexactness]
	\label{lemma_1}
	A random walk update is $\hat{\gamma}$-inexact if
	\begin{equation}
	\|\nabla F_{i^k}(\mathbf{w}^k)\| \leq \hat{\gamma}\|\nabla F_{i^{k-K}}(\mathbf{w}^{k-K})\|,
	\end{equation}
	where $\hat{\gamma}=\prod_{q=k-K}^{k-1} \gamma_{i^q}^{q}$, and $K$ is the number of random walk epochs in a communication round.
\end{lemma}
\begin{proof}
	Appendix A provides a detailed proof of Lemma \ref{lemma_1}.
\end{proof}

Specifically, the inexactness parameter $\hat{\gamma}$ adjusts the local communication and computation load by influencing the number of local random walk iterations performed on each random walk trajectory during the $t$-th communication round, thereby combating the heterogeneity of devices. Notably, a smaller $\hat{\gamma}$ indicates higher accuracy in solving the local model through random walk, which is achieved by performing more frequent random walk updates.

\subsection{Random Walk Updates and Decentralized Aggregation}
\renewcommand{\thefootnote}{\dag}
We design a decentralized model update scheme on the graph $\mathcal{G}$, where each device communicates with a random neighbor within one-hop through a random walk, without the involvement of a central server. Specifically, a random walk trajectory on $\mathcal{G}$ starts from an initial device $i^0\in\mathcal{V}$, where $i^0$ is randomly and uniformly sampled from $\mathcal{V}$. At each step of the trajectory, the current device activates the next device to visit. In the $k$-th update, device $i^{k}$ uses the model parameters $\mathbf{w}^{k}$ received from the previous neighboring device and its own local data $\xi^k$ to update the local model, where $i^{k}$ is the $k$-th device on the random walk trajectory. The updated model is then projected back into the feasible set $\mathcal{W}$. We illustrate this with the first-order SGD method as an example\footnote{Although our analysis is conducted specifically with SGD, our work is also applicable to a broad range of update schemes that utilize unbiased estimations.}, with the update rule as follows \cite{14Sun}.
\begin{equation}
	\label{eq_3}
	\mathbf{w}^{k+1} \leftarrow \boldsymbol{\Pi}_{\mathcal{W}}\left(\mathbf{w}^k-\eta^k \mathbf{g}_{i^k}\right),
\end{equation}
where $\boldsymbol{\Pi}_{\mathcal{W}}$ denotes the projection operator, and $\mathbf{g}_{i^k} = \nabla F_{i^k}(\mathbf{w}^k, \xi^k)$ represents the local gradient computed at the $k$-th step (in device $i^k$) of the random walk trajectory. Since the process is inherently stochastic, to ensure the stochastic algorithm converges to the optimal solution, we employ a decreasing step size strategy \cite{26Li}, denoted by $\eta^k$.

We employ the Metropolis-Hastings (MH) \cite{25Ayache} decision rule to uniformly sample devices in the graph for random walk updates, using only local information of the devices to design the transition matrix $\mathbf{P}\in \mathbb{R}^{n \times n}$ such that the Markov chain converges to the stationary distribution $\boldsymbol{\Pi}^*:=[(\pi^*)^{\top},(\pi^*)^{\top}, \ldots,(\pi^*)^{\top}]^{\top} \in \mathbb{R}^{n \times n}$. Initially, at the $k$-th step of the random walk, a neighboring device $j$ is chosen as a candidate for the ($k+1$)-th step based on probability $a_\mu(i^{k}, j)=\min \{1, \frac{\operatorname{deg}(i^{k})}{\operatorname{deg}(j)}\}$. If device $j$ is not accepted, the random walk remains at device $i^k$ for the ($k+1$)-th step, resulting in a self-loop $i^{k+1}=i^{k}$. Otherwise, $i^{k+1}=j$. The transition matrix is denoted by \cite{33Ayache}

\begin{equation}
	\label{eq_10}
	\mathbf{P}(i^{k}, j)=\left\{\begin{array}{cc}
		\frac{a_\mu(i^{k}, j)}{\operatorname{deg}(i^{k})}, \qquad j \neq i^{k}, j \in \mathcal{N}(i^{k}), \\
		1-\sum_{j \in \mathcal{N}(i^{k})} \frac{a_\mu(i^{k}, j)}{\operatorname{deg}(i^{k})}, \qquad j=i^{k},
	\end{array}\right.
\end{equation}
where $\text{deg}(i)$ denotes the degree of device $i$. The random walk transition probabilities are designed according to the MH rule align with the stationary distribution $\boldsymbol{\Pi}^*$ \cite{8Mao}.

In the iterative update process, the random walk effectively addresses the explicit objectives of the learning model while also capturing statistical heterogeneity across different devices. By iterating (\ref{eq_3}) for $K$ times, the updated result $\mathbf{w}^{K-1}$ of the random walk trajectory is obtained.

We define a random walk on $\mathcal{G}$ using a Markov chain, where the active device $i^k$ represents a state within the Markov chain, with the state space $\mathcal{V}$. The transition matrix $\mathbf{P}$ inherits the original connectivity structure from $\mathcal{G}$.

\begin{definition}[Finite-state time-homogeneous Markov chain \cite{9Sun}]
	\label{definition_3}
    A random process $X_1, X_2, \ldots$ is referred to as a finite-state time-homogeneous Markov chain if the transition probabilities between states in a finite state space $\mathcal{V}$ remain constant over time. Its transition matrix satisfies $\mathbb{P}(X_{k+1}=j \mid X_0=i_0, \ldots, X_k=i)=\mathbb{P}(X_{k+1}=j \mid X_k=i)=\mathbf{P}_{i, j}$.
\end{definition}

Consider a strongly connected graph $\mathcal{G}$ that includes self-loops, resulting in an aperiodic and irreducible random walk $(i^k)_{k \geq 0}$, $i^k \in \mathcal{V}$, where transition matrix $\mathbf{P}_{i,j}=\mathbb{P}(i^{k+1}=j \mid i^k=i) \in[0,1]$, and the stationary distribution satisfies $\pi^*=\lim _k \pi^k=[\pi_1^*, \pi_2^*, \ldots, \pi_K^*]$ with $\sum_{i=1}^K \pi_i^*=1$, $\min _i\{\pi_i^*\}>0$, and $(\pi^*)^\top \mathbf{P}=(\pi^*)^\top$. For $(i^k)_{k \geq 0}$, achieving convergence through multiple iterations can be described using the Markov chain mixing time \cite{38Levin}.

\renewcommand{\thefootnote}{\ddag}
\begin{definition}
	\label{definition_4}
	Let $\lambda_1=1>\lambda_2 \geq \cdots \geq \lambda_K$ to represent eigenvalues\footnote{If the absolute values of all eigenvalues are less than 1, the state distribution of the random walk converges to a stationary distribution in the long term. The closer the absolute values of the eigenvalues are to 1, the slower the rate of convergence.} of the transition matrix $\mathbf{P}$ associated with the random walk, we define $
	\lambda_P:=\frac{\max \left\{|\lambda_2(\mathbf{P})|,|\lambda_K(\mathbf{P})|\right\}+1}{2} \in[0,1)$ \cite{38Levin}.
\end{definition}

\begin{lemma}[Convergence of Markov chain \cite{8Mao}]
	\label{lemma_10}
	Consider an aperiodic and irreducible random walk. Let integer ${\tau(\zeta)}$ be the mixing time, we establish the bound as
	\begin{equation}
		\max _i\|\boldsymbol{\Pi}^*-\mathbf{P}^{\tau(\zeta)}(i,:)\| \leq \zeta \lambda_{P}^{\tau(\zeta)}.
	\end{equation}
\end{lemma}

The inequality indicates that, for any time $k$, the state distribution of state $i$ after $\tau(\zeta)$ steps of the Markov chain is constrained by $\zeta \lambda_{P}^{\tau(\zeta)}$ with respect to the stationary distribution $\boldsymbol{\Pi}^*$. As $k \geq K_\mathbf{P}$, $\zeta$ is a constant related to the Jordan canonical form of $\mathbf{P}$, $K_\mathbf{P}$ is a constant that depends on $\lambda_2(\mathbf{P})$ and $\lambda_K(\mathbf{P})$ \cite{38Levin}, and ${\mathbf{P}}^{\tau(\zeta)}$ is the ${\tau(\zeta)}$-th power of ${\mathbf{P}}$.

We employ a decentralized federated averaging approach for the weighted aggregation of local models. In the $t$-th communication round, device $i$ aggregates model parameters from a subset of its neighbors and updates using a locally weighted average \cite{25Ayache}
\begin{equation}
	\label{eq_4}
	\mathbf{w}_i^{t+1} \leftarrow \sum_{l \in \mathcal{N}_A(i)} \frac{n_l^t}{m^t} \mathbf{w}_l^t,
\end{equation}
where $\mathcal{N}_A(i)$ is the subset of devices that are directly connected to device $i$ and participate in the model aggregation. $m^t=\sum_{l \in \mathcal{N}_A(i)} n_l^t$, and $n_l^t$ denotes the number of local examples on device $l$.

\section{Proposed DFedRW and Its Quantized Version}
\subsection{Decentralized Federated Averaging via Random Walk}
Unlike DFedAvgM, our DFedRW replaces local updates with random walk updates among different devices to better combat statistical heterogeneity. In contrast to other decentralized random walk approaches that utilize a single update path \cite{8Mao,9Sun,35Shah}, DFedRW incorporates multiple parallel random walk trajectories and aggregates the local modes using a decentralized communication scheme. For clarity, the notations used in DFedRW are summarized in Table \ref{table_2}.

\begin{table}[!h]
	\renewcommand{\arraystretch}{1}
	\caption{Notations Used in DFedRW}
	\label{table_2}
	\centering
	\setlength{\tabcolsep}{5mm}{
		\begin{tabular}{p{0.3cm} p{6.5cm}}
			\toprule
			\textbf{\textbf{Notation}} & \textbf{Description} \\
			\midrule
			$i^{t,k}$ & The $k$-th device on the random walk trajectory in the $t$-th communication round \\
			$\mathbf{w}_i^{t, k}$ & Local model parameters held by device $i^{t,k}$ \\
			$\hat{\mathbf{g}}_{i^{t,k}}$  & Unbiased estimation of local gradient on device $i^{t,k}$ \\
			$\eta^{\bar{k}}$ & Globally decreasing learning rate \\
			$\mathbf{w}_i^{t, 0}$ & Initial model parameters on the device $i$ in the $t$-th communication round \\
			$\mathbf{w}_l^{t, last}$ & Pre-aggregated model parameters obtained by device $l$ during its last update in the $t$-th communication round \\
			$\mathcal{N}_A(i)$ & 
			Subset of devices participating in decentralized model aggregation on device $i$ \\
			$\mathcal{N}_c(i)$ & Subset of devices that receive the local model sent by device $i$ \\
			$\widetilde{\mathbf{w}}_{i}^{t,k}$ & Local model parameters held by $i^{t,k}$ in QDFedRW \\
			$\widetilde{\mathbf{g}}_{i^{t,k}}$ & Local gradient unbiased estimation on $i^{t,k}$ in QDFedRW  \\
			$\mathbf{Q}^{t,k}(i)$ & Quantized difference of updated parameters sent by device $i$ to the next device on a random walk trajectory
			 \\
			$\mathbf{Q}^{t}(i)$ & Quantized difference of updated parameters participating in global aggregation sent by device $i$ to $\mathcal{N}_c(i)$ \\
			\bottomrule
	\end{tabular}}
\end{table}

Therefore, the random walk update mechanism in DFedRW is not as straightforward as depicted in (\ref{eq_3}). In the $t$-th communication round, the $k$-th random walk update is obtained as
\begin{equation}
	\label{eq_5}
	\mathbf{w}_i^{t, k+1} \leftarrow \boldsymbol{\Pi}_{\mathcal{W}}\left(\mathbf{w}_i^{t, k}-\eta^{\bar{k}} \hat{\mathbf{g}}_{i^{t,k}}\right),
\end{equation}
where $\hat{\mathbf{g}}_{i^{t,k}} = \hat{\nabla} F_{i^{t,k}}(\mathbf{w}_i^{t, k})$ is the unbiased estimate of the local gradient at the walking device $i^{t,k}$, and $\eta^{\bar{k}}$ denotes the globally decreasing learning rate with $\bar{k}=(t-1)K+k$. For convenience, we omit the subscript of $\mathbf{w}_{i^{t,k}}^{t,k}$ and $\mathbf{w}_{i^{t,k+1}}^{t,k+1}$. In fact, $\mathbf{w}_i^{t,k}$ should be written as $\mathbf{w}_{i^{t,k}}^{t,k}$, indicating the local model parameters held by device $i^{t,k}$.

The main distinction between (\ref{eq_5}) and (\ref{eq_3}) lies in their handling of Markov chains on $\mathcal{G}$. In (\ref{eq_3}), there is only a single Markov chain performing serial random walk updates, and the final result is the global model learned at a specific device. However, in (\ref{eq_5}), we employ $M<n$ random walk trajectories to derive the model parameters $\mathbf{w}_i^{t, last}$ for all devices from their last update in the $t$-th communication round. These models are not the desired global model, they are merely intermediate results used for decentralized aggregation in the subsequent communication round.

Therefore, the decentralized weighted averaging aggregation scheme in (\ref{eq_4}) can rewritten as
\begin{equation}
	\label{eq_6}
	\mathbf{w}_i^{t+1,0} \leftarrow \sum_{l \in \mathcal{N}_A(i)} \frac{n_l^t}{m^t} \mathbf{w}_l^{t, last},
\end{equation}
where $\mathbf{w}_i^{t+1,0}$ is the initial model of the device $i$ in the ($t+1$)-th communication round. We denote $\mathbf{w}_l^{t, last}$ as the pre-aggregated model obtained by device $l$ after its last update on any random walk trajectory. Notably, any device in the system can participate in a decentralized aggregation, it is not limited to only those devices at the end of the random walk trajectory that possess the parameters $\mathbf{w}_i^{t, K_m-1}$, $m \in [M]$, where $K_m$ denotes the number of random walk updates in the $m$-th trajectory. That is, device $l$ may participate in updates on multiple random walk trajectories during the $t$-th communication round, but only uses the last updated parameters in the aggregation.

Given the system heterogeneity, the number of devices in each random walk trajectory is different. Consequently, the number of random walk epochs $K$ varies across each Markov chain. Specifically, in the $t$-th communication round, each of the $M$ chains performs $K_m$ random walk SGD to obtain model parameters. These parameters are then combined using decentralized aggregation with weighted averaging. Iterate the above update and aggregation process until convergence. The DFedRW algorithm is described in Algorithm \ref{alg_1}.

\begin{algorithm}[!h]
	\caption{DFedRW}
	\label{alg_1}
	\begin{algorithmic}[1]
		\STATE \textbf{Paremeters}: $\eta^0>0$, $M$, $K_m \in \mathbb{Z}^+$ \\
		\textbf{Initialization}: random network parameter $\mathbf{w}^{1,0}\in\mathcal{W}$
		
		\FOR{$t=1,2,\ldots$}
		\STATE Randomly and uniformly sample $M$ random walk initial devices $\{i_1^{t, 0}, i_2^{t, 0}, \ldots, i_M^{t, 0}\}$ from $\mathcal{V}$
		\FORALL{$m = \{1, 2, \ldots, M\}$ \textbf{in parallel}}
		\FOR{$k=\{0,2,\ldots,K_m-1\}$}
		\STATE Device $i^{t, k}(m)$ computes the gradient estimate $\hat{\mathbf{g}}_{i^{t,k}}$, where $i^{t, k}(m)$ is the updating device on the $m$-th trajectory.
		\STATE (\textit{training}) Device $i^{t, k}(m)$ performs the random walk SGD (\ref{eq_5}) and sends the updated parameters $\mathbf{w}_i^{t, k+1}$ to the next random neighbor $i^{t, k+1}(m)$
		\ENDFOR
		\ENDFOR
		\STATE (\textit{communication}) Device $i \in \mathcal{V}$ periodically performs weighted aggregation (\ref{eq_6}) of local models from the set of random neighbors $\mathcal{N}_A(i)$
		\ENDFOR
	\end{algorithmic}
\end{algorithm}

When the number of random walk epochs $K_m\geq1$ and consists entirely of self-loops, the DFedRW algorithm reduces to DFedAvgM without momentum (DFedAvg). Conversely, when $K_m=1$, DFedRW simplifies to DSGD.

Regarding the handling of stragglers, Algorithm 1 permits the random walk trajectory to perform a variable number of random walk updates based on the capabilities of the devices along the chain, rather than simply discarding chains that cannot complete the specified number of computational rounds. As a result, DFedRW allows integration of partial contributions from stragglers to address system heterogeneity. In fact, whether in centralized or decentralized FL frameworks, the statistical heterogeneity and the system heterogeneity are always interrelated \cite{26Li}. DFedRW removes the limitation that local updates can only be performed on a single device between two communication rounds, and supports flexible local update counts while addressing both system and statistical heterogeneity. This enhances the stability of distributed systems.

In the implementation of DFedRW, a global controller manages the training process by: 1) Selecting the initial devices for the random walk chain. 2) Assigning an aggregation timer $\varphi$ to the initial devices. This timer is passed along with the random walk and decays over time. When $\varphi=0$, device $i$ performs decentralized aggregation using the models of its subset of aggregation neighbors $\mathcal{N}_A(i)$. Additionally, the selection of random walk neighbors is determined by device $i$. Each device maintains an independent neighbor cache, storing recently discovered neighbors. Based on this cache, device $i$ identifies its connected neighbor subset $\mathcal{N}(i)$ and randomly selects the next random walk device $j$. To minimize unnecessary communication overhead, the neighbor cache is updated only when the device starts up, remains inactive for a prolonged period, or fails to communicate with a neighbor.

Although this paper does not investigate the data privacy, it is apparent that the use of random walk updates in DFedRW inherently provides a higher privacy level. This is because the random walk trajectories are not shared with any neighbors. In other words, each device only knows the information about the devices immediately preceding and succeeding it on the trajectory, not the entire set of devices on the trajectory. As a result, it is difficult for any adversary to infer the local training data of other devices without being able to trace back the entire trajectory. This offers higher privacy protection compared to methods involving multiple local iterations such as diffusion \cite{11Hashemi}, consensus \cite{39Wang}, DFedAvgM, and D-ADMM \cite{40Mota}.

\subsection{Efficient Communication via Quantization}
In DFedRW, device $i$ sends $\mathbf{w}_i$ to a random neighbor $j \in \mathcal{N}(i)$ during the update process and to its neighbor subset $\mathcal{N}_c(i)$ during the aggregation process. Therefore, when the number of random walk trajectories $M$ and the size of $|\mathcal{N}_c(i)|$ are large, the communication between devices becomes the primary bottleneck limiting the efficiency of the algorithm.

To address this, we consider using stochastic quantization techniques \cite{14Sun,41Alistarh} to reduce the communication cost of the algorithm. We achieve quantization by stochastically rounding the updated model parameters at the device to a set of discrete values, while preserving the original statistical properties. By ensuring that the expected value remains unchanged and introducing minimal variance, we quantize the normalized parameter $\frac{|w^\nu|}{\|\mathbf{w}\|}$ to a discrete set. Given a quantization interval $s>0$ and a quantization bit width $b \in \mathbb{Z}^{+}$, each $\frac{|w^\nu|}{\|\mathbf{w}\|}$ is stochastically mapped to a discrete set $\{0, s, 2 s, \ldots,(2^{b-1}-1) s\}$ , where one bit in $b$ is reserved as the sign bit. For any $\mathbf{w} \in \mathbb{R}^d$, the stochastic quantization function \cite{41Alistarh} is defined as:
\begin{equation}
	\label{eq_7}
	Q(w^\nu)=\left\{\begin{array}{ll}
		s \ell \|\mathbf{w}\| \cdot \operatorname{sgn}(w^\nu), \quad\text{w.p. } 1-\Phi(\mathbf{w},\nu,\ell); \\
		s (\ell+1)\|\mathbf{w}\| \cdot \operatorname{sgn}(w^\nu), \quad\text{otherwise},
	\end{array}\right.
\end{equation}
where $w^\nu$ is a component of $\mathbf{w}$, and $Q(\mathbf{w})=[Q(w^1), Q(w^2), \ldots, Q(w^d)]$. $\ell \in \mathbb{Z}^+$ satisfies $0 \leq \ell<\frac{1}{2^{b-1}}$ such that $\frac{|w^\nu|}{\|\mathbf{w}\|} \in[s \ell, s(\ell+1)]$. Clearly, for any $w^\nu \in \mathbf{w}$, $\nu \in [d]$, since $\Phi(\mathbf{w},\nu,\ell) = \frac{|w^\nu|}{s\|\mathbf{w}\|}-\ell$ is the relative position of $\frac{|w^\nu|}{\|\mathbf{w}\|}$ in the quantization interval, we have $\mathbb{E}[Q(w^\nu)]=\|\mathbf{w}\| \cdot \operatorname{sgn}(w^\nu) \cdot (s \ell(1-\Phi)+s (\ell+1)\Phi)=w^\nu$, which implies that our stochastic quantization scheme is unbiased.

For any quantized $w^\nu$, we output a tuple $(\Lambda, s,\|\mathbf{w}\|)$, where $\Lambda$ encodes the quantization index $\ell$ along with the sign bit. For vector $\mathbf{w}$, the additional cost of transmitting the 32-bit $\|\mathbf{w}\|$ and $s$ ensures relatively stable quantization error across a wide range of gradient scales.

\begin{lemma}[Bounded quantization variance]
	\label{lemma_2}
	For any $\mathbf{w} \in \mathbb{R}^d$, the variance between the quantized parameter $Q(\mathbf{w}) \in \mathbb{R}^d$ and its true value is bounded, that is, $\mathbb{E}\left[\|Q(\mathbf{w})-\mathbf{w}\|^2\right] \leq \frac{\sigma^2 d s^2}{4}$, where $\sigma$ is the bound of the parameter norm $\|\mathbf{w}\|$.
\end{lemma}
\begin{proof}
	Appendix B provides a detailed proof of Lemma \ref{lemma_2}.
\end{proof}

As indicated by Lemma \ref{lemma_2}, the stochastic quantization scheme achieves a balance between quantization precision and variance. Since we quantize the normalized parameters, the quantization interval $s$ is small, leading to a correspondingly smaller upper bound on the quantization variance.

In \cite{41Alistarh}, communication costs of SGD are reduced by transmitting quantized gradients. However, directly quantizing parameters or gradients is feasible for smooth loss functions, but poses challenges in deep neural networks. The paper \cite{14Sun} suggests quantizing the differences in model parameters to preserve performance, as quantizing these differences does not impact the continuity and smoothness of the loss function itself. For DFedAvgM, only the parameter differences transmitted in a decentralized manner during aggregation need to be quantized. Given the different update mechanisms, the communication process in our DFedRW involves not just decentralized aggregation but also the critical random walk updates. As $M$ or $K$ of parallel random walk trajectories increases, the communication cost scales proportionally, presenting greater challenges for our quantization scheme.

To mitigate the accumulation of quantization errors, QDFedRW adopts incremental quantization, similar to \cite{14Sun}. This method effectively reduces information loss avoids the severe error accumulation, particularly beneficial for DNNs with non-smooth loss functions or LSTMs that rely on gating mechanisms and long-term memory.
	
In addition, we quantize all communication processes in DFedRW. During the local update process of the random walk, (\ref{eq_5}) is reformulated as
\begin{equation}
	\label{eq_8}
	\begin{aligned}
		\mathbf{w}_{i}^{t, k+1} \leftarrow \boldsymbol{\Pi}_{\mathcal{W}}\left(\widetilde{\mathbf{w}}_{i}^{t,k}-\eta^{\bar{k}} \widetilde{\mathbf{g}}_{i^{t,k}}\right),
	\end{aligned}
\end{equation}
where $\widetilde{\mathbf{w}}_{i}^{t,k} = \mathbf{w}_{i}^{t, k}+\mathbf{Q}^{t,k}(j)$ and $\widetilde{\mathbf{g}}_{i^{t,k}} = \hat{\nabla} F_{i^{t,k}}\big(\mathbf{w}_{i}^{t, k}+\mathbf{Q}^{t,k}(j)\big)$.

The above equation can be described as follows. At the $k$-th step of the random walk, device $i^{t,k}$ receives the quantized update differences $\mathbf{Q}^{t,k}(j) = Q(\mathbf{w}_{j}^{t, k}-\mathbf{w}_{j}^{t, k-1})$ sent by device $j=i^{t,k-1}$. This received difference is added to the local model parameters $\mathbf{w}_{i}^{t, k}$ of device $i^{t,k}$, serving as a substitute for the pre-update local model. Device $i^{t,k}$ then uses this model to perform the next epoch of the random walk update based on its local data. At the aggregation time, device $l$ sends the quantized differences $\mathbf{Q}^{t}(i) = Q(\mathbf{w}_i^{t, last}-\mathbf{w}_i^{t, 0})$ between its final updated parameters and the initial parameters in this round to a subset $\mathcal{N}_c(l)$ of its neighbors. Device $i$ also periodically aggregates the quantized values received from its neighbors $\mathcal{N}_A(i)$ . The quantized version of (\ref{eq_6}) is reformulated as 
\begin{equation}
	\label{eq_9}
	\begin{aligned}
		\mathbf{w}_i^{t+1,0} \leftarrow \mathbf{w}_i^{t, 0}+\sum_{l \in \mathcal{N}_A(i)} \frac{n_l^t}{m^t} \mathbf{Q}^{t}(l).
	\end{aligned}
\end{equation}

We will analyze the convergence of the quantized DFedRW from the perspective of transmitting differences in the next section. The quantized DFedRW algorithm is presented in Algorithm \ref{alg_2}.

\renewcommand{\thefootnote}{\P}
\begin{algorithm}[!h]
	\caption{Quantized DFedRW (QDFedRW)}
	\label{alg_2}
	\begin{algorithmic}[1]
		\STATE \textbf{Paremeters}: $\eta^0>0$, $M$, $K_m \in \mathbb{Z}^+$, $s$ \\
		\textbf{Initialization}: random network parameter $\mathbf{w}^{1,0}\in\mathcal{W}$
		
		\FOR{$t=1,2,\ldots$}
		\STATE Randomly and uniformly sample $M$ random walk initial devices $\{i_1^{t, 0}, i_2^{t, 0}, \ldots, i_M^{t, 0}\}$ from $\mathcal{V}$
		\FORALL{$m = \{1, 2, \ldots, M\}$ \textbf{in parallel}}
		\FOR{$k=\{0,2,\ldots,K_m-1\}$}
		\STATE Device $i^{t, k}(m)$ computes an unbiased estimate $\widetilde{\mathbf{g}}_{i^k}$ of the gradient using local $\xi_i$, $\mathbf{w}_i^{t, k}$, and $\mathbf{Q}^{t,k}(j)$ received from neighbor $i^{t, k-1}(m)$ 
		\STATE (\textit{training}) Device $i^{t, k}(m)$ performs the random walk SGD (\ref{eq_8}), sends the updated parameters difference $\mathbf{Q}^{t,k+1}(i)$ and the quantization precision $s$ to the next random neighbor $i^{t, k+1}(m)$
		\ENDFOR
		\ENDFOR
		\STATE (\textit{communication}) Device $i \in \mathcal{V}$ periodically performs weighted aggregation (\ref{eq_9}) of local models from the set of random neighbors $\mathcal{N}_A(i)$
		\ENDFOR
	\end{algorithmic}
\end{algorithm}

In each update of the random walk, device $i$ in DFedRW sends $32d$ bits of model parameters to a random neighbor. In contrast, quantized DFedRW only requires sending ($64+bd$) bits of quantized parameters, which includes 32 bits for the quantization precision parameter $s$, 32 bits for $\|\mathbf{w}\|$, and $bd$ bits for the quantization index and sign bit. In each communication round, device $i$ in DFedRW and quantized DFedRW sends $32 d\left|\mathcal{N}_c(i)\right|$ bits and $(64+b d)\left|\mathcal{N}_c(i)\right|$ bits of parameters, respectively, to a subset of random neighbors. Consequently, quantized DFedRW can significantly reduce communication overhead when the model dimension $d$ is large, and the number of quantization bits $b$ is less than the original $32$ bits.

\section{Convergence Analysis}
\subsection{Assumptions}
In this section, we present the assumptions required for proving the convergence of DFedRW. The assumptions are commonly used in theoretical analyses within the field of machine learning \cite{9Sun}, \cite{25Ayache}, \cite{32Sun}.

\newtheorem{assumption}{Assumption}
\begin{assumption}
	\label{assumption_1}
	The local loss function $F_i$ at device $i$ is convex and differentiable with respect to $\mathcal{W} \subseteq \mathbb{R}^d$. For any $\mathbf{w}_1, \mathbf{w}_2 \in \mathcal{W}$, we have $F_i(\mathbf{w}_1) \geq F_i(\mathbf{w}_2)+\nabla F_i(\mathbf{w}_2)^{\top}(\mathbf{w}_1-\mathbf{w}_2)$.
\end{assumption}

\begin{assumption}
	\label{assumption_3}
	Decreasing step size $\eta^k$ satisfies $\sum_{k=1}^{\infty} \eta^k=+\infty$, and $\sum_{k=1}^{\infty} \ln k \cdot(\eta^k)^2<+\infty$.
\end{assumption}

\begin{assumption}
	\label{assumption_5}
	The Markov chain $(i^k)_{k \geq 0}$ is time-homoge-\\neous, irreducible, and aperiodic. It has a transition matrix $\mathbf{P}$ and a stationary distribution $\pi^*$.
\end{assumption}

\subsection{Convergence Analysis of DFedRW and its Quantified Version}
Our objective is to delineate the factors that influence the asymptotic convergence bound of (quantized) DFedRW. Our primary contribution is the integration of random walk with decentralized optimization, enriched by the inclusion of more challenging settings of statistical and system heterogeneity. This framework provides a comprehensive convergence analysis for DFedRW under convex conditions. Furthermore, we conduct an analysis and comparison of the quantized DFedRW. 

For convex and $L$-smooth settings, SGD has a well-established convergence bound of order $\mathcal{O}(\frac{1}{\sqrt{k}})$ \cite{13Lian}. It has been demonstrated in \cite{32Sun} and \cite{33Ayache}, that random walk SGD with decreasing step sizes achieves a convergence bound of order $\mathcal{O}(\frac{1}{k^{1-q}})$ for convex loss functions. In the subsequent convergence analysis, we show that the proposed DFedRW algorithm also achieves a similar convergence bound. Furthermore, the constants in the convergence bound are influenced by the heterogeneities setting and the quantization technique employed.

As shown in (\ref{eq_5}) and (\ref{eq_8}), the random walk trajectory remains unbroken by the aggregation process. In other words, the random walk before and after aggregation forms a continuous trajectory rather than two independent trajectories. Our convergence analysis is conducted from this perspective.

\newtheorem{theorem}{Theorem}
\begin{theorem}[Convergence of DFedRW]
	\label{theorem_1}
	Under Assumptions \ref{assumption_1}, \ref{assumption_3} and \ref{assumption_5}, let $\eta^{k}=\mathcal{O}(\frac{1}{k^q})$, $\frac{1}{2}<q<1$, then the following convergence rate is upper bounded by
	\begin{equation}
	\begin{aligned}
		& \mathbb{E}\left[f\left(\overline{\mathbf{w}}^{k}\right)-f\left(\mathbf{w}^*\right)\right] \\
		& \quad \leq \frac{\frac{1}{2}\left\|\mathbf{w}^{0}-\mathbf{w}^*\right\|^2+\psi\left(n, \lambda_P\right)+\kappa(\delta, \hat{\gamma})+\lambda(n)}{k^{1-q}},
	\end{aligned}
    \end{equation}
	where
	\vspace{-2ex}
	\[\begin{aligned}
	    \overline{\mathbf{w}}^{k}=\frac{\sum_{i=1}^k \eta^{i} \mathbf{w}^{i}}{\sum_{i=1}^k \eta^{i}}, \quad\lambda(n)=\sum_{k} \frac{n \eta^{k}}{2 k},
    \end{aligned}\]\\
    \vspace{-2ex}
	\[\begin{aligned}
		\psi\left(n, \lambda_P\right)=\frac{(1+n) \cdot 2 D^2 \sum_{k=\bar{K}}^{\infty} \ln k \cdot(\eta^{k})^2}{\ln \left(1 / \lambda_P\right)},
	\end{aligned}\]\\
	\vspace{-2ex}
	and \\
	\vspace{-2ex}
	\[\begin{aligned}
	    \kappa(\delta, \hat{\gamma})=\frac{\left(\delta^2+\hat{\gamma}^2\right) D^2}{4} \sum_{k}(\eta^{k})^2.
    \end{aligned}\]
\end{theorem}
\begin{proof}
	The proof of Theorem \ref{theorem_1} is provided in Appendix A.
\end{proof}

From Theorem \ref{theorem_1}, it is evident that the convergence bound $\mathcal{O}(\frac{1}{k^{1-q}})$ of DFedRW under convexity assumptions is essentially the same as $\mathcal{O}(\frac{1}{\sqrt{k}})$ of SGD. The convergence efficiency of the algorithm is significantly influenced by factors such as the topology of the communication network and heterogeneities. Specifically, superior expansion properties of the graph lead to an enhanced convergence rate. The Non-IID of local data among participating devices exacerbates the inconsistency between local objective functions, thereby slowing down the convergence rate of DFedRW. Additionally, the computational capabilities of devices along the random walk trajectory directly impact the performance of algorithm. Specifically, devices with stronger computational power exhibit a smaller $\hat{\gamma}$, facilitating faster convergence of the algorithm. The analytical conclusions drawn above are highly consistent with the results obtained from numerical simulations in the subsequent section.

Subsequently, we provide a quantified convergence guarantee.

\begin{theorem}[Convergence of quantized DFedRW]
	\label{theorem_2}
	Assume Assumptions \ref{assumption_1}, \ref{assumption_3} and \ref{assumption_5} hold, let $\eta^{k}=\mathcal{O}(\frac{1}{k^q})$, $\frac{1}{2}<q<1$, $\tau^{k}=\min \{k, \max \{ \lceil \frac{\ln (2 \zeta k)}{\ln \left(1 / \lambda_P\right)} \rceil ,K_\mathbf{P}\}\}$, then the following convergence rate is upper bounded by
		\begin{equation}
		\begin{aligned}			
			& \mathbb{E}\left[f\left(\overline{\mathbf{w}}^{k}\right)-f\left(\mathbf{w}^*\right)\right] \\
			& \leq \frac{\frac{1}{2}\left\|\mathbf{w}^{0}-\mathbf{w}^*\right\|^2+\psi\left(n, \lambda_P\right)+\kappa(\delta, \hat{\gamma})+\lambda(n)+\psi(d, s)}{k^{1-q}},
		\end{aligned}
		\end{equation}
		where
		\vspace{-2ex}
		\[
		\begin{aligned}
			\psi(d, s) = (1+n) \frac{D \sigma \sqrt{d} s}{2} \sum_{k} \eta^{k} \tau^{k}.
		\end{aligned}
		\]
\end{theorem}
\begin{proof}
	The proof of Theorem \ref{theorem_2} is provided in Appendix B.
\end{proof}

Theorem \ref{theorem_2} establishes the relationship between the convergence rate of quantized DFedRW and the parameter dimensions $d$ and quantization interval $s$. Compared to the non-quantized version, the quantized version introduces an additional upper bound term of order $\mathcal{O}(\sqrt{d}s)$, resulting in a more relaxed bound. This implies that as $d$ decreases and quantization precision improves, the convergence bound of quantized DFedRW approaches that of the original non-quantized version. This observation suggests that while quantization reduces the amount of parameters transmitted during the communication process, it may require more communication rounds. However, revisiting the scenario where $\sqrt{d}s$ is relatively small and heterogeneities are strong, the convergence bound in Theorem \ref{theorem_2} is dominated by $\psi(n, \lambda_P)$ and $\kappa(\delta, \hat{\gamma})$. Therefore, the greater the heterogeneities, the more pronounced the rate advantage of quantized DFedRW. This can be verified in numerical simulations.

We proceed to investigate the sufficient conditions for saving communication cost in quantized DFedRW.

\newtheorem{proposition}{Proposition}
\begin{proposition}
	\label{proposition_1}
	Under Assumptions \ref{assumption_1}, \ref{assumption_3} and \ref{assumption_5}, let $\eta^{k}=\mathcal{O}(\frac{1}{k^q})$, $\frac{1}{2}<q<1$. The quantized DFedRW requires more iterations $k_q$ for convergence compared to the non-quantized version $k_{nq}$, i.e., $T_q=\varrho T_{nq}$, $\varrho>1$. Using $b$ bits to quantize DFedRW, if the expected error satisfies
	\begin{equation}
	\epsilon>\frac{(1+n) D \sigma \sqrt{d} s} {\sqrt{\ln \left(1 / \lambda_P\right)}},
	\end{equation}
	and $b<\frac{32}{\varrho}-\frac{64}{d}$, then the $b$ bits quantized DFedRW requires less communication overhead than the 32 bits DFedRW.
\end{proposition}
\begin{proof}
	Detailed proof of Proposition \ref{proposition_1} is provided in Appendix C.
\end{proof}

Proposition \ref{proposition_1} indicates that the dimension $d$, the quantization interval $s$ and the eigenvalue $\lambda_P$ of the random walk transition matrix $\mathbf{P}$ have a significant impact on the expected error $\epsilon$. High-dimensional problems with low quantization accuracy and slow-converging quantized DFedRW exhibit greater $\epsilon$.

\section{Numerical Results}
We utilize the proposed DFedRW and its communication-quantized variant to train DNNs for image classification. Our goal is to validate that (quantized) DFedRW can effectively train distributed models under conditions of different statistical and system heterogeneity. Furthermore, we evaluate it on large-scale language modelling to demonstrate the effectiveness of DFedRW on a real-world problem with natural partitioning of client-side heterogeneous data.

\subsection{Setup}
Our experiments are conducted on a workstation equipped with an Intel i9-11900K CPU, 32GB DDR4 3200MHz RAM, and an NVIDIA RTX 3070 8GB GPU. The implementation is developed using Python 3.9, PyTorch 2.0, and CUDA 11.8.

For image classification, we employ 20 devices to train two models for MNIST and Fashion-MNIST classification: a two-layer fully connected network (2FNN) and a three-layer fully connected network (3FNN). Both models take 784-dimensional inputs and produce 10-dimensional outputs, with log-softmax applied at the output layer. The 2FNN consists of a single hidden layer with 100 units, whereas the 3FNN has two hidden layers with 200 units each. All hidden layers use ReLU activation. The MNIST and Fashion-MNIST datasets each consist of 60,000 training images and 10,000 test images, all in 28×28 grayscale format representing handwritten digits or clothing categories. The training set is partitioned heterogeneously (as detailed below), while the test set retains its original IID distribution to serve as a fair evaluation benchmark. We investigate the effectiveness of (quantized) DFedRW in addressing statistical and systemic heterogeneities. To reduce the effect of randomness and ensure the reliability of the findings, we perform each simulation 7 times and report the average.

In statistical heterogeneity settings, we implement two distinct data partitioning approaches. 1) Deterministic partitioning: For devices with $u$\% similarity, their local datasets contain $u$\% data from a randomly shuffled IID shared data pool, while the remaining $(100-u)$\% from Non-IID data pool. We categorize the data in the Non-IID data pool by label, and each category is divided into four equal-sized shards, for a total of 40 shards. Each device receives any two of these shards. Specifically, in IID setting, $u = 100$. 2) Probabilistic partitioning: We construct a Non-IID setting with varying label distributions across devices using Dirichlet data partitioning. Specifically, we first generate a device distribution matrix $\mathbf{P} \in \mathbb{R}^{C \times n}$, where each row vector $\mathbf{p}_c \sim \mathsf{Dir}(\alpha_d)$ represents the allocation proportions of class $c \in C$ across $n$ devices.  As $\alpha_d \to 0$, the data distribution becomes highly heterogeneous.

In system heterogeneity settings, we apply DFedRW to aggregate partial solutions from multiple random walk trajectory updates. Specifically, the local update frequency of each random walk trajectory is constrained by a global clock, and the local workload depends on the product of implicit values $\hat{\gamma}=\prod_{q=k-K}^{k-1} \gamma_{i^q}^{q}$ across all devices on the trajectory. In cases where devices exhibit no heterogeneity, the number of epochs between two communication rounds for random walk SGD updates remains $K$. System heterogeneity results in certain random walk trajectories performing fewer updates than $K$ epochs. For an FL system with $h$\% heterogeneity, we allocate $K^{\prime} < K$ participating devices for the $h$\% of random walk chains.

\renewcommand{\thefootnote}{*}
For language modelling, we train LSTM for word prediction. We use publicly available Reddit dataset\footnote{https://bigquery.cloud.google.com/dataset/ fh-bigquery:reddit\_comments} from November 2017 for training, with each Reddit user corresponding to an independent federated learning device. To ensure training quality, we only retain users with post counts between $150-500$, totalling 83,293 devices. LSTM first looks up the input words in a dictionary containing 50K words and maps their indices to 128-dimensional word embeddings through an embedding layer. The sequence of word embeddings is processed by a two-layer LSTM, generating a 256-dimensional sequence of hidden states. The hidden state of the last time step is taken and mapped to the dimension of the vocabulary size through the fully connected layer to generate the probability distribution of the next word. The dictionary includes the most frequently occurring words in the training dataset. Additionally, we randomly select 5K comments from the October 2017 Reddit dataset for testing.

\subsection{Comparison between DFedRW and Baselines}
\captionsetup[subfloat]{font=scriptsize}
\captionsetup[subfloat]{labelformat=empty}

\begin{figure*}[!t]
	\centering	
	\subfloat[MNIST ($u$=100, $h$=0) \label{acc_iid_sys0}]{
		\includegraphics[width=0.24\linewidth]{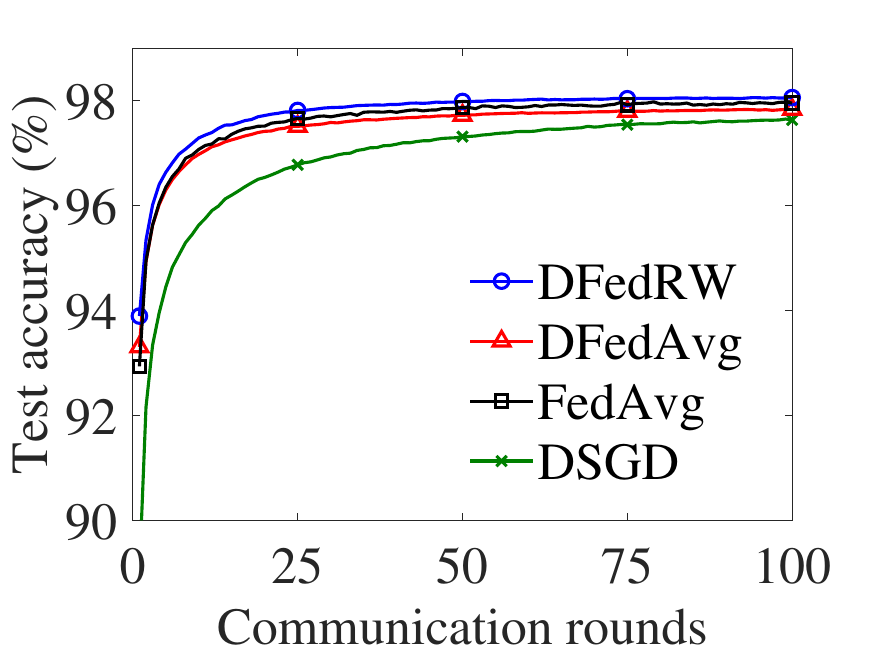}
	}
	\subfloat[MNIST ($u$=50, $h$=0) \label{acc_noniid50_sys0}]{
		\includegraphics[width=0.24\linewidth]{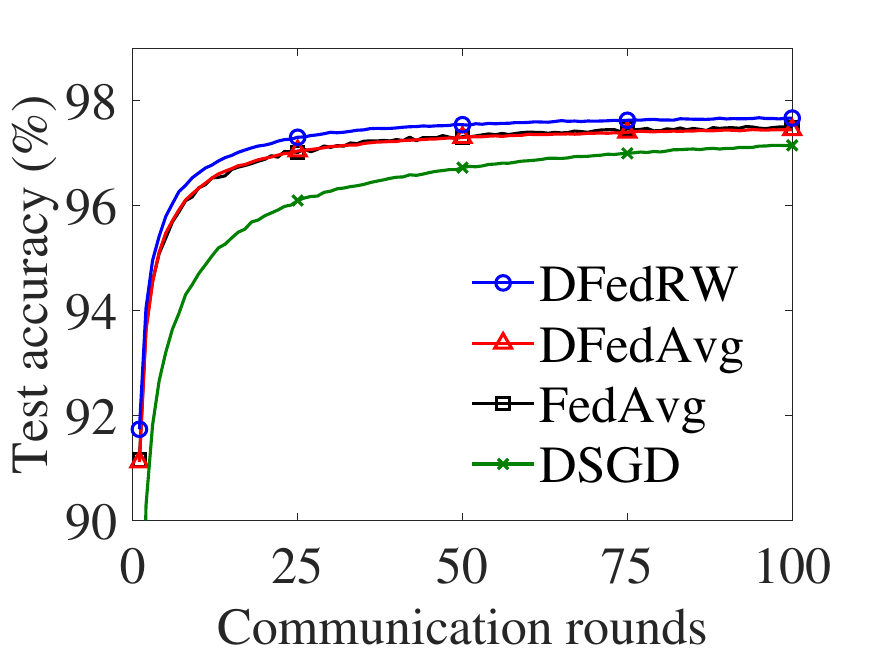}
	}
	\subfloat[MNIST ($u$=0, $h$=0) \label{acc_noniid0_sys0}]{
		\includegraphics[width=0.24\linewidth]{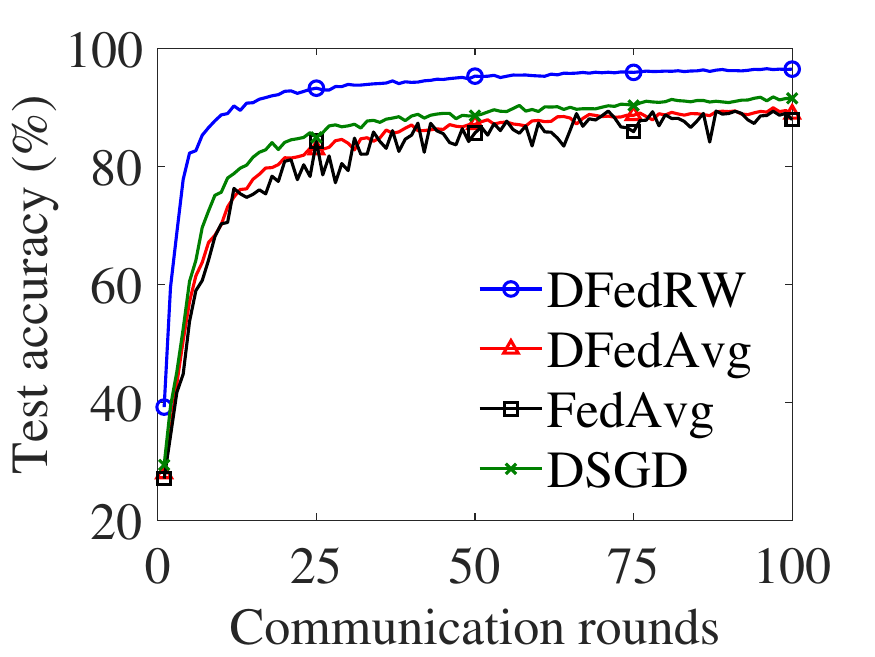}
	}
	\subfloat[MNIST ($u$=0 \& nonbalance, $h$=0) \label{acc_noniid0nonbalance_sys0}]{
		\includegraphics[width=0.24\linewidth]{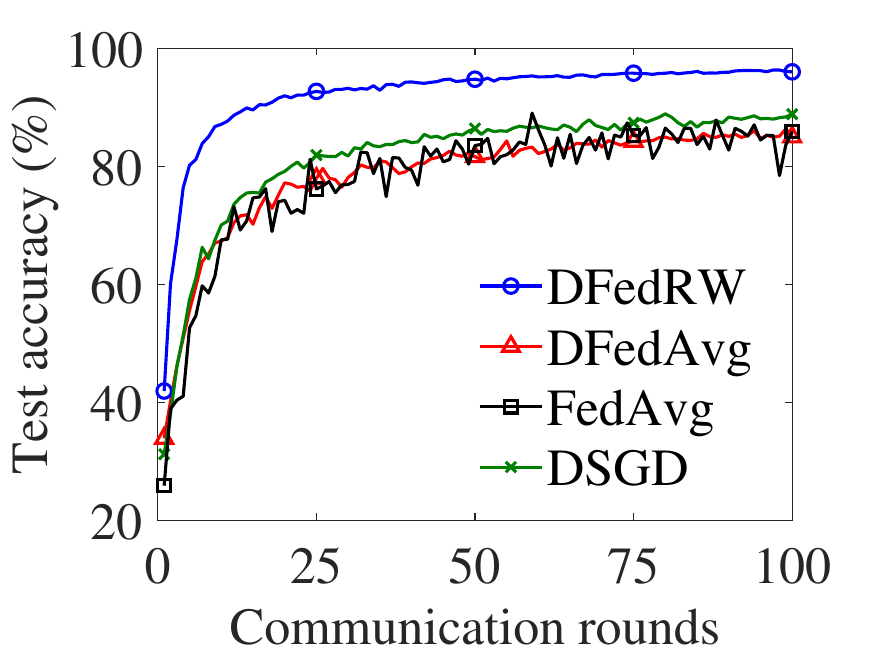}
	}
	
	\vspace{-0.4cm} 
	
	\subfloat[Fashion-MNIST ($u$=100, $h$=0) \label{acc_iid_sys0_FMNIST}]{
		\includegraphics[width=0.24\linewidth]{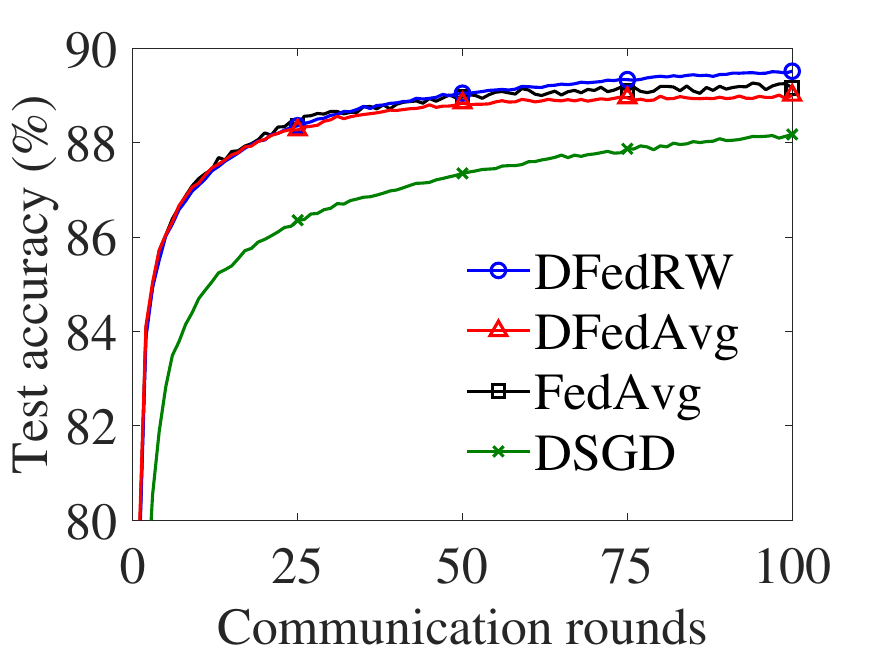}
	}
	\subfloat[Fashion-MNIST ($u$=50, $h$=0) \label{acc_noniid50_sys0_FMNIST}]{
		\includegraphics[width=0.24\linewidth]{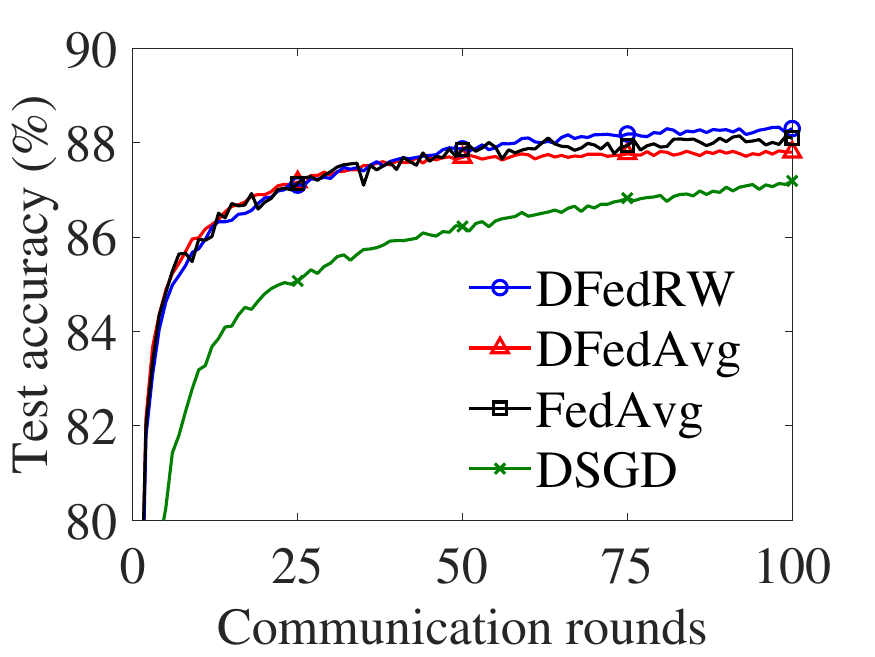}
	}
	\subfloat[Fashion-MNIST ($u$=0, $h$=0) \label{acc_noniid0_sys0_FMNIST}]{
		\includegraphics[width=0.24\linewidth]{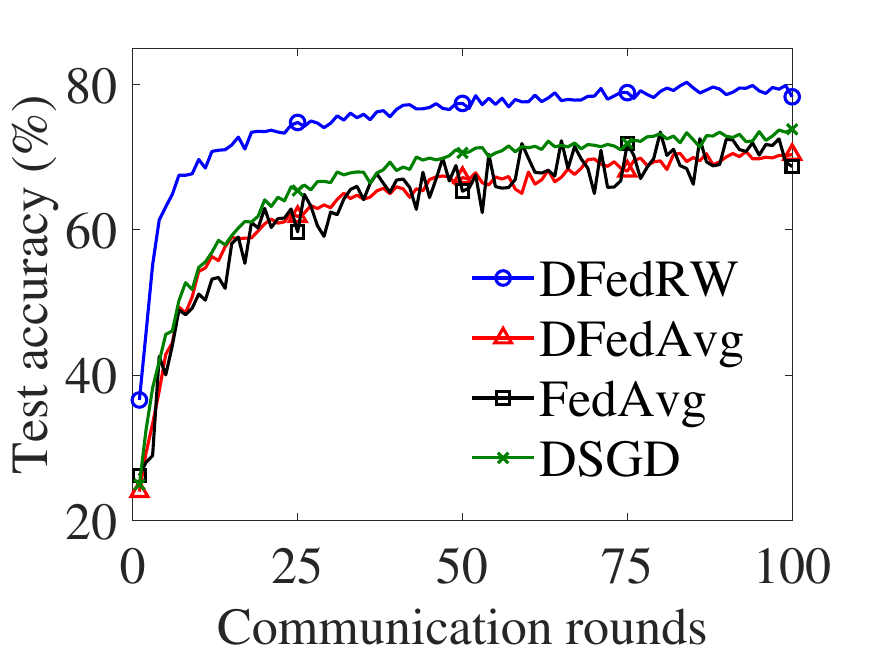}
	}
	\subfloat[Fashion-MNIST ($u$=0 \& nonbalance, $h$=0) \label{acc_noniid0nonbalance_sys0_FMNIST}]{
		\includegraphics[width=0.24\linewidth]{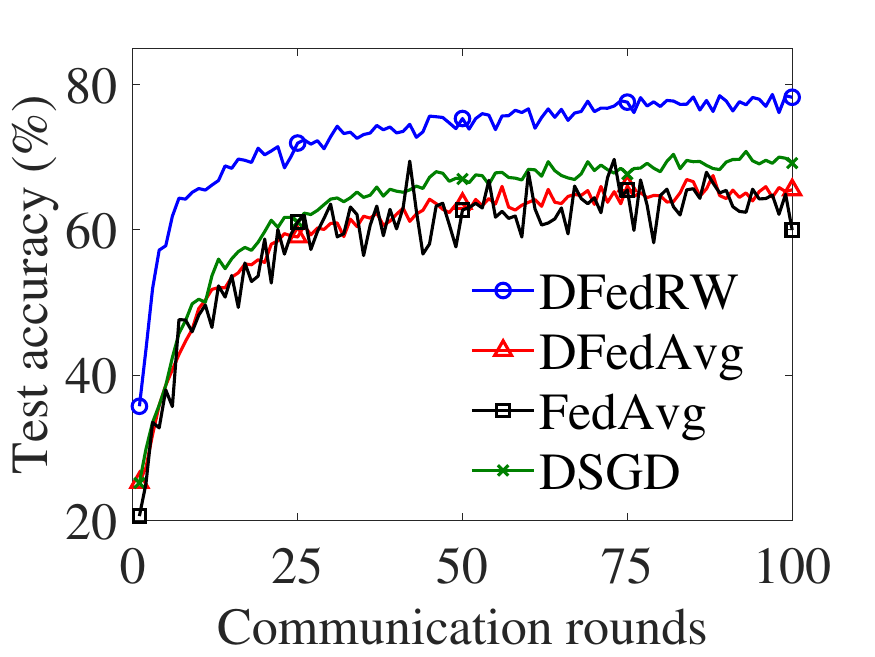}
	}
	
	\caption{Efficiency comparison of DSGD, FedAvg, DFedAvg and DFedRW in training 3FNN for different levels of Non-IID image classification. Fixed absence of system heterogeneity.}
	\label{fig_data}
\end{figure*}

\begin{figure}[!t]
	\centering	
	\subfloat[MNIST ($u$=100, $h$=0) \label{testloss_iid_sys0}]{
		\includegraphics[width=0.49\linewidth]{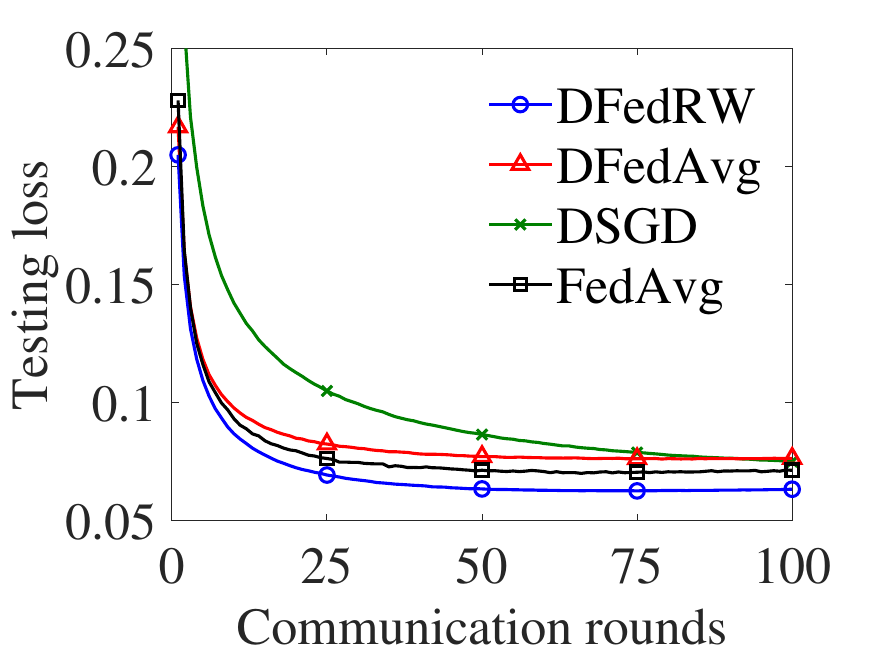}
	}
	\subfloat[MNIST ($u$=50, $h$=0) \label{testloss_noniid50_sys0}]{
		\includegraphics[width=0.49\linewidth]{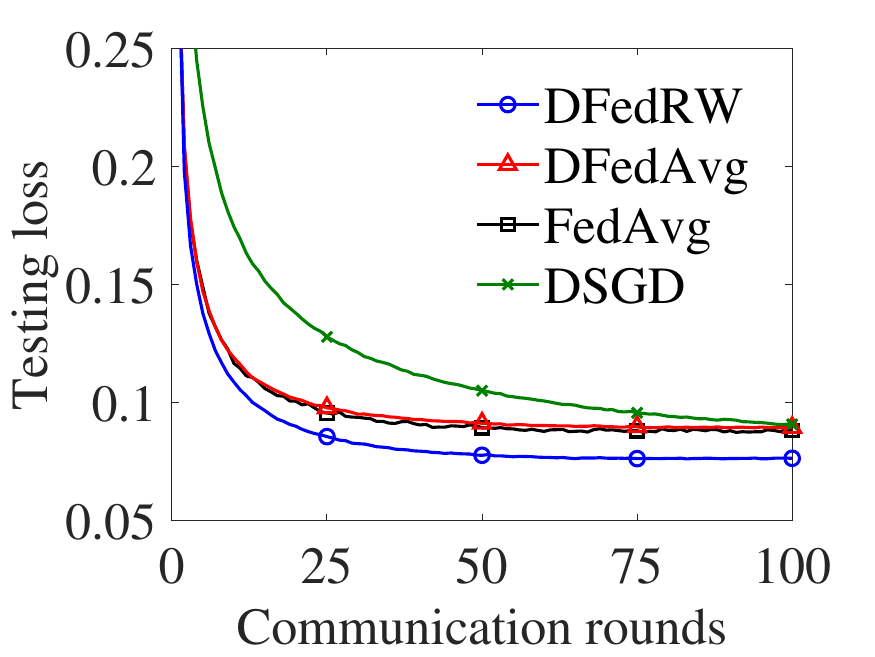}
	}
	
	\vspace{-0.4cm} 
	
	\subfloat[MNIST ($u$=0, $h$=0) \label{testloss_noniid_sys0_FMNIST}]{
		\includegraphics[width=0.49\linewidth]{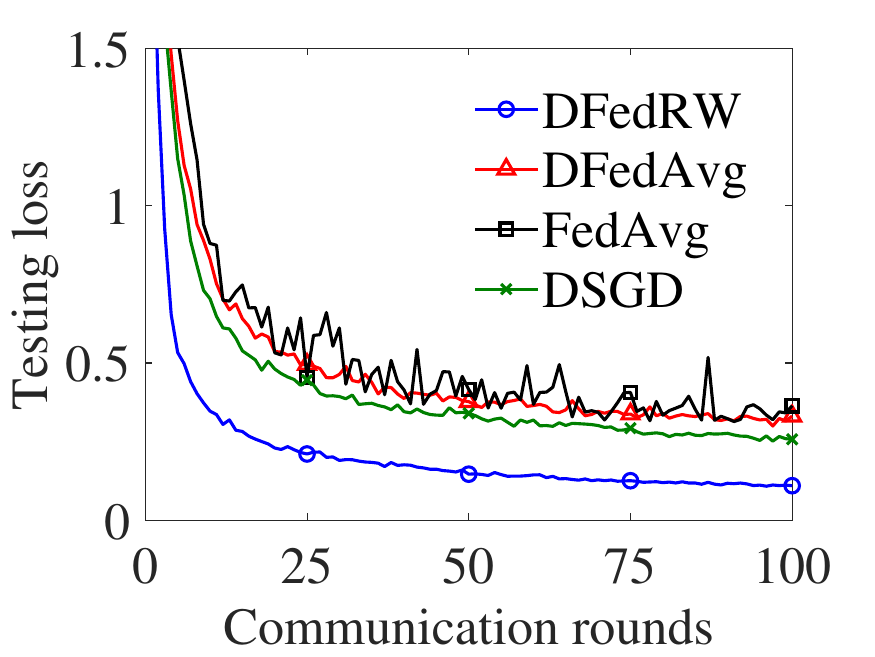}
	}
	\subfloat[MNIST ($u$=0 \& nonbalance, $h$=0) \label{test_MNIST_homo_noniid_nonbalance}]{
		\includegraphics[width=0.49\linewidth]{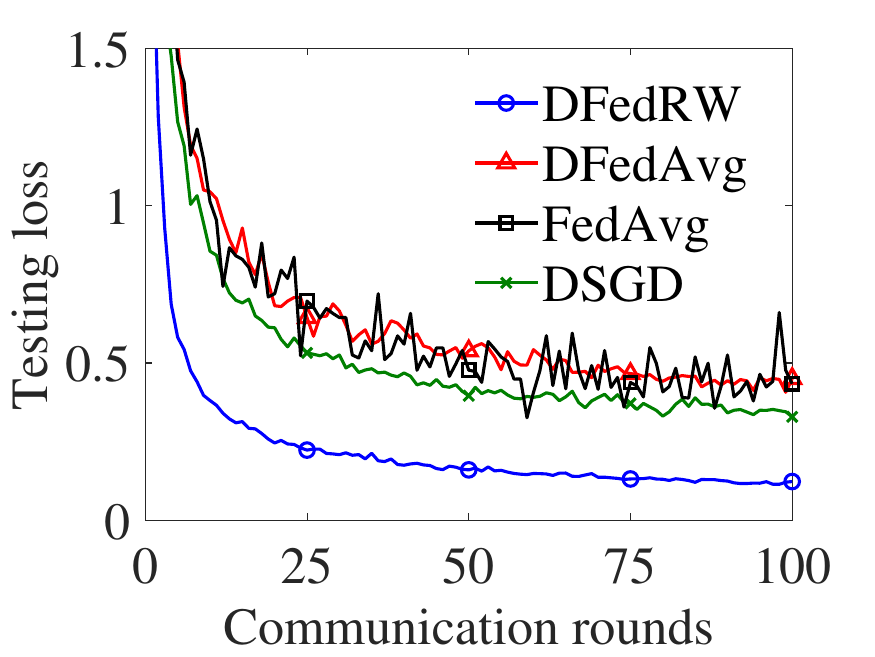}
	}

	\caption{Testing loss comparison of DSGD, FedAvg, DFedAvg and DFedRW in training 3FNN for different data heterogeneity in MNIST classification.}
	\label{fig_data_testloss}
\end{figure}

We select DFedAvg, FedAvg, and DSGD as baselines to evaluate the effectiveness of DFedRW. Both DFedAvg and DFedRW are decentralized, differing in update strategies. DFedAvg aggregates after multiple local updates, while DFedRW employs random walk updates, assessing its impact in heterogeneous settings. As the standard centralized approach, FedAvg serves as a benchmark for analyzing the effects of decentralized aggregation and random walk updates on convergence. DSGD leverages single-step updates to alleviate local drift in Non-IID environments, positioning it as a more reliable baseline under high heterogeneity. During training, we set the batch size of DFedRW and the baselines to 50, with a decreasing learning rate $\eta^{\bar{k}}=\frac{1}{R \bar{k}^{0.499}}$. This is because random sampling algorithms usually adopt a decreasing learning rate to ensure convergence \cite{8Mao}. Specifically, $R=10$ when $u=0, h=90$ (for MNIST) or $u=0, h \geq 50$ (for Fashion-MNIST); otherwise, $R=5$. Each communication round aggregates 25\% of the devices. In the IID setting, the number of local updates for the (D)FedAvg is 5, and correspondingly, the random walk updates for DFedRW are also set to 5. In the non-IID setting, all algorithms adapt their update counts based on device computational capacity, ensuring a fair comparison.

Firstly, we validate the capability of DFedRW to handle heterogeneous data. Fig. \ref{fig_data} compares the test accuracy of different algorithms training the 3FNN under progressively increasing levels of Non-IID data, measured in terms of communication rounds, and Fig. \ref{fig_data_testloss} shows the corresponding testing loss. We set the system heterogeneity to 0\% ($h=0$). In the IID setting, both FedAvg and DFedAvg achieve similar test accuracy and convergence rates. DFedRW tests slightly more accurate than (D)FedAvg. All three outperform DSGD in both aspects. For the Non-IID setting, when 50\% of the data comes from Non-IID shards, the accuracy and convergence rates of DFedRW and the baselines are affected to almost the same extent, and the advantages of DFedRW in addressing non-IID data are not yet apparent. As the level of Non-IID increases further, the accuracy of (D)FedAvg oscillates significantly with rounds when $u = 0$, and the gap with DFedRW increases significantly, to 7.5\% and 9.5\% on the two datasets, respectively. This is due to the fact that multiple updates on heterogeneous local data cause the local model to deviate from the global model. While DSGD suffers relatively less performance loss due to the fact that only one local update is performed in each round. DFedRW achieves the highest accuracy and the fastest convergence rate by traversing more diverse data categories by random walk in the local update. This indicates that the higher the degree of statistical heterogeneity, the more pronounced the advantage of DFedRW becomes.

In particular, we test scenarios under deterministic partitioning where local label distributions are imbalanced, but each device has the same total number of samples, i.e. $u=0$ \& nonbalance in Fig. \ref{fig_data}. In the Non-IID-nonbalance partitioning, the data for devices come from multiple labels, and the amount of data for each label is imbalanced. The process is as follows. We randomly select a label and allocate data to each device, up to a maximum of 1500 samples per label. Then, we switch to the next label and continue this allocation until data budget of the device is reached. Compared to completely heterogeneous data, introducing imbalance significantly reduces the test accuracy and convergence bound of the baselines and exacerbates the magnitude of the oscillation. In contrast, DFedRW maintains accuracy similar to that with balanced data, at the cost of only a slightly reduced convergence bound.

In the Dirichlet data partitioning scheme, we set $\alpha_d = 0.1$ while keeping other settings consistent with the deterministic partitioning scheme to validate the effectiveness of DFedRW in a heterogeneous environment. Each device has a private dataset with a different label distribution and varying data volume. As shown in Fig. \ref{fig_Dirichlet}, DFedRW demonstrates superior accuracy on both datasets and exhibits the smallest oscillation amplitude.

\begin{figure}[!t]
	\centering	
	\subfloat[MNIST (Dirichlet, $\alpha_d=0.1$) \label{acc_MNIST_Dirichlet}]{
		\includegraphics[width=0.49\linewidth]{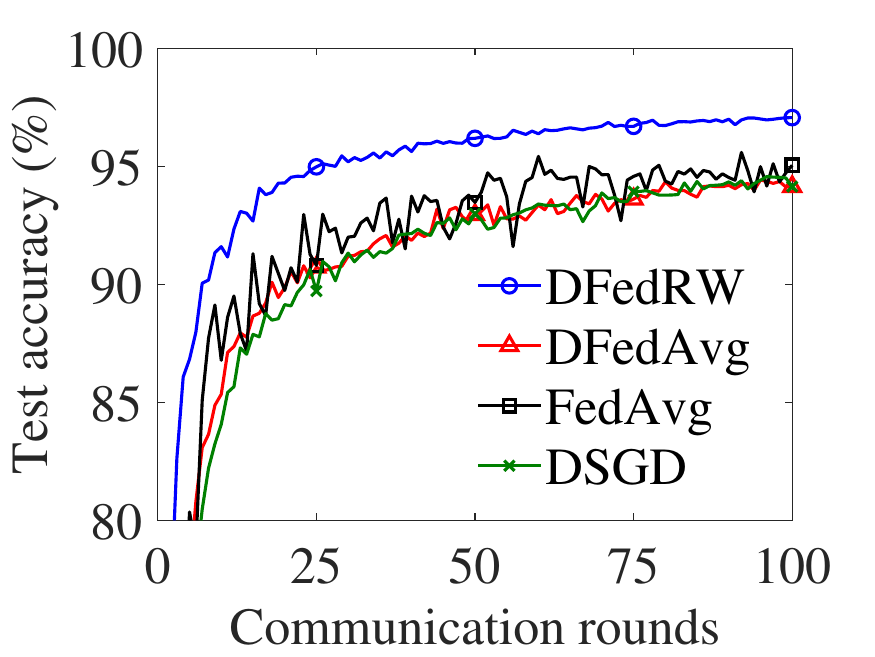}
	}
	\subfloat[MNIST (Dirichlet, $\alpha_d=0.1$) \label{test_MNIST_Dirichlet}]{
		\includegraphics[width=0.49\linewidth]{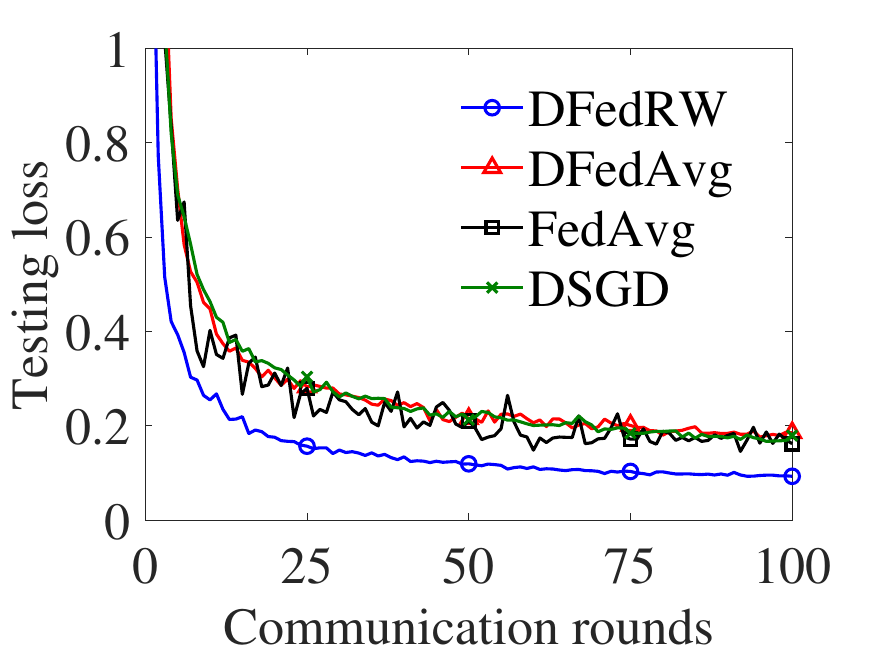}
	}
	
	\vspace{-0.4cm} 
	
	\subfloat[Fashion-MNIST (Dirichlet, $\alpha_d=0.1$) \label{acc_Fashion_MNIST_Dirichlet}]{
		\includegraphics[width=0.49\linewidth]{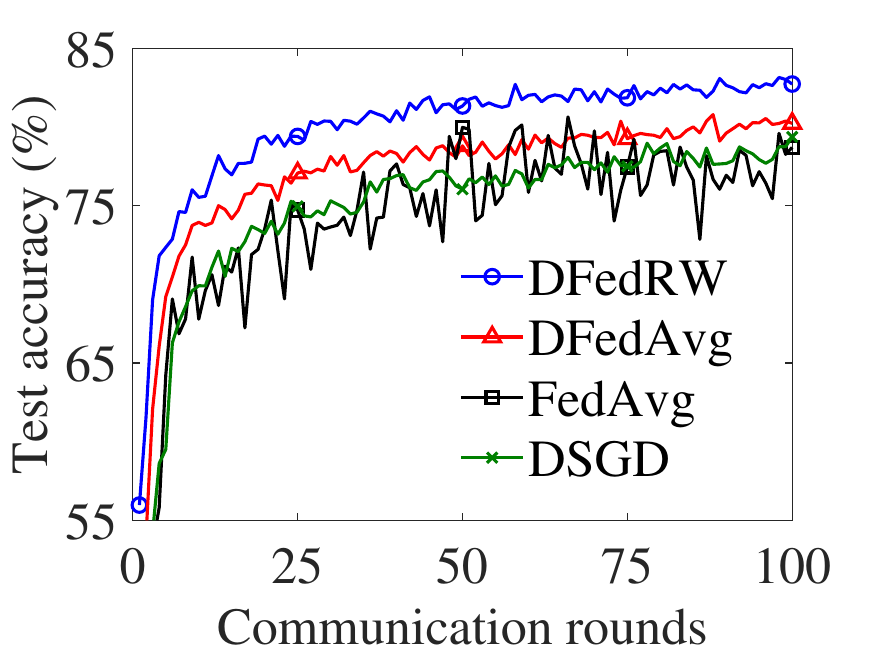}
	}
	\subfloat[Fashion-MNIST (Dirichlet, $\alpha_d=0.1$) \label{test_Fashion_MNIST_Dirichlet}]{
		\includegraphics[width=0.49\linewidth]{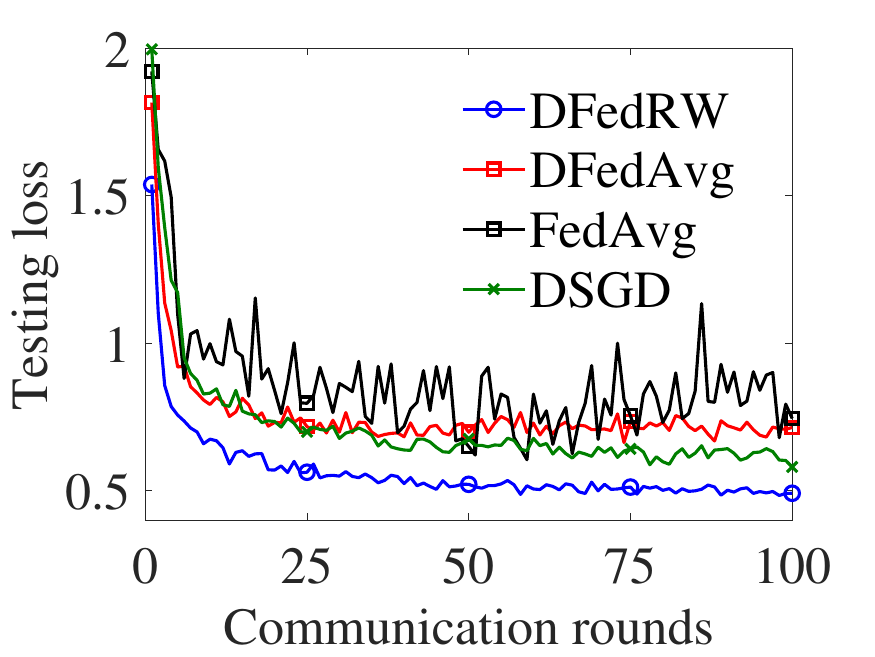}
	}
	
	\caption{Efficiency comparison of DSGD, FedAvg, DFedAvg and DFedRW in training 3FNN for Non-IID image classification with heterogeneous label distributions and sample sizes across devices.}
	\label{fig_Dirichlet}
\end{figure}

\begin{figure*}[!t]
	\centering	
	\subfloat[MNIST ($u$=100, $h$=50) \label{acc_iid_sys50}]{
		\includegraphics[width=0.24\linewidth]{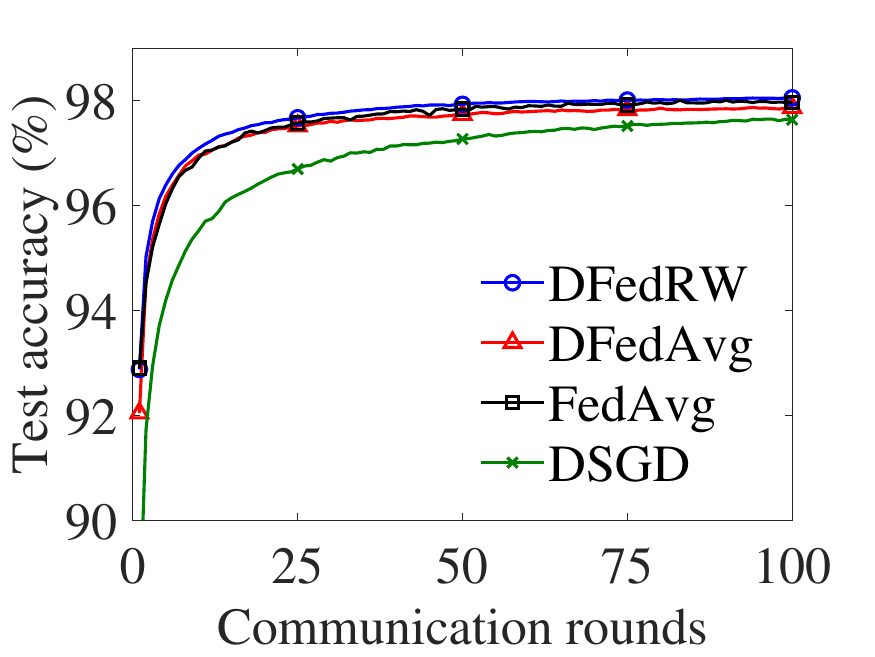}
	}
	\subfloat[MNIST ($u$=100, $h$=90) \label{acc_iid_sys90}]{
		\includegraphics[width=0.24\linewidth]{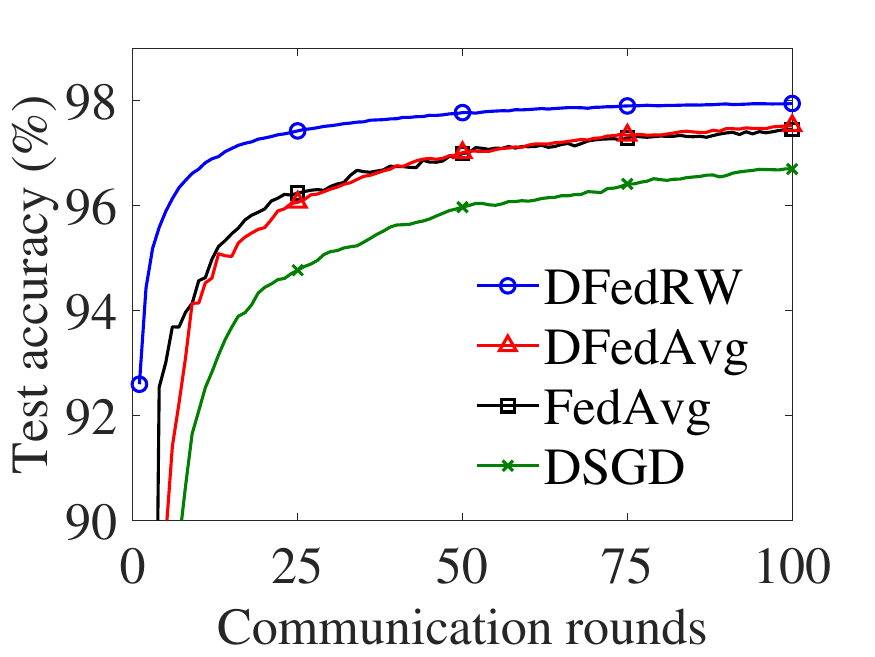}
	}
	\subfloat[MNIST ($u$=0, $h$=50) \label{acc_noniid_sys50}]{
		\includegraphics[width=0.24\linewidth]{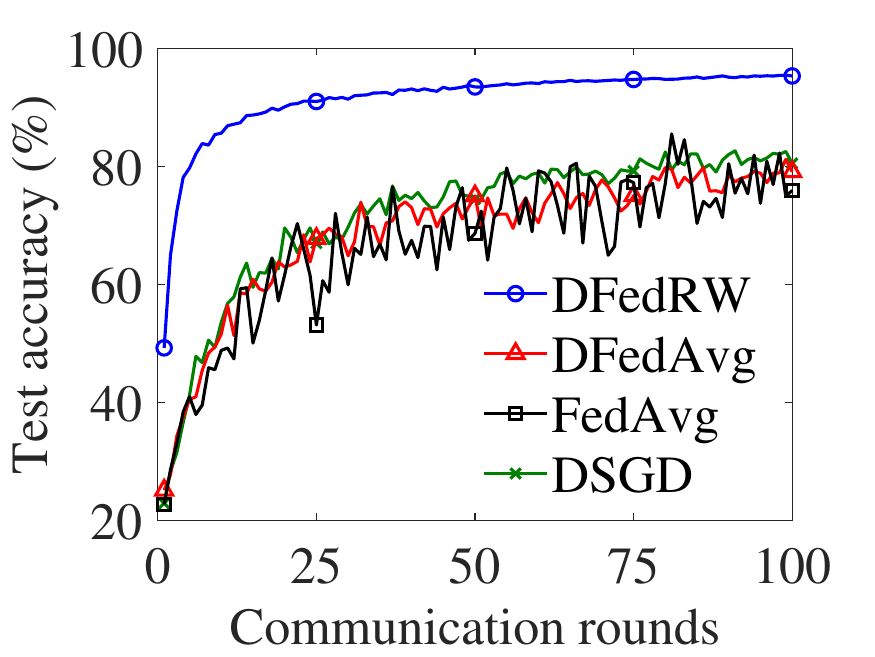}
	}
	\subfloat[MNIST ($u$=0, $h$=90) \label{acc_noniid_sys90}]{
		\includegraphics[width=0.24\linewidth]{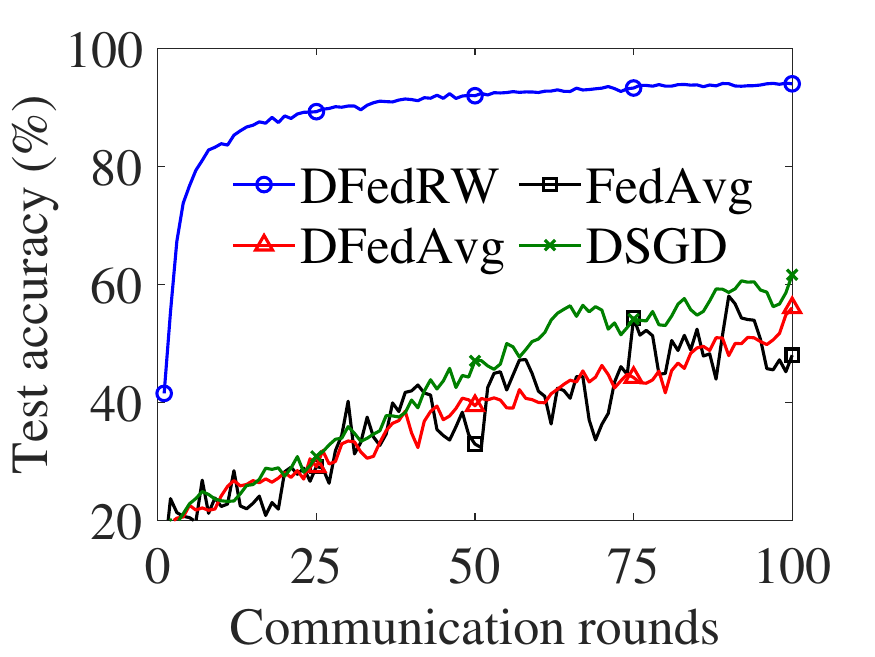}
	}
	
	\vspace{-0.4cm} 
	
	\subfloat[Fashion-MNIST ($u$=100, $h$=50) \label{acc_iid_sys50_FMNIST}]{
		\includegraphics[width=0.24\linewidth]{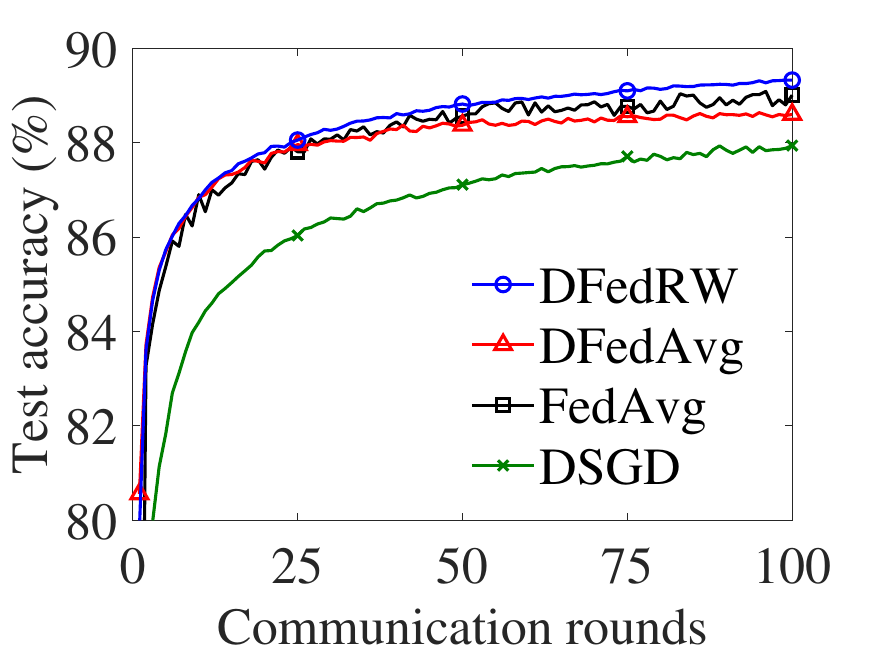}
	}
	\subfloat[Fashion-MNIST ($u$=100, $h$=90) \label{acc_iid_sys90_FMNIST}]{
		\includegraphics[width=0.24\linewidth]{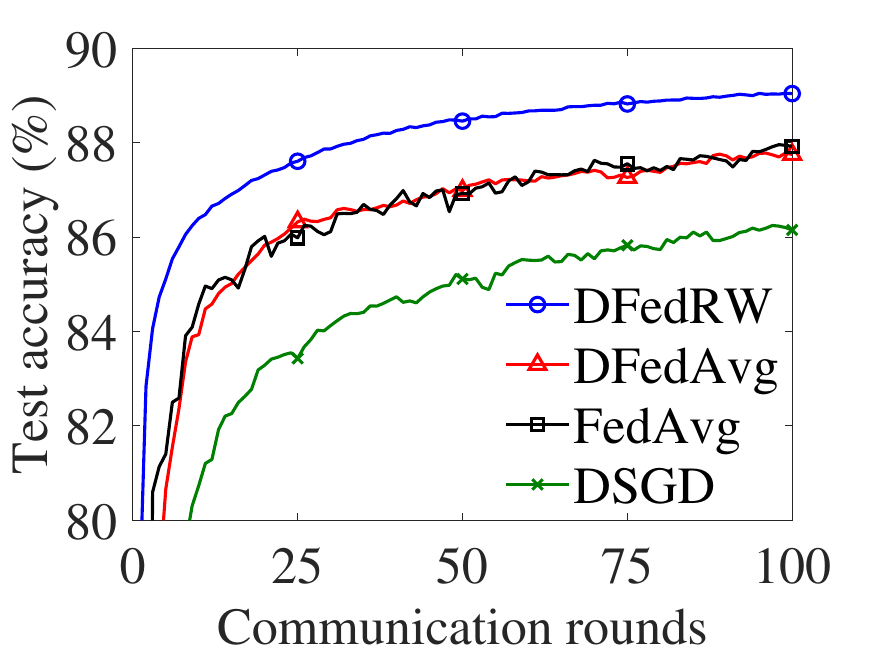}
	}
	\subfloat[Fashion-MNIST ($u$=0, $h$=50) \label{acc_noniid_sys50_FMNIST}]{
		\includegraphics[width=0.24\linewidth]{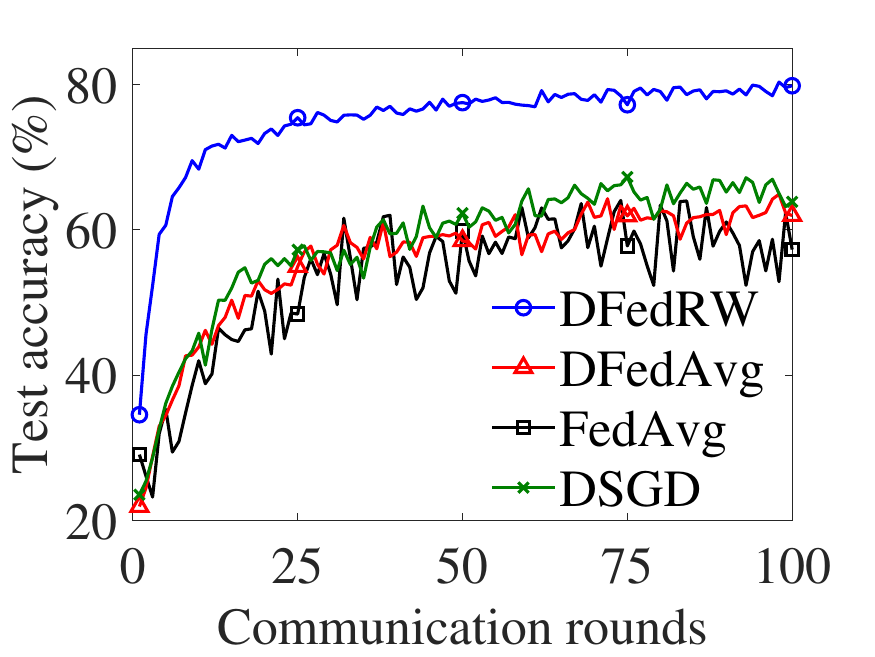}
	}
	\subfloat[Fashion-MNIST ($u$=0, $h$=90) \label{acc_noniid_sys90_FMNIST}]{
		\includegraphics[width=0.24\linewidth]{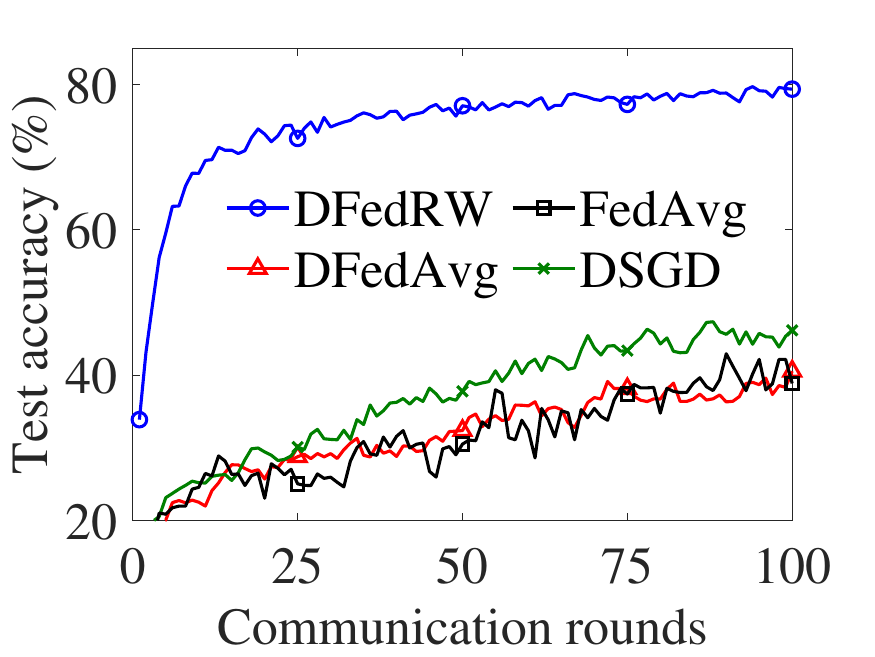}
	}
	
	\caption{Efficiency comparison of DSGD, FedAvg, DFedAvg and DFedRW in training 3FNN for different system heterogeneity in image classification. The left two columns are fixed as IID, while the right two columns are fixed as 100\% Non-IID, comparing performance under 50\% and 90\% system heterogeneity.}
	\label{fig_sys}
\end{figure*}

\begin{figure}[!t]
	\centering	
	\subfloat[Fashion-MNIST ($u$=100, $h$=50) \label{test_Fashion_MNIST_hete50_iid}]{
		\includegraphics[width=0.49\linewidth]{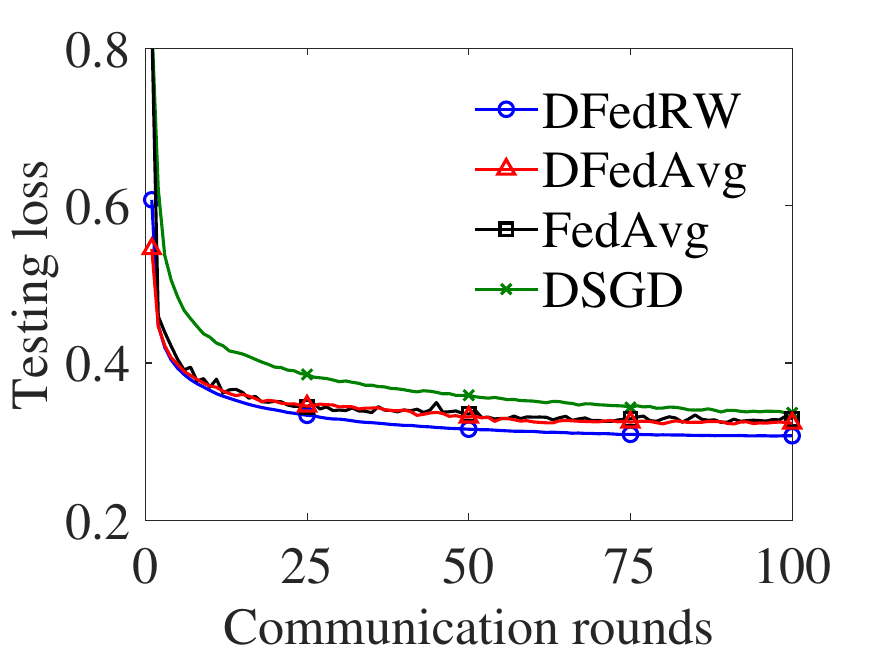}
	}
	\subfloat[Fashion-MNIST ($u$=100, $h$=90) \label{test_Fashion_MNIST_hete90_iid}]{
		\includegraphics[width=0.49\linewidth]{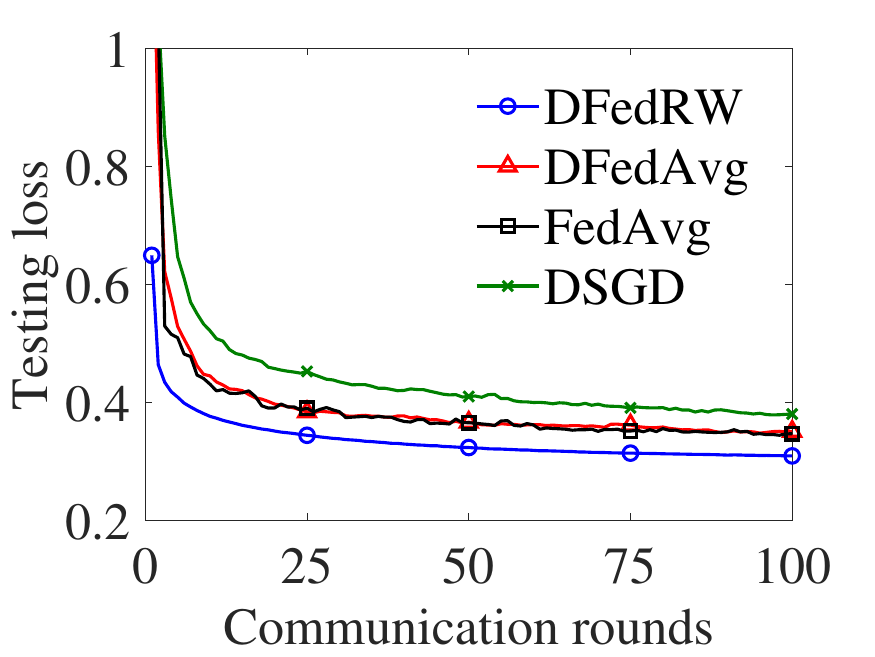}
	}
	
	\vspace{-0.4cm} 
	
	\subfloat[Fashion-MNIST ($u$=0, $h$=50) \label{test_Fashion_MNIST_hete50_noniid}]{
		\includegraphics[width=0.49\linewidth]{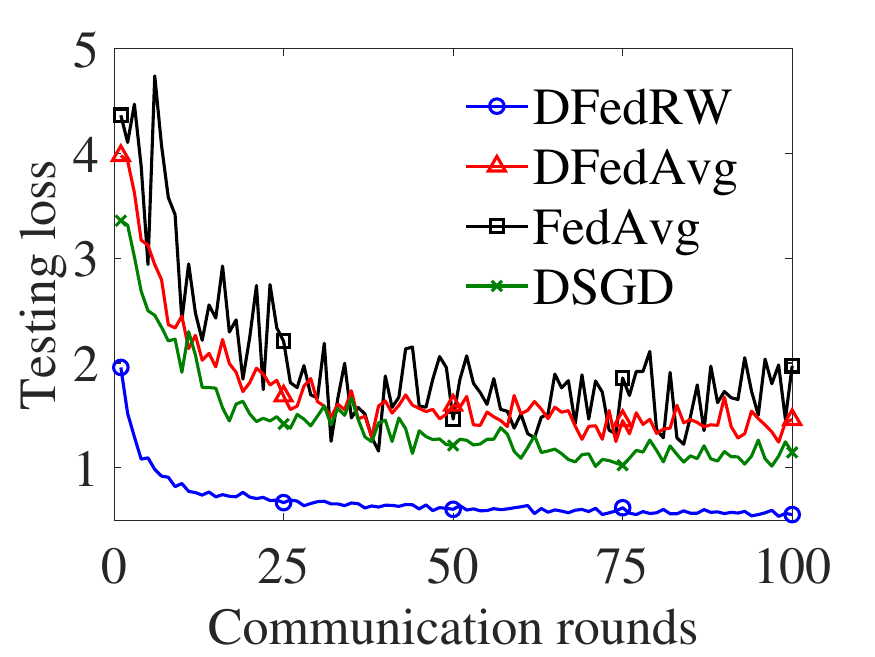}
	}
	\subfloat[Fashion-MNIST ($u$=0, $h$=90) \label{test_Fashion_MNIST_hete90_noniid}]{
		\includegraphics[width=0.49\linewidth]{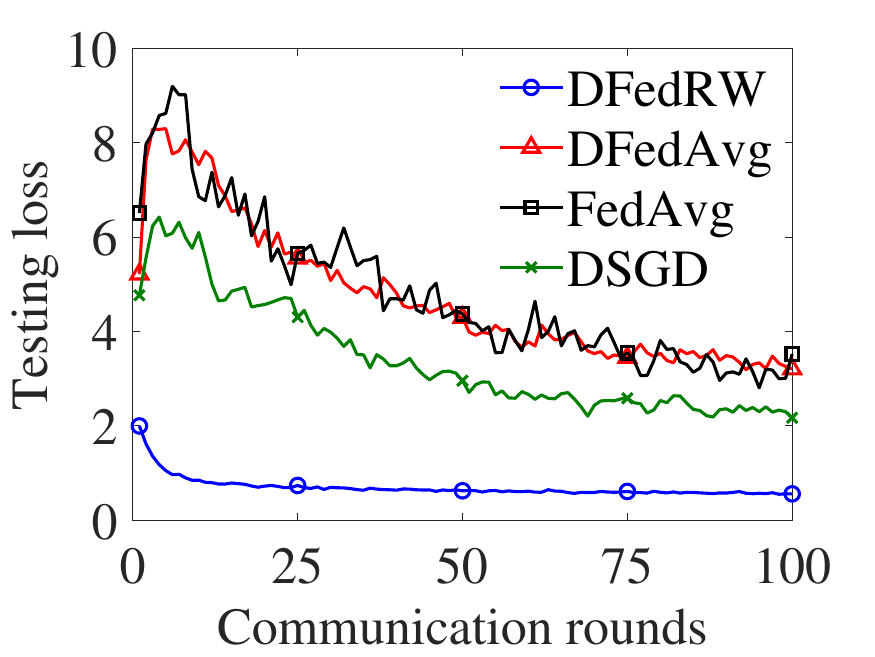}
	}
	
	\caption{Testing loss comparison of DSGD, FedAvg, DFedAvg and DFedRW in training 3FNN for different system heterogeneity in Fashion-MNIST.}
	\label{fig_sys_testloss}
\end{figure}

Secondly, we assess how effectively DFedRW combats system heterogeneity. In Fig. \ref{fig_sys}, we present the test accuracies of training the 3FNN under varying levels of system heterogeneity, and give the testing losses in Fig. \ref{fig_sys_testloss}. In the IID setting, the presence of 50\% heterogeneous devices has minimal impact on the accuracy for all algorithms. This is because the baselines dropped fewer stragglers, resulting in less data loss, and the discarded data are also IID. As the proportion of heterogeneous devices increases, when $h = 90$, system heterogeneity further reduces the test accuracy and convergence rates of the baseline algorithms, while DFedRW is largely unaffected. When the Non-IID setting is introduced, the impact of system heterogeneity on model accuracy and convergence rate is magnified. The convergence rates of the baselines sharply decrease as the level of heterogeneity increases, and the oscillation increases. Under the two levels of system heterogeneity ($h=50$ and $h=90$), DFedRW consistently reaches 95\% and 80\% accuracies at a comparable rate in two datasets. This is attributed to its ability to aggregate partial updates from these stragglers, whereas the baselines dropped them. At the highest level of test heterogeneity ($u = 0$, $h = 90$), DFedRW outperforms the baselines by an average of 38.8\% and 37.5\% in accuracy, respectively.

In Fig. \ref{fig_data}, \ref{fig_Dirichlet} and \ref{fig_sys}, we observe that different algorithms exhibit varying oscillation amplitudes and frequencies in test accuracy under different settings. This phenomenon primarily occurs in highly heterogeneous data distributions (i.e., $u=0$ or $\alpha_d=0.1$). Among the evaluated algorithms, DFedRW demonstrates the most stable performance across all settings, while FedAvg exhibits the highest oscillation amplitude. The primary cause of this behavior is device data drift in heterogeneous environments. When data is highly Non-IID, local model updates at each device may diverge significantly, leading to unstable global model updates. Specifically, FedAvg accumulates drift due to multiple local updates before aggregation. Since all device updates are aggregated at a central server simultaneously, the global model is directly influenced by the varying update directions, resulting in stronger oscillations. In contrast, DFedRW mitigates drift accumulation by enabling local model updates to traverse a broader range of heterogeneous data via random walks, thereby improving training stability. Additionally, DFedAvg and DSGD also exhibit lower oscillation amplitudes than FedAvg. DFedAvg and DSGD employ decentralized aggregation, where each device exchanges model parameters only with a small set of neighbors rather than relying on a central server. This localized aggregation reduces the rapid accumulation of drift, leading to a more stable training process. Meanwhile, DSGD performs fewer local updates per round compared to FedAvg, further limiting the impact of local drift and resulting in smaller oscillations.

\subsection{Comparison of DFedRW across Various Graphs}
\begin{figure}[!t]
	\centering	
	\subfloat[MNIST ($u$=100, $h$=0) \label{graph_acc_iid_sys0}]{
		\includegraphics[width=0.48\linewidth]{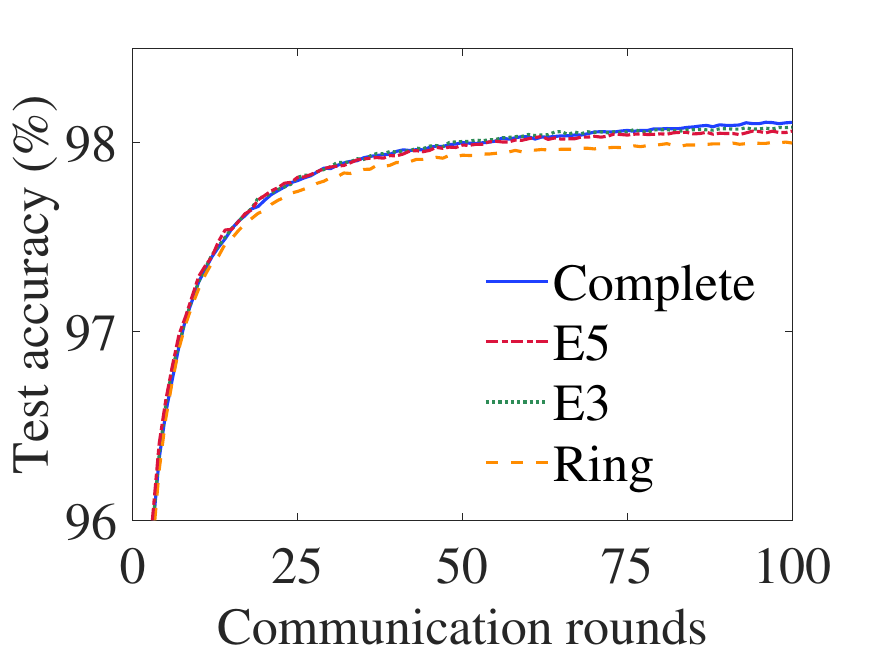}
	}
	\subfloat[MNIST ($u$=100, $h$=90) \label{graph_acc_iid_sys90}]{
		\includegraphics[width=0.48\linewidth]{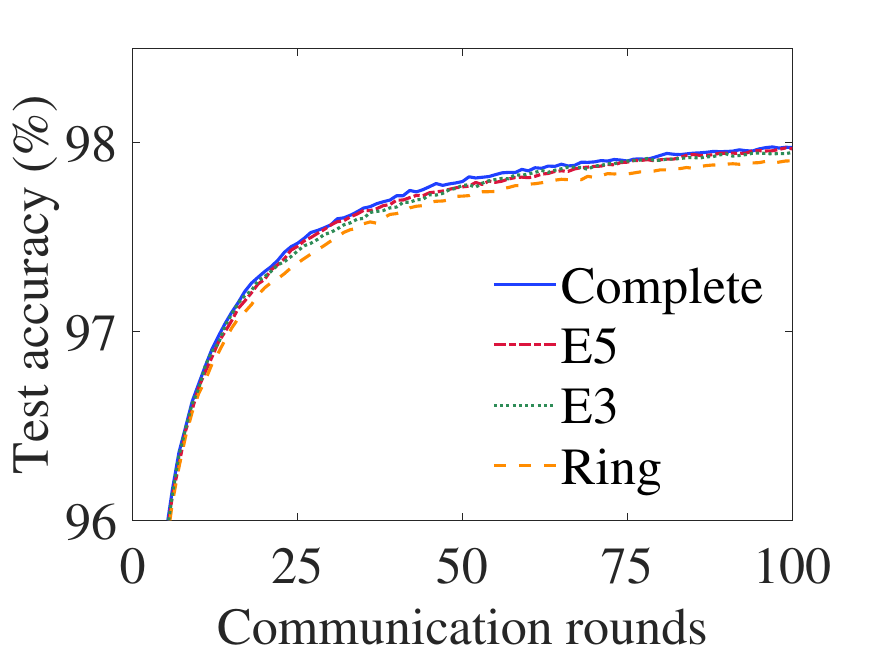}
	}
	
	\vspace{-0.4cm} 
	
	\subfloat[MNIST ($u$=0, $h$=0) \label{graph_acc_noniid0_sys0}]{
		\includegraphics[width=0.48\linewidth]{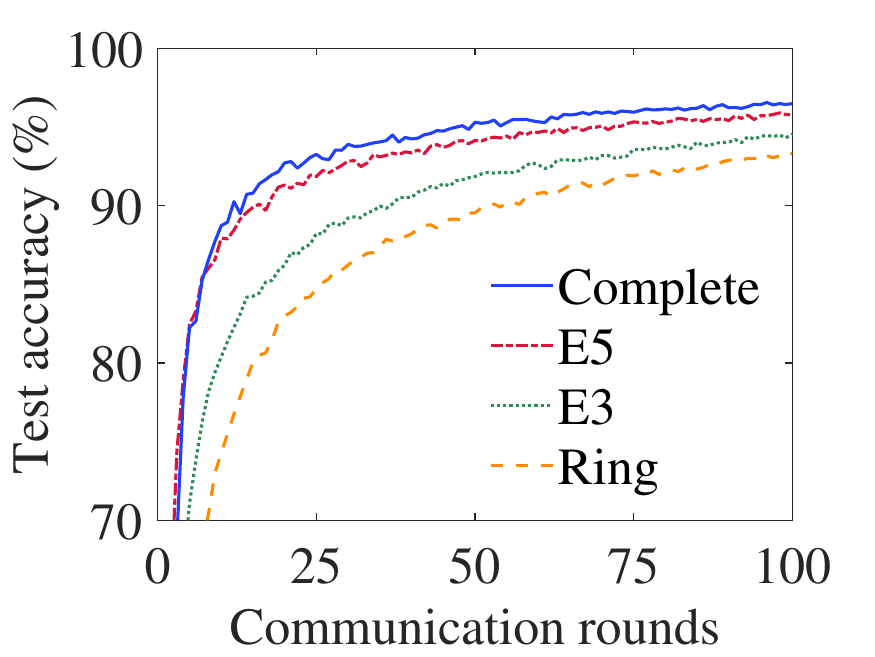}
	}
	\subfloat[MNIST ($u$=0, $h$=90) \label{graph_acc_noniid0_sys90}]{
		\includegraphics[width=0.48\linewidth]{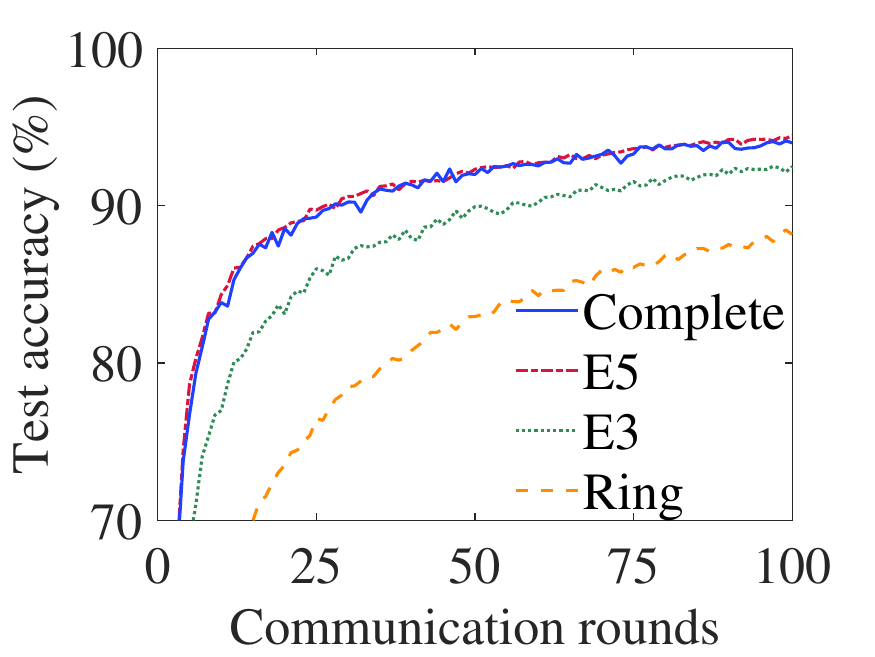}
	}
	
	\caption{Training 3FNN on different graphs for MNIST classification using DFedRW.}
	\label{fig_graph}
\end{figure}

We also consider various communication topologies on the graph, including the ring and $c$-regular expander graphs \cite{42Angel}. The parameter settings for DFedRW are consistent with those in the previous subsection. Fig. \ref{fig_graph} compares the test accuracy of the 3FNN under different graphs and varying levels of system heterogeneity. For the IID setting, regardless of the level of heterogeneities, DFedRW achieves nearly the same accuracy on ring graph, 3-expander graph (E3), and 5-expander graph (E5) as on the complete graph. Therefore, in the absence of statistical heterogeneity, using sparse topologies can effectively maintain the performance of DFedRW, allowing for reduced communication cost without compromising model accuracy. In the Non-IID setting, the accuracy on E3 and ring graphs are lower than those on the other graphs. This decrease is attributed to the limited connections. In the E3 and ring graphs, fewer than five aggregation devices reduce data diversity, weakening the effectiveness of DFedRW. When increasing system heterogeneity, the performance drop is even more pronounced on the ring graph due to its comparatively sparser topology. These findings show that while sparse graphs can reduce communication requirements in the presence of statistical heterogeneity in local data, overly sparse topologies can cause substantial performance degradation. Thus, designing a reasonable topology to achieve optimal trade-off between reduced communication and model performance is our future focus.

\begin{figure}[!t]
	\centering	
	\subfloat[MNIST ($u$=100, $h$=0) \label{quantized_acc_iid_sys0}]{
		\includegraphics[width=0.48\linewidth]{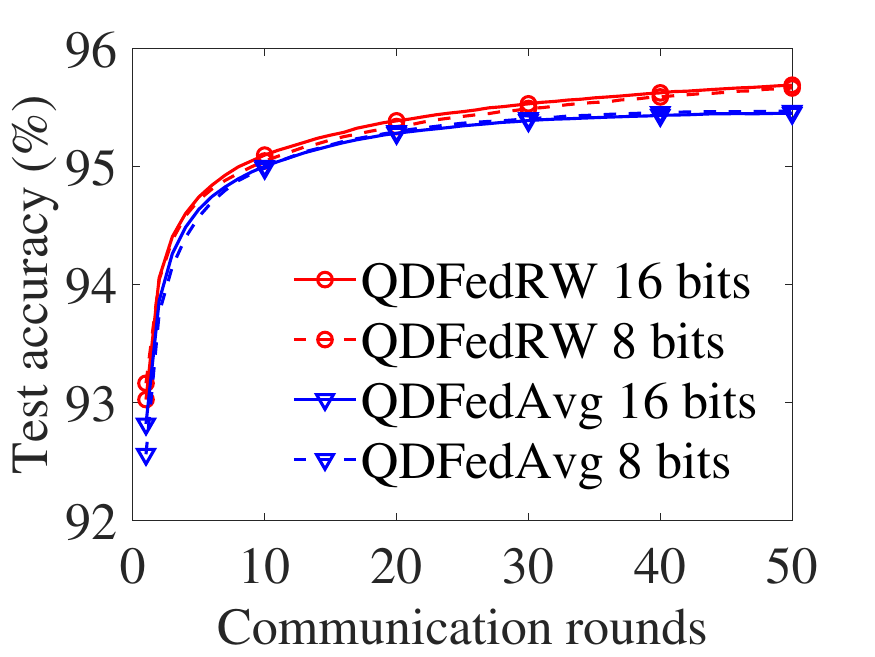}
	}
	\subfloat[MNIST ($u$=0, $h$=90) \label{quantized_acc_noniid_sys90}]{
		\includegraphics[width=0.48\linewidth]{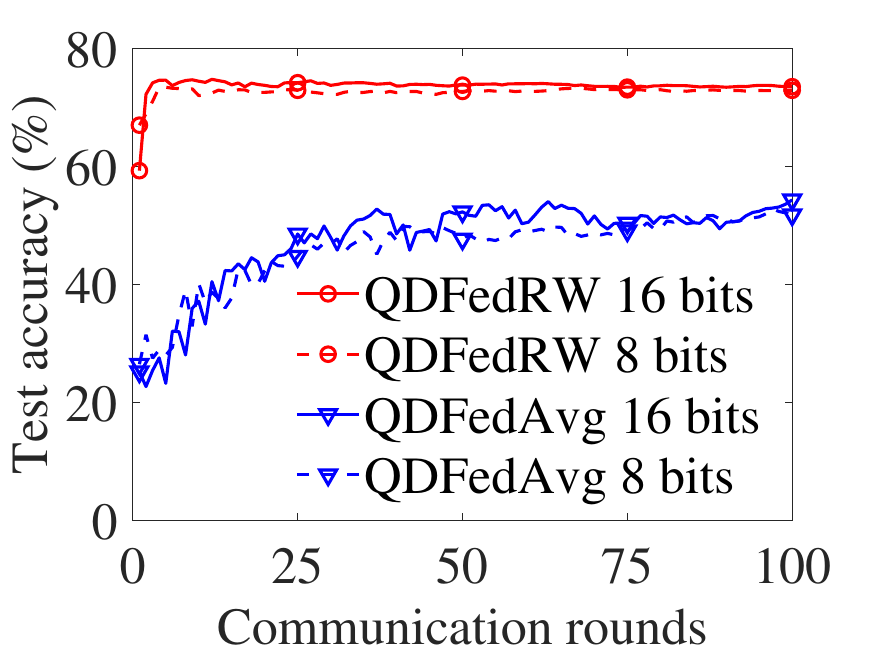}
	}
	
	\vspace{-0.4cm} 
	
	\subfloat[Fashion-MNIST ($u$=100, $h$=0) \label{quantized_acc_iid_sys0_FMNIST}]{
		\includegraphics[width=0.48\linewidth]{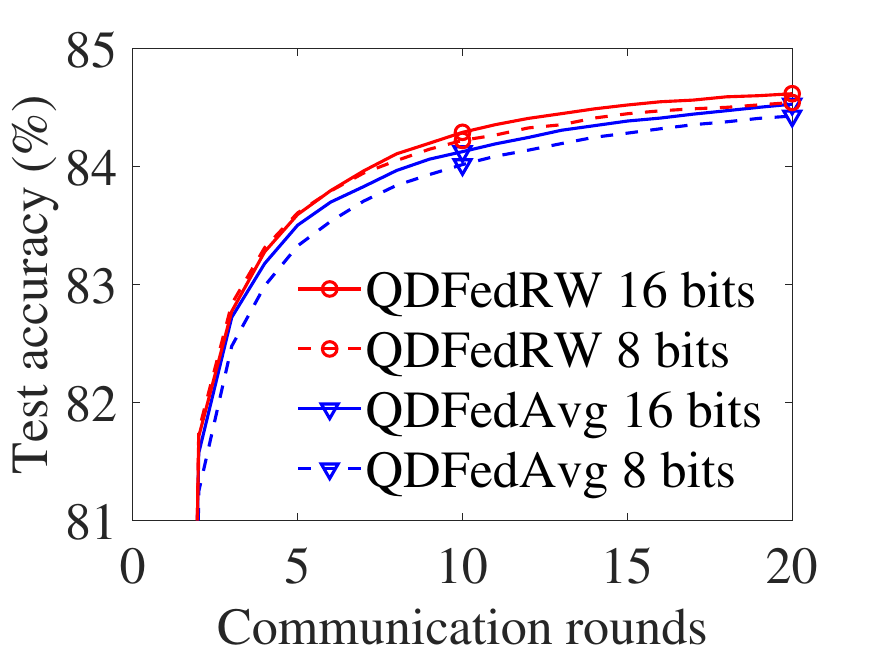}
	}
	\subfloat[Fashion-MNIST ($u$=0, $h$=90) \label{quantized_acc_noniid_sys90_FMNIST}]{
		\includegraphics[width=0.48\linewidth]{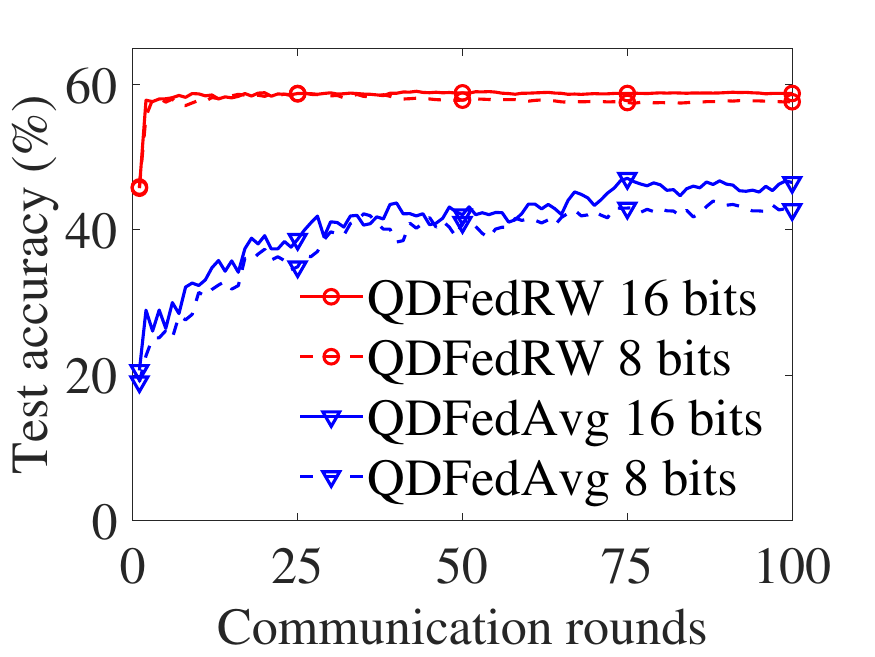}
	}
	
	\caption{Efficiency comparison of QDFedRW and QDFedAvg in training 2FNN for IID and Non-IID image classification with different communication bits.}
	\label{fig_quantized}
\end{figure}

\subsection{Efficiency of DFedRW}
To validate the efficiency of DFedRW under different communication bits and random walk epochs, we compared it with DFedAvg. Fig. \ref{fig_quantized} shows the results of training the 2FNN for MNIST and Fashion-MNIST classification using these two algorithms under different communication bits. We set $q=0.999$, with the number of aggregation devices per communication round being 20, the batch size set to 50, and both the DFedRW random walk epochs and DFedAvg local epochs set to 5. In a fully homogeneous setting, the quantized DFedRW achieves 0.23\% and 0.07\% higher accuracies than the quantized DFedAvg. This improvement increases to 18.6\% and 12.3\% under pronounced statistical and system heterogeneity. Notably, reducing the communication rate from 16 bits to 8 bits has no impact on the performance of either algorithm.

Fig. \ref{fig_epoch} illustrates the test accuracies of training the 3FNN classifier using two algorithms under different epochs. We fixed the full precision parameters for communication and only changed the number of aggregation devices to 5, as well as $q=0.501$. It can be observed that in the absence of heterogeneity, increasing the number of epochs improves the convergence rate of both algorithms, with DFedRW showing a greater improvement. In the scenario of Non-IID and 90\% system heterogeneity, more epochs do not aid the training of DFedAvg, whereas DFedRW experiences an improvement in accuracy.
\begin{figure}[!t]
	\centering	
	\subfloat[MNIST ($u$=100, $h$=0) \label{epoch_acc_iid_sys0}]{
		\includegraphics[width=0.48\linewidth]{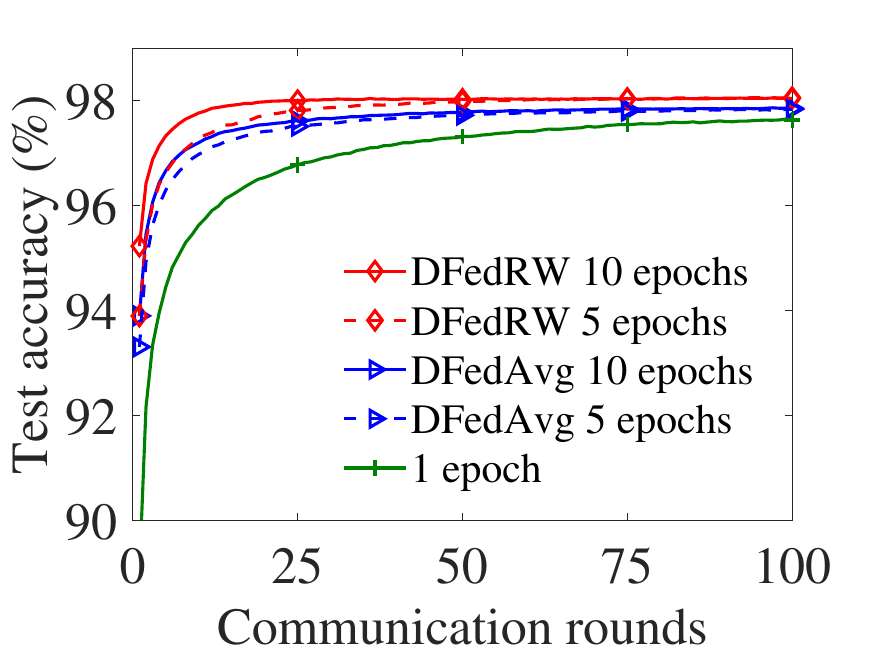}
	}
	\subfloat[MNIST ($u$=0, $h$=90) \label{epoch_acc_noniid_sys90}]{
		\includegraphics[width=0.48\linewidth]{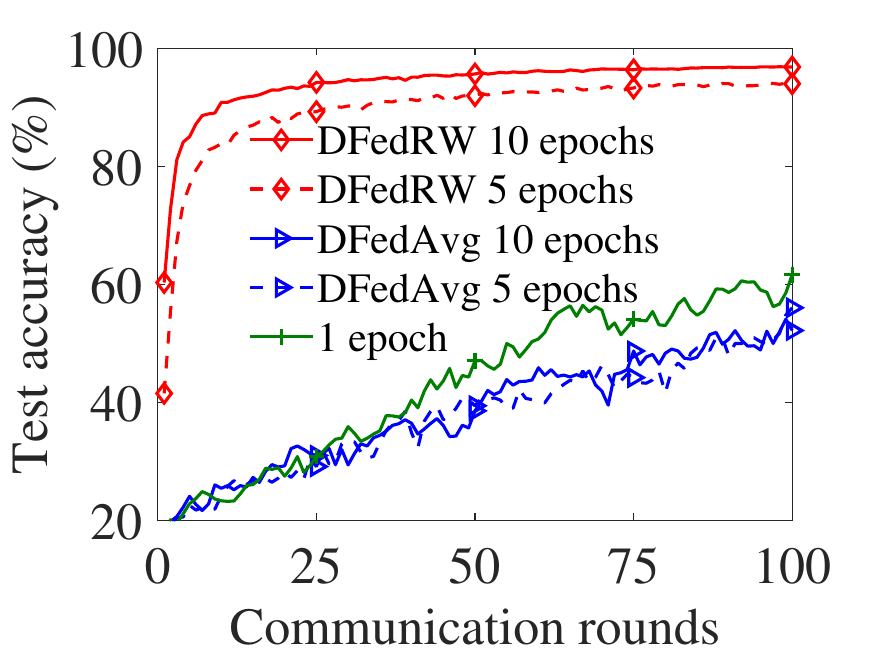}
	}
	\caption{Efficiency comparison of DFedRW and DFedAvg in training 2FNN for IID and Non-IID MNIST classification with different epochs. DFedAvg performs local epochs, while DFedRW executes random walk epochs. Both algorithms that execute 1 epoch are essentially the same.}
	\label{fig_epoch}
\end{figure}

Based on the theoretical analysis in Theorems \ref{theorem_1} and \ref{theorem_2}, we investigate how data and system heterogeneity, network topology, and quantization affect the actual convergence behavior. To this end, we set $M=20$, $K=3$ and $R=5$, while keeping other settings consistent with DFedRW in Section VI-B. Fig. \ref{fig_convergence} compares the empirical convergence bounds under different conditions, showing that DFedRW ($u=1$, $h=0$) achieves the lowest empirical convergence bound when none of these factors are present, which is consistent with the tightest theoretical upper bound in Theorem \ref{theorem_1}. Our experimental results further reveal that increasing data and system heterogeneity, reducing network connectivity, and lowering quantization precision all lead to a more relaxed convergence bound. Among these, heterogeneity and network sparsity play a dominant role in governing the degree of bound relaxation, whereas the impact of quantization is relatively minor. These observations align well with the theoretical insights presented in Theorem \ref{theorem_2}.

\begin{figure}[!t]
	\centering	
	\subfloat[MNIST \label{acc_mnist_convergence}]{
		\includegraphics[width=0.48\linewidth]{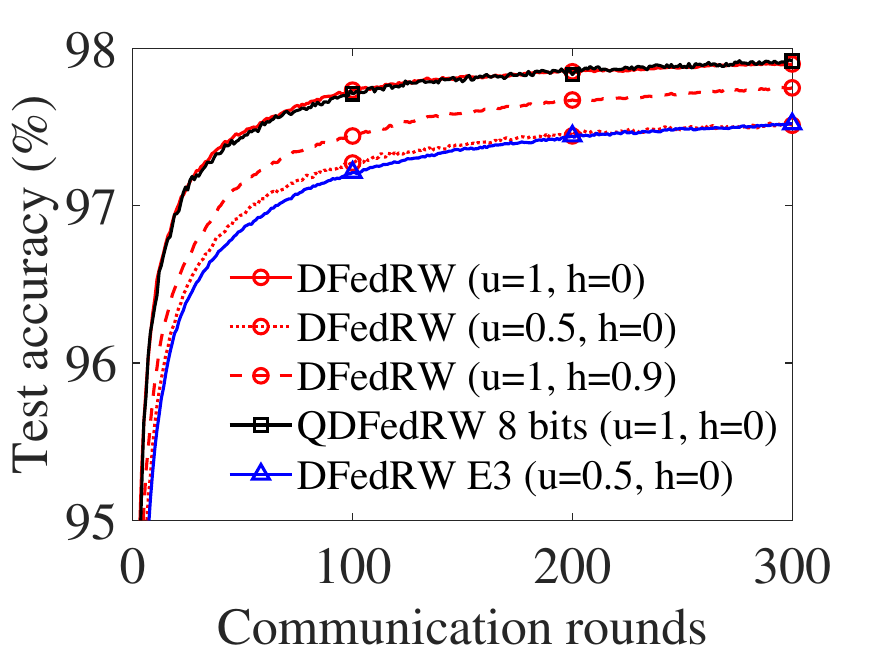}
	}
	\subfloat[MNIST \label{test_mnist_convergence}]{
		\includegraphics[width=0.48\linewidth]{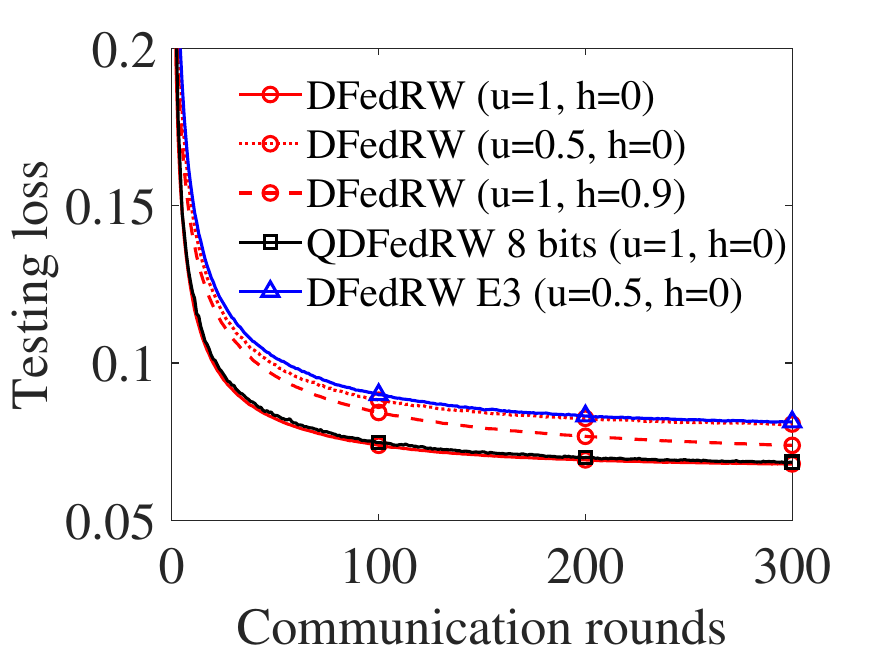}
	}
	
	\vspace{-0.4cm} 
	
	\subfloat[Fashion-MNIST \label{acc_fashion_mnist_convergence}]{
		\includegraphics[width=0.48\linewidth]{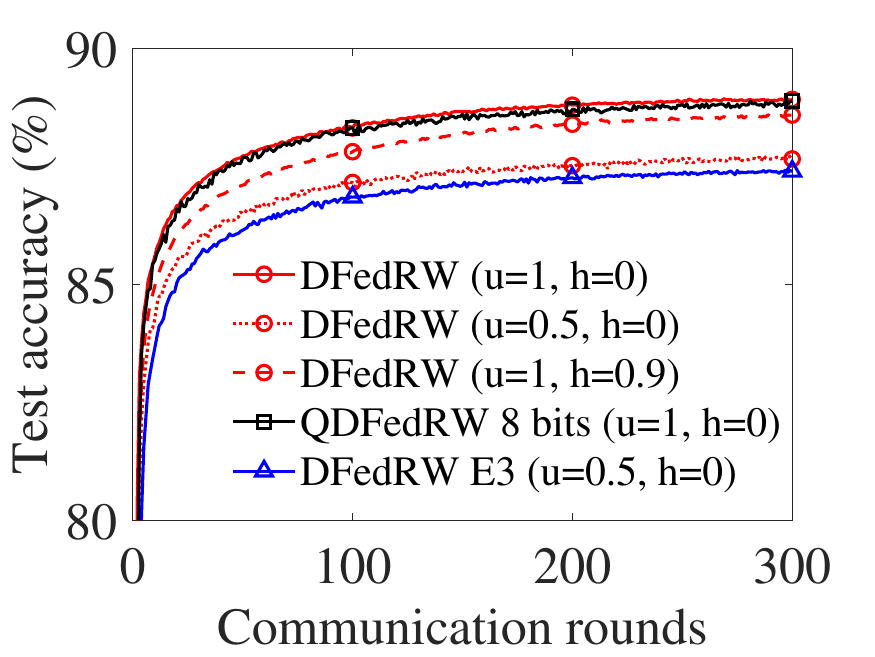}
	}
	\subfloat[Fashion-MNIST \label{test_fashion_mnist_convergence}]{
		\includegraphics[width=0.48\linewidth]{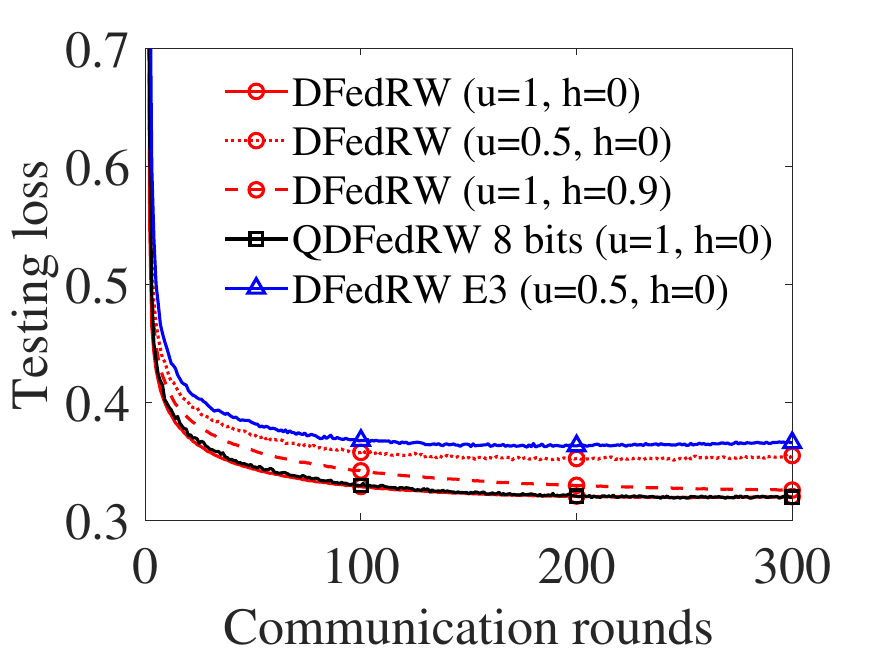}
	}
	
	\caption{Comparison of the impact of relaxing constraints on the empirical convergence bound.}
	\label{fig_convergence}
\end{figure}

\subsection{Comparison of Communication Costs under Various Factors}
We first evaluate DFedRW, DFedRW-E3, 8-bit QDFedRW, and full-precision fully connected baselines under moderate heterogeneity (i.e., $u=50$, $h=50$), analyzing their accuracy curves as a function of the communication cost (in MB) of the busiest device. We set $M = 5$, $R = 5$, with $|\mathcal{N}_A| = 5$ for the complete graph and $|\mathcal{N}_A| = 3$ for the E3 graph. In FedAvg, the server acts as the busiest device, sending the global model to $M$ selected devices and collecting their updates in each round, resulting in a total communication cost of $\mathcal{C}_A = 2M\phi$, where $\phi$ denotes the size of the transmitted model parameters. For DFedRW, the busiest device may participate in both random walk updates and decentralized aggregation. The total communication cost $\mathcal{C}_R$ is given by
\begin{equation}
\mathcal{C}_R=\mathcal{C}_{\mathrm{upd}}+\mathcal{C}_{\mathrm{agg}}=2 \sum_{m=1}^M \theta_{i^*}^m \Gamma_{i^*}^m \phi+ \left|\mathcal{N}_c(i^*)\right| \left|\mathcal{N}_A(i^*)\right| \phi,
\end{equation}
where $\mathcal{C}_{\mathrm{upd}}$ accounts for the communication during updates, and $\mathcal{C}_{\mathrm{agg}}$ corresponds to the cost of aggregation. $i^*$ represents the busiest device in the current round, $\theta_{i^*}^m \in \{0,1\}$ indicates whether $i^*$ is selected in chain $m$, and $\Gamma_{i^*}^m$ denotes the number of times $i^*$ is selected in chain $m$. The terms $\left|\mathcal{N}_c\left(i^*\right)\right|$ and $\left|\mathcal{N}_A\left(i^*\right)\right|$ denote the number of devices receiving the model from device $i^*$ and the number of neighboring devices participating in aggregation, respectively. In contrast, DFedAvg and DSGD only support local updates, so the communication cost for the busiest device is limited to $\mathcal{C}_{\mathrm{agg}}$.

As shown in Fig. \ref{fig_CommOH}, 8-bit QDFedRW achieves the highest accuracy per MB and the fastest convergence rate, followed by DFedRW-E3 and DFedRW. Notably, even DFedRW outperforms all baselines in terms of accuracy per MB for the busiest device. These results indicate that DFedRW attains superior model accuracy and convergence rate without imposing additional communication costs on the busiest device. Moreover, as heterogeneity increases (i.e., $u=0$, $h=50$), the advantages of 8-bit QDFedRW, DFedRW-E3, and DFedRW in accuracy and convergence rate become more pronounced.

\begin{figure}[!t]
	\centering	
	\subfloat[MNIST ($u$=50, $h$=50) \label{acc_mnist_u50_CommOH}]{
		\includegraphics[width=0.48\linewidth]{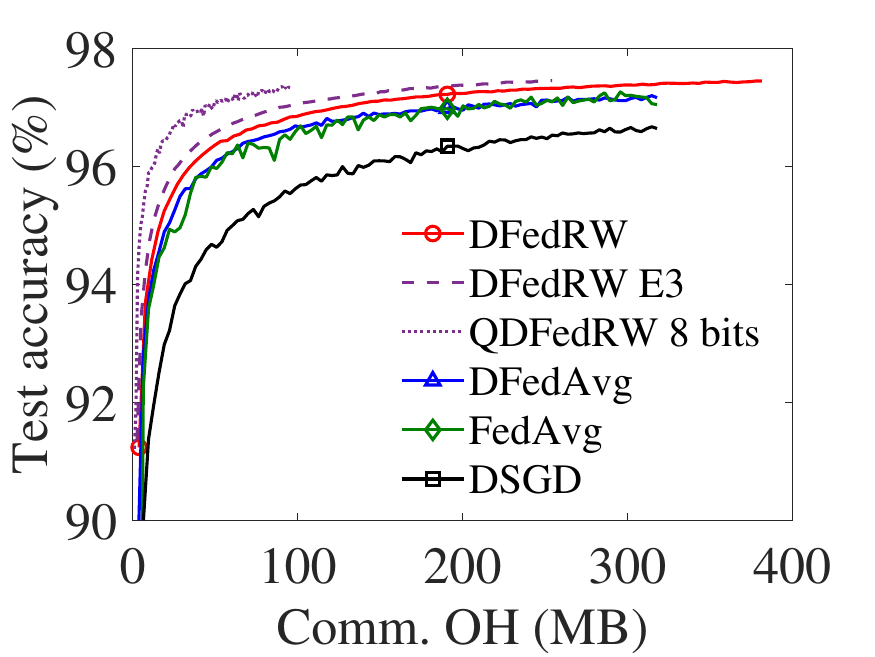}
	}
	\subfloat[MNIST ($u$=50, $h$=50) \label{test_mnist_u50_CommOH}]{
		\includegraphics[width=0.48\linewidth]{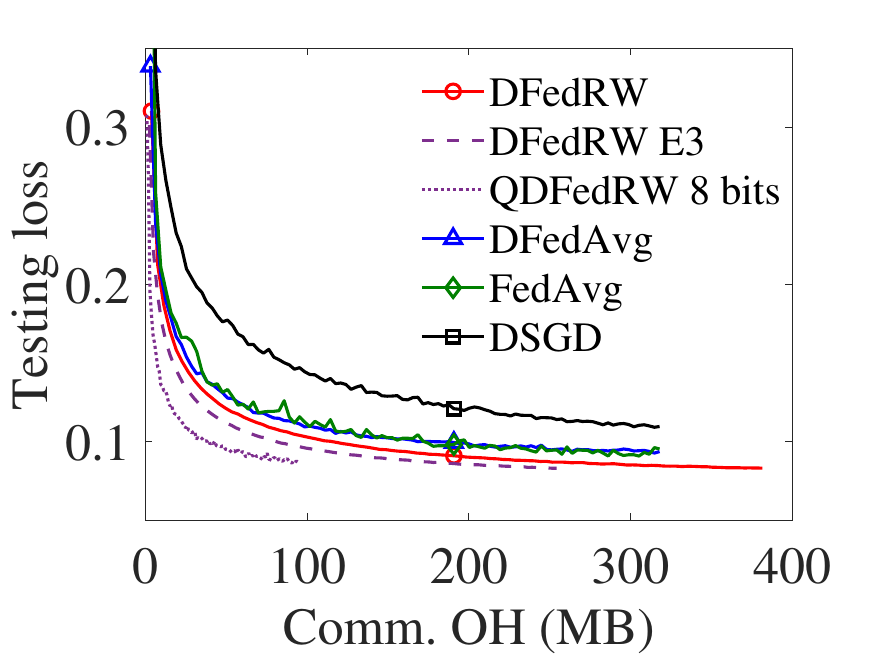}
	}
	
	\vspace{-0.4cm} 
	
	\subfloat[MNIST ($u$=0, $h$=50) \label{acc_mnist_u0_CommOH}]{
		\includegraphics[width=0.48\linewidth]{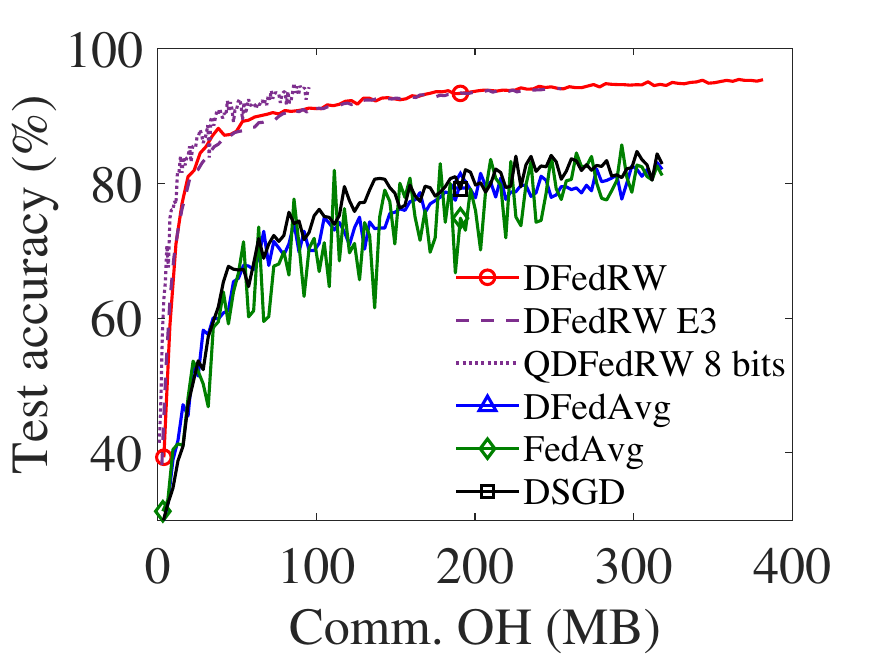}
	}
	\subfloat[MNIST ($u$=0, $h$=50) \label{test_mnist_u0_CommOH}]{
		\includegraphics[width=0.48\linewidth]{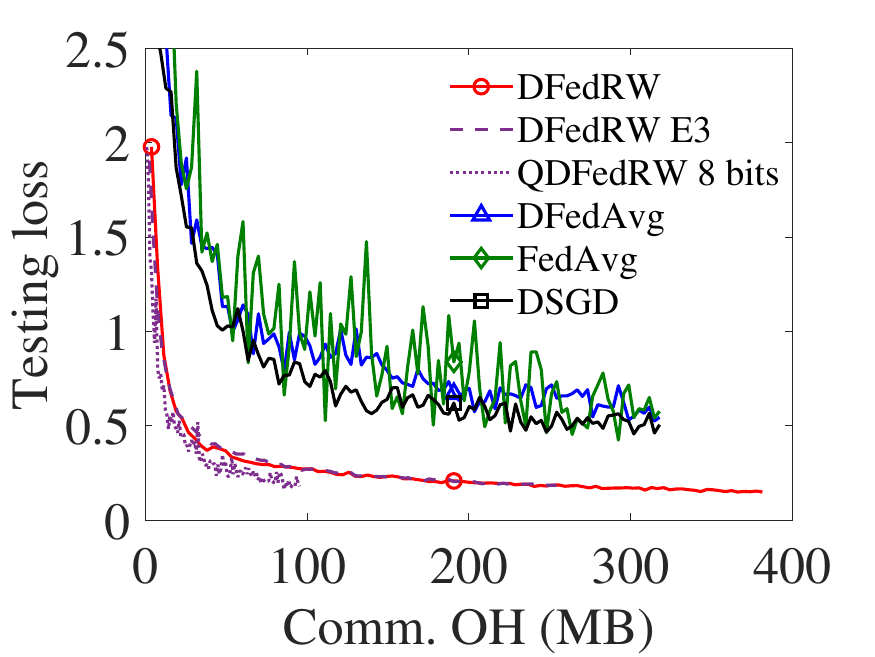}
	}
	
	\caption{Impact of network topology and data size on the communication cost of the busiest device when training 3FNN for MNIST classification with DFedRW and baselines under varying heterogeneity. Communication overhead is denoted as Comm. OH.}
	\label{fig_CommOH}
\end{figure}

\subsection{Large-scale Language Modelling}
To evaluate the effectiveness of DFedRW in real-world applications, we conducted large-scale language modeling experiments on the Reddit dataset. With up to 83,293 distributed devices available, we define the aggregation neighbor set of device $i \in i_m^{t,last}$, $m=\{1,2,... ,M\}$ in round $t$ as $\mathcal{N}_A(i)=\{i_1^{t,last},...,i_{m-1}^{t,last},i_{m+1}^{t,last},...,i_M^{t,last}\}$, where $i_m^{t,last}$ represents the last device in the $m$-th random walk chain. In each training round, the starting device of each random walk chain is directly inherited from the last device of the corresponding chain in the previous round, i.e., $i_m^{t,0}=i_m^{t-1,last}$. This allows us to select only a small number of devices ($M<<n$) per round while maintaining strong performance. In contrast, DFedAvg and DSGD require a significantly larger number of devices in each round, rendering them impractical in this scenario.

\begin{figure}[!t]
	\centering	
	\subfloat[  Reddit (M=10) \label{acc_reddit}]{
		\includegraphics[width=0.48\linewidth]{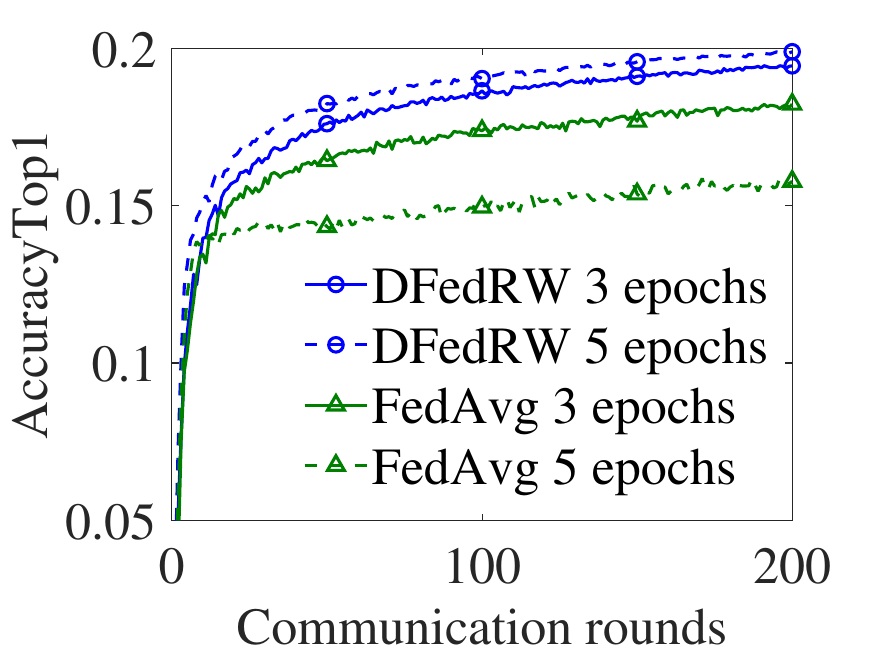}
	}
	\subfloat[  Reddit (M=10) \label{test_reddit}]{
		\includegraphics[width=0.48\linewidth]{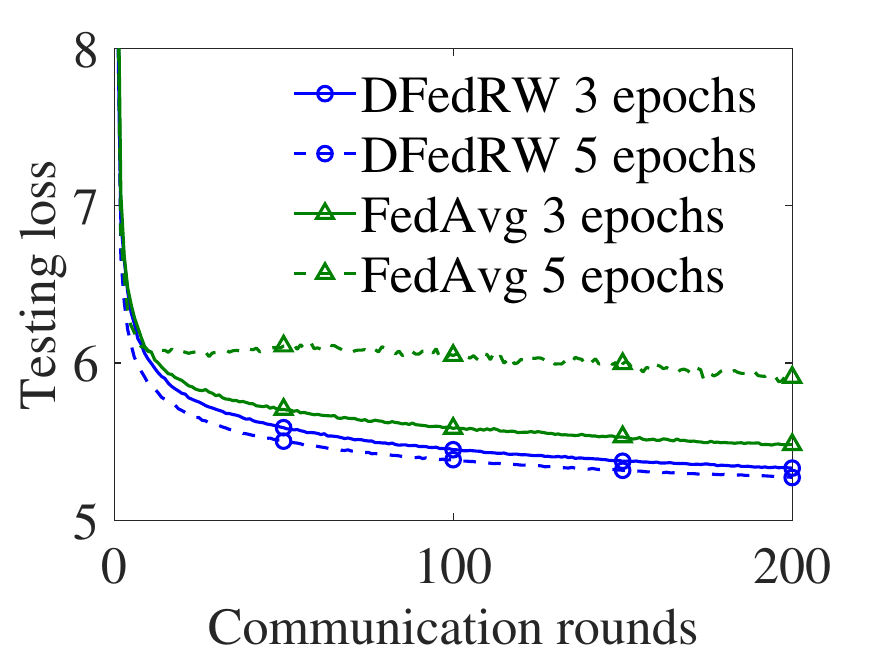}
	}
	\caption{Efficiency comparison of FedAvg and DFedRW in training LSTM for large-scale language modelling. FedAvg performs local epochs, while DFedRW performs random walk epochs.}
	\label{fig_reddit}
\end{figure}

We set the batch size of both DFedRW and FedAvg to 64 and adopt a decreasing learning rate of $\eta^{\bar{k}}=\frac{1}{0.5 \bar{k}^{0.499}}$. For DFedRW, we set $M=10$, while FedAvg selects 10 devices per round for global aggregation. We use \textit{AccuracyTop1} to evaluate the probability that the word assigned the highest probability by the model is correct \cite{43McMahan}. As shown in Fig. \ref{fig_reddit}, when both algorithms perform 3 epochs, DFedRW outperforms FedAvg in AccuracyTop1 by 0.012. As the number of local epochs increases to 5, FedAvg suffers from overfitting, whereas the random walk updates in DFedRW avoid overfitting, resulting in a AccuracyTop1 that is 0.005 higher than its performance at 3 epochs. Additionally, we also evaluate the impact of quantization on QDFedRW by training an LSTM model with different quantization bits on the Reddit dataset. We set $K=2$ while keeping all other configurations unchanged. The results in Fig. \ref{fig_reddit_bits} confirm the effectiveness of QDFedRW for large-scale language modeling, showing that the number of communication bits does not affect its AccuracyTop1. On the contrary, reducing the communication bits actually accelerates global convergence.

\begin{figure}[!t]
	\centering	
	\subfloat[  Reddit (M=10, K=2) \label{acc_reddit_bits}]{
		\includegraphics[width=0.48\linewidth]{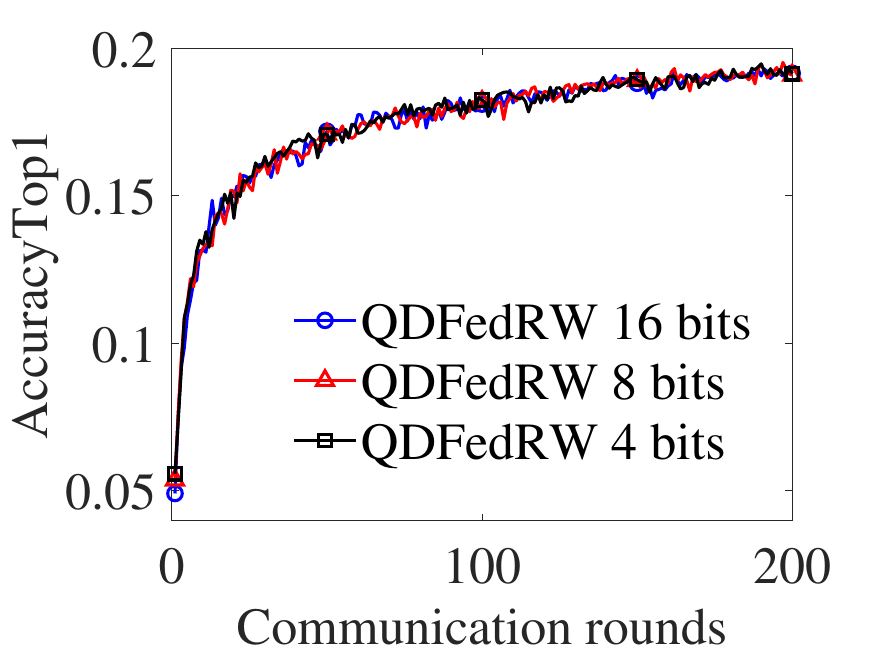}
	}
	\subfloat[  Reddit (M=10, K=2) \label{test_reddit_bits}]{
		\includegraphics[width=0.48\linewidth]{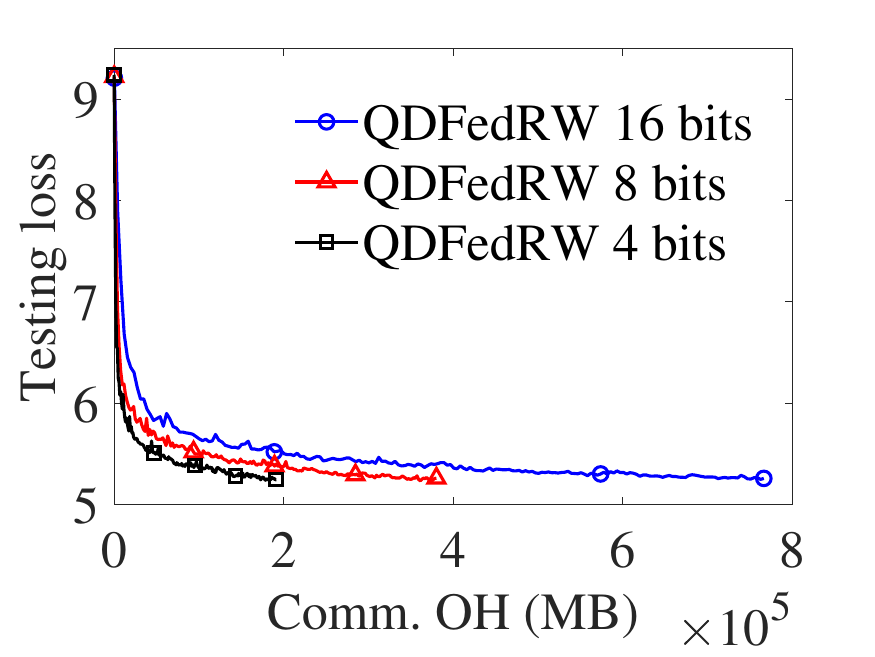}
	}
	\caption{Efficiency comparison of QDFedRW in training LSTM for language modeling with different communication bits. }
	\label{fig_reddit_bits}
\end{figure}

Furthermore, we analyze whether DFedRW becomes impractical in large-scale networks due to latency. We assume that the local computation latency for one local epoch is $T_p$, the communication latency for transmitting model updates between devices or between a device and the server is $T_c$, and the number of random walk epochs is $K$. Notably, this setup does not account for the heterogeneity in device computational capabilities, which actually favors FedAvg since its training time is determined by the slowest device, and multiple local updates exacerbate delays for slower devices. For FedAvg, the training latency per round is $T_{A} = K T_p + 2T_c$, whereas for DFedRW, the random walk updates result in a training latency of $T_{R} = KT_p + 2T_c + (K-1)T_c = KT_p + (K+1)T_c$. Both methods have the same local computation latency ($KT_p$) per round. However, since $K \geq 1$, the relative contribution of $KT_p$ to the total latency in FedAvg is higher than that in DFedRW. To highlight DFedRW’s advantages, we compare it under the most unfavorable scenario, i.e., $T_p=0$. When $K=3$, $T_{R} = 2T_{A} = 4T_c$, Table \ref{table_3} shows that even under the least favorable conditions for DFedRW, it achieves an AccuracyTop1 of at least 0.18 with significantly lower training latency than FedAvg in large-scale networks.

\begin{table}[!t]
	\centering
	\caption{Training Latency and AccuracyTop1 Comparison Under Favorable FedAvg and Unfavorable DFedRW Conditions}
	\label{table_3}
	\renewcommand{\arraystretch}{1}
	\begin{tabular}{lllll}
		\toprule
		AccuracyTop1 & 0.16 & 0.17 & 0.18 & 0.19 \\ 
		\midrule
		FedAvg & 64$T_c$ & 132$T_c$ & 316$T_c$ & 760$T_c$\\
		DFedRW & 88$T_c$ & 152$T_c$ & 252$T_c$ & 536$T_c$\\
		\bottomrule
	\end{tabular}
\end{table}

\section{Conclusion}
In this paper, we propose DFedRW and its quantized version, where local training is accomplished through random walks. Our proposed algorithms achieve higher accuracy and training rates under high statistical and system heterogeneity, while being relatively inexpensive. We establish the convergence bound of the algorithms under convex conditions and derive their relationship with heterogeneities. We also analyze when the quantized DFedRW saves communication. Extensive numerical experiments validate the advantages of our proposed algorithms in terms of convergence rate and accuracy, with these advantages increasing as heterogeneities grows.

Future research can focus on the following directions: 1) Investigating the adaptability of DFedRW in dynamic network topologies. In practice, device arrivals, departures, or movements disrupt the stationary distribution of random walks, leading to gradient discontinuities or data shifts. 2) As data distributions evolve over time, the inconsistency between training and testing distributions may affect model generalization. A key challenge is detecting and adapting to Non-IID data with concept drift. 3) Extending the theory to non-convex settings. By leveraging the Lipschitz gradient continuity assumption \cite{14Sun}, and incorporating data and system heterogeneity along with random walk updates, construct an appropriate Lyapunov function.

\begin{appendices}
\section{Proof of Theorem 1: Convergence Bound of DFedRW}
\subsection{Key Lemmas}
To prove Theorems \ref{theorem_1} and \ref{theorem_2}, we present several technical results that are utilized in the proof. The proof techniques are fundamentally inspired by the works in \cite{25Ayache}, \cite{32Sun}, and \cite{33Ayache}, and are extended to address the more challenging settings of statistical and system heterogeneity.

\begin{lemma}[Lipschitzness]
	\label{lemma_3}
	Under Assumption \ref{assumption_1}, there exists a constant $v_i \geq 0$ such that for any closed and bounded subset $\mathcal{W} \subseteq \mathbb{R}^d$ and any $\mathbf{w}_1, \mathbf{w}_2 \in \mathcal{W}$ satisfies $\left|F_i\left(\mathbf{w}_1\right)-F_i\left(\mathbf{w}_2\right)\right| \leq D\left\|\mathbf{w}_1-\mathbf{w}_2\right\|$, where $D=\sup_{i \in \mathcal{V}} v_i$.
\end{lemma}
\begin{proof}
	A proof of Lemma \ref{lemma_3} can be found in \cite{32Sun}.
\end{proof}

\newtheorem{corollary}{Corollary}
\begin{corollary}[Bounded gradient]
	\label{corollary_1}
	Under assumptions \ref{assumption_1}, for a closed bounded subset $\mathcal{W} \subseteq \mathbb{R}^d$, $\left\|\nabla F_i(\mathbf{w})\right\| \leq D$, $\forall \mathbf{w} \in \mathcal{W}$.
\end{corollary}
\begin{proof}
	A proof of Corollary \ref{corollary_1} can be found in \cite{25Ayache}, \cite{33Ayache}.
\end{proof}

\begin{lemma} Using Lemma \ref{lemma_3} and Corollary \ref{corollary_1}, we have
	\label{lemma_5}
	\begin{equation}
		\begin{aligned}
			& \sum_{k} \eta^{k} \mathbb{E}\left[F_{i^{k}}(\mathbf{w}^{k-\tau^{k}})-F_{i^{k}}(\mathbf{w}^{k})\right] \\
			& \quad \leq \frac{D^2}{2} \sum_{k} \tau^{k}(\eta^{k})^2+\frac{D^2}{2} \sum_{k} \sum_{n=k-\tau^{k}}^{k-1}\left(\eta^{n}\right)^2.
		\end{aligned}
	\end{equation}
\end{lemma}

\begin{lemma} Under Assumption \ref{assumption_5}, by applying Definition \ref{definition_4} and Lemma \ref{lemma_10}, we have
	\label{lemma_6}
	\begin{equation}
		\begin{aligned}	
			& \sum_{k} \eta^{k} \mathbb{E}\left[f(\mathbf{w}^{k-\tau^{k}})-f\left(\mathbf{w}^*\right)\right] \\
			& \quad \leq \sum_{k} \eta^{k} \mathbb{E}\left[F_{i^k}(\mathbf{w}^{k-\tau^{k}})-F_{i^{k}}(\mathbf{w}^*)\right] +\sum_{k} \frac{n \eta^{k}}{2 k}.
		\end{aligned}
	\end{equation}
\end{lemma}

\begin{lemma} Under Assumption \ref{assumption_3}, by applying Lemma \ref{lemma_3} and Corollary \ref{corollary_1}, we have
	\label{lemma_7}
	\begin{equation}
		\begin{aligned}
			\sum_{k} \eta^{k} \mathbb{E}\left[f(\mathbf{w}^{k})-f(\mathbf{w}^{k-\tau^{k}})\right] \leq \frac{2 n D^2}{\ln \left(1 / \lambda_p\right)} \sum_{k=\bar{K}}^{\infty} \ln k \cdot(\eta^{k})^2.
		\end{aligned}
	\end{equation}
\end{lemma}

\subsection{Completing the Proof of Theorem 1}
\begin{proof}[Proof of Theorem 1]
	We analyze the expected decrease in the objective when executing one step of DFedRW. Given that $\mathcal{W}$ is convex, it follows that
	\begin{equation}
		\label{eq_21}
		\begin{aligned}
			& \|\mathbf{w}^{k+1}-\mathbf{w}^*\|^2\\
			&\quad=\|\boldsymbol{\Pi}_{\mathcal{W}}(\mathbf{w}^{k}-\eta^{k} \hat{\nabla} F_{i^{k}}(\mathbf{w}^{k}))-\boldsymbol{\Pi}_{\mathcal{W}}(\mathbf{w}^*)\|^2 \\
			& \quad \leq\|\mathbf{w}^{k}-\eta^{k} \hat{\nabla} F_{i^{k}}(\mathbf{w}^{k})-\mathbf{w}^*\|^2 \\
			& \quad =\|\mathbf{w}^{k}-\mathbf{w}^*\|^2-2 \eta^{k} \hat{\nabla} F_{i^{k}}(\mathbf{w}^{k})^\top (\mathbf{w}^{k}-\mathbf{w}^*)\\
			& \qquad+(\eta^{k})^2\|\hat{\nabla} F_{i^{k}}(\mathbf{w}^{k})\|^2.
		\end{aligned}
	\end{equation}
	Using the convexity of function $f_i$, as well as Lemma \ref{lemma_1} and gradient boundedness, we have
	\begin{equation}
		\label{eq_22}
		\begin{aligned}
			& \|\mathbf{w}^{k+1}-\mathbf{w}^*\|^2 \\
			& \quad \leq\|\mathbf{w}^{k}-\mathbf{w}^*\|^2 -2 \eta^{k}\left(F_{i^{k}}(\mathbf{w}^{k})-F_{i^{k}}\left(\mathbf{w}^*\right)\right) \\
			& \qquad +(\eta^{k})^2\|\nabla F_{i^{k}}(\mathbf{w}^{k})\|^2 \\
			& \quad \leq\|\mathbf{w}^{k}-\mathbf{w}^*\|^2-2 \eta^{k}\left(F_{i^{k}}(\mathbf{w}^{k})-F_{i^{k}}\left(\mathbf{w}^*\right)\right) \\
			& \qquad+\frac{(\eta^{k})^2}{2}\|\nabla F_{i^{k}}(\mathbf{w}^{k})\|^2+\frac{(\eta^{k})^2}{2} \hat{\gamma}^2 D^2.
		\end{aligned}
	\end{equation}
	
	By re-arranging (\ref{eq_22}), we come to
	\begin{equation}
		\label{eq_11}
		\begin{aligned}
			&\eta^{k}\left(F_{i^{k}}(\mathbf{w}^{k})-F_{i^{k}}\left(\mathbf{w}^*\right)\right)  \\
			& \quad \leq \frac{1}{2}\left(\|\mathbf{w}^{k}-\mathbf{w}^*\|^2-\|\mathbf{w}^{k+1}-\mathbf{w}^*\|^2\right) \\
			& \qquad+\frac{(\eta^{k})^2}{4}\|\nabla F_{i^{k}}(\mathbf{w}^{k})\|^2+\frac{(\eta^{k})^2}{4} \hat{\gamma}^2 D^2.
		\end{aligned}
	\end{equation}
	
	Summing the above equation for $k$ and using the boundedness of $\mathcal{W}$,
	\begin{equation}
		\label{eq_12}
		\begin{aligned}
			& \sum_{k} \eta^{k}\left(F_{i^{k}}(\mathbf{w}^{k})-F_{i^{k}}\left(\mathbf{w}^*\right)\right) \leq \frac{1}{2}\|\mathbf{w}^{0}-\mathbf{w}^*\|^2 \\
			& \quad+\sum_k \frac{(\eta^{k})^2}{4}\|\nabla F_{i^{k}}(\mathbf{w}^{k})\|^2+\sum_k \frac{(\eta^{k})^2}{4} \hat{\gamma}^2 D^2.
		\end{aligned}
	\end{equation}
	
	Using Lemma \ref{lemma_10}, we get \cite{25Ayache}
	\begin{equation}
		\label{eq_13}
		\max _i \|\boldsymbol{\Pi}^*-\mathbf{P}^{\tau^{k}}(i,:)\| \leq \zeta \lambda_P^{\tau^{k}} \leq \frac{1}{2 k},
	\end{equation}
	for $\tau^{k}=\min \{k, \max \{ \lceil \frac{\ln (2 \zeta k)}{\ln \left(1 / \lambda_P\right)} \rceil ,K_\mathbf{P}\}\}$.
	
	By using (\ref{eq_12}) and Lemma \ref{lemma_5}, we have
	\begin{equation}
		\label{eq_14}
		\begin{aligned}
			& \sum_{k} \eta^{k} \mathbb{E}\left[F_{i^{k}}(\mathbf{w}^{k-\tau^{k}})-F_{i^{k}}\left(\mathbf{w}^*\right)\right] \\
			& \leq \frac{D^2}{2} \Big[\sum_{k} \tau^{k}(\eta^{k})^2+ \sum_{k} \sum_{n={k}-\tau^{k}}^{{k}-1}(\eta^{n})^2 \Big] +\frac{1}{2}\|\mathbf{w}^{0}-\mathbf{w}^*\|^2\\ 
			& \quad  +\sum_{k} \frac{(\eta^{k})^2}{4} \mathbb{E}\|\nabla F_{i^{k}}(\mathbf{w}^{k})\|^2 +\sum_{k} \frac{(\eta^{k})^2}{4} \hat{\gamma}^2 D^2.
		\end{aligned}
	\end{equation}
	
	We introduce the statistical heterogeneity of data. According to (\ref{eq_14}), Definition \ref{definition_1}, and Corollary \ref{corollary_1}, it follows that
	\begin{equation}
		\label{eq_15}
		\begin{aligned}
			&\sum_{k} \eta^{k} \mathbb{E}\left[F_{i^{k}}(\mathbf{w}^{k-\tau^{k}})-F_{i^{k}}\left(\mathbf{w}^*\right)\right] \\
			& \quad \leq \frac{D^2}{2} \sum_{k} \tau^{k}(\eta^{k})^2+\frac{D^2}{2} \sum_{k} \sum_{n={k}-\tau^{k}}^{{k}-1}(\eta^{n})^2\\
			& \qquad +\frac{1}{2}\|\mathbf{w}^{0}-\mathbf{w}^*\|^2 +\frac{(\delta^2+\hat{\gamma}^2) D^2}{4} \sum_{k}(\eta^{k})^2.
		\end{aligned}
	\end{equation}
	
	According to the step size summability in Assumption \ref{assumption_3}, the results are as follows \cite{32Sun}
	\begin{equation}
		\label{eq_16}
		\begin{aligned}
			\sum_{{k}=\bar{K}}^{\infty} \sum_{n=k-\tau^{k}}^{k-1}\left(\eta^{n}\right)^2 &\leq \sum_{{k}=\bar{K}}^{\infty} \tau^{k}(\eta^{k})^2\\ 
			&\leq \frac{2}{\ln \left(1 / \lambda_P\right)} \sum_{{k}=\bar{K}}^{\infty} \ln {k} \cdot(\eta^{k})^2<\infty.
		\end{aligned}
	\end{equation}
	where $\bar{K}$ is a large integer. It follows from (\ref{eq_15}), (\ref{eq_16}) and Lemma \ref{lemma_6} that
	\begin{equation}
		\label{eq_16-1}
		\begin{aligned}
			& \sum_{k} \eta^{k} \mathbb{E}\left[f(\mathbf{w}^{k-\tau^{k}})-f(\mathbf{w}^*)\right] \\
			& \leq \frac{2 D^2}{\ln \left(1 / \lambda_P\right)} \sum_{k=\bar{K}}^{\infty} \ln k \cdot(\eta^{k})^2+\frac{1}{2}\|\mathbf{w}^{0}-\mathbf{w}^*\|^2 \\
			& \quad+\frac{(\delta^2+\hat{\gamma}^2) D^2}{4} \sum_{k}(\eta^{k})^2+\sum_{k} \frac{n \eta^{k}}{2 k}
		\end{aligned}
	\end{equation}
	
	Finally, combining (\ref{eq_16-1}) and Lemma \ref{lemma_7},
	\begin{equation}
		\label{eq_17}
		\begin{aligned}
			& \sum_{k} \eta^{k} \mathbb{E}\left[f(\mathbf{w}^{k})-f\left(\mathbf{w}^*\right)\right] \\
			& \quad \leq \frac{(1+n) \cdot 2 D^2}{\ln \left(1 / \lambda_P\right)} \sum_{k=\bar{K}}^{\infty} \ln k \cdot(\eta^{k})^2+\frac{1}{2}\|\mathbf{w}^{0}-\mathbf{w}^*\|^2 \\
			& \qquad +\frac{(\delta^2+\hat{\gamma}^2) D^2}{4} \sum_{k}(\eta^{k})^2+\sum_{k} \frac{n \eta^{k}}{2 k}.
		\end{aligned}
	\end{equation}
	
	By using the convexity of $f$ and Jensen’s inequality, we have
	\begin{equation}
		\label{eq_18}
		\begin{aligned}
			& \sum_{i=1}^k \eta^{i} \cdot \mathbb{E}\left[f(\overline{\mathbf{w}}^{k})-f\left(\mathbf{w}^*\right)\right] \leq \sum_{i=1}^k \eta^{i} \mathbb{E}\left[f\left(\mathbf{w}^{i}\right)-f\left(\mathbf{w}^*\right)\right].
		\end{aligned}
	\end{equation}
	
	Then
	\begin{equation}
		\label{eq_19}
		\begin{aligned}
			& \mathbb{E} {\left[f(\overline{\mathbf{w}}^{k})-f\left(\mathbf{w}^*\right)\right] } \\
			& \quad\leq \frac{(1+n) \cdot 2 D^2 \sum_{k=\bar{K}}^{\infty} \ln k \cdot(\eta^{k})^2}{\ln \left(1 / \lambda_P\right) \sum_{k} \eta^{k}}+\frac{\|\mathbf{w}^{0}-\mathbf{w}^*\|^2}{2 \sum_{k} \eta^{k}} \\
			& \qquad +\frac{(\delta^2+\hat{\gamma}^2) D^2}{4 \sum_{k} \eta^{k}} \sum_{k}(\eta^{k})^2+\frac{\sum_{k} \frac{n \eta^{k}}{2 k}}{\sum_{k} \eta^{k}}.
		\end{aligned}
	\end{equation}
	
	We have completed the proof of Theorem \ref{theorem_1}.
\end{proof}

\subsection{Deferred Proofs of Key Lemmas}
\begin{proof}[Proof of Lemma \ref{lemma_1}]
	In accordance with Definition \ref{definition_2}, if only $K=1$ local update is performed in the $t$-th communication round, 
	\begin{equation}
		\label{eq_lemma1}
		\|\nabla F_i(\mathbf{w}^{k})\| \leq \gamma_i^{k-1}\|\nabla F_i(\mathbf{w}^{k-1})\|.
	\end{equation}
	DFedRW distributes $K>1$ local updates across different devices, executing (\ref{eq_lemma1}) on $i^{k-K}, i^{k-K+1}, \dots, i^{k-1}$ respectively. We recursively apply the above inequality and obtain
	\begin{equation}
		\|\nabla F_{i^k}(\mathbf{w}^k)\| \leq \gamma_{i^{k-1}}^{k-1}\gamma_{i^{k-2}}^{k-2}\cdots\gamma_{i^{k-K}}^{k-K} \|\nabla F_{i^{k-K}}(\mathbf{w}^{k-K})\|.
	\end{equation}
\end{proof}

\begin{proof}[Proof of Lemma \ref{lemma_5}]
	Based on Lemma \ref{lemma_3} and Corollary \ref{corollary_1}, and using the Cauchy-Schwarz inequality, we obtain 
	\begin{equation}
		\label{eq_23}
		\begin{aligned}
			& \eta^{k} \mathbb{E}\left[F_{i^{k}}(\mathbf{w}^{k-\tau^{k}})-F_{i^{k}}(\mathbf{w}^{k})\right]  \leq D \eta^{k} \mathbb{E}\|\mathbf{w}^{k-\tau^{k}}-\mathbf{w}^{k}\| \\
			& \quad \leq D \eta^{k} \sum_{n=k-\tau^{k}}^{k-1} \mathbb{E}(\|\mathbf{w}^{n+1}-\mathbf{w}^{n}\|) \\
			& \quad \leq D^2 \eta^{k} \sum_{n=k-\tau^{k}}^{k-1} \eta^{n}  \leq \frac{D^2}{2} \sum_{n=k-\tau^{k}}^{k-1}\left((\eta^{n})^2+(\eta^{k})^2\right) \\
			& \quad \leq \frac{D^2}{2} \tau^{k}(\eta^{k})^2+\frac{D^2}{2} \sum_{n=k-\tau^k}^{k-1}\left(\eta^{n}\right)^2.
		\end{aligned}
	\end{equation}
	Summing the above equation over $k$, we obtain
	\begin{equation}
		\label{eq_24}
		\begin{aligned}
			& \sum_{k} \eta^{k} \mathbb{E}\left[F_{i^{k}}(\mathbf{w}^{k-\tau^{k}})-F_{i^{k}}(\mathbf{w}^{k})\right] \\
			& \quad \leq \sum_{k} \frac{D^2}{2} \tau^{k}(\eta^{k})^2+\frac{D^2}{2} \sum_{k} \sum_{n=k-\tau^{k}}^{k-1}\left(\eta^{n}\right)^2.
		\end{aligned}
	\end{equation}
\end{proof}

\begin{proof}[Proof of Lemma \ref{lemma_6}]
	By using the Markov property and directly calculating from Definition \ref{definition_4} and Lemma \ref{lemma_10}, we have
	\begin{equation}
		\label{eq_25}
		\displaybreak[4]
		\begin{aligned}
			& \mathbb{E}_{i^{k}}\left[F_{i^{k}}(\mathbf{w}^{k-\tau^{k}})-F_{i^{k}}\left(\mathbf{w}^*\right) \mid  X_0, X_1, \ldots, X_{k-\tau^{k}}\right] \\
			& = \sum_{i=1}^n\left(F_i(\mathbf{w}^{k-\tau^{k}})-F_i(\mathbf{w}^*)\right)\\
			& \qquad \times \mathbb{P}\left(i^{k}=i \mid X_0, X_1, \ldots, X_{k-\tau^{k}}\right) \\
			& = \sum_{i=1}^n\left(F_i(\mathbf{w}^{k-\tau^{k}})-F_i\left(\mathbf{w}^*\right)\right) \times \mathbb{P}\left(i^{k}=i \mid X_{k-\tau^{k}}\right) \\
			& =\sum_{i=1}^n\left(F_i(\mathbf{w}^{k-\tau^{k}})-F_i\left(\mathbf{w}^*\right)\right) \times [\mathbf{P}^{\tau^{k}}]_{i^{k}-\tau^{k},i} \\
			& \geq \left(f(\mathbf{w}^{k-\tau^{k}})-f\left(\mathbf{w}^*\right)\right)-\frac{n}{2 k}.
		\end{aligned}
	\end{equation}
	Therefore,
	\begin{equation}
		\label{eq_26}
		\begin{aligned}
			& \sum_{k} \eta^{k} \mathbb{E}\left[f(\mathbf{w}^{k-\tau^{k}})-f\left(\mathbf{w}^*\right)\right] \\
			& \quad \leq \sum_{k} \eta^{k} \mathbb{E}\left[F_{i^k}(\mathbf{w}^{k-\tau^{k}})-F_{i^{k}}\left(\mathbf{w}^*\right)\right] +\sum_{k} \frac{N \eta^{k}}{2 k}.
		\end{aligned}
	\end{equation}
\end{proof}

\begin{proof}[Proof of Lemma \ref{lemma_7}]
	From Lemma \ref{lemma_3}, Corollary \ref{corollary_1}, and the summability of $k$, we obtain
	
	\begin{equation}
		\label{eq_27}
		\begin{aligned}
			& \eta^{k} \mathbb{E}\left[f(\mathbf{w}^{k})-f(\mathbf{w}^{k-\tau^{k}})\right] \leq n D \eta^{k} \mathbb{E}\|\mathbf{w}^{k}-\mathbf{w}^{k-\tau^{k}}\| \\
			& \quad \leq n D \eta^{k} \sum_{n=k-\tau^{k}}^{k-1} \mathbb{E}(\|\mathbf{w}^{n+1}-\mathbf{w}^{n}\|) \\
			& \quad \leq n D^2 \eta^{k} \sum_{n=k-\tau^{k}}^{k-1} \eta^{n} \leq \frac{N D^2}{2} \sum_{n=k-\tau^{k}}^{k-1}\left(\left(\eta^{n}\right)^2+(\eta^{k})^2\right) \\
			& \quad \leq \frac{n D^2}{2} \tau^{k}(\eta^{k})^2+\frac{n D^2}{2} \sum_{n=k-\tau^{k}}^{k-1}\left(\eta^{n}\right)^2 .
		\end{aligned}
	\end{equation}
	
	Namely,
	\begin{equation}
		\label{eq_28}
		\begin{aligned}
			&\sum_{k} \eta^{k} \mathbb{E}\left[f(\mathbf{w}^{k})-f(\mathbf{w}^{k-\tau^{k}})\right] \leq \frac{2 n D^2}{\ln (1 / \lambda_p)} \sum_{k=\bar{K}}^{\infty} \ln k \cdot (\eta^{k})^2.
		\end{aligned}
	\end{equation}
\end{proof}

\section{Proof of Theorem 2: Convergence Bound of quantized DFedRW}
First, we prove Lemma \ref{lemma_2}, and then prove Theorem 2.

\begin{proof}[Proof of Lemma \ref{lemma_2}]
	Let $\varpi=\frac{|w^\nu|}{\|\mathbf{w}\|}-s \ell$, where $s \ell \leq \frac{|w^\nu|}{\|\mathbf{w}\|}<s(\ell+1)$, $0 \leq \varpi<s$. To obtain the bound on the quantized variance $\mathbb{E}[\|Q(\mathbf{w})-\mathbf{w}\|^2]$, we first calculate the expectation $\mathbb{E}[Q(w^\nu)-w^\nu]$ of the parameter difference. We have
	
	\begin{equation}
		\begin{aligned}
			& \mathbb{E}\left[Q(w^\nu)-w^\nu\right]\\
			& \quad =  \left(s\ell-\frac{|w^\nu|}{\|\mathbf{w}\|}\right)\|\mathbf{w}\|\operatorname{sgn}(w^\nu)\left(1-\frac{|w^\nu|}{s\|\mathbf{w}\|}+\ell\right) \\
			& \qquad + \left(s(\ell+1)-\frac{|w^\nu|}{\|\mathbf{w}\|}\right)\|\mathbf{w}\|\operatorname{sgn}(w^\nu)\left(\frac{|w^\nu|}{s\|\mathbf{w}\|}-\ell\right)\\
			& \quad =\|\mathbf{w}\|\operatorname{sgn}(w^\nu)\left[-\varpi\left(1-\frac{\varpi}{s}\right)+(s-\varpi) \frac{\varpi}{s}\right]=0,
		\end{aligned}
	\end{equation}
	and
	\begin{equation}
		\begin{aligned}
			\mathbb{E}\left[(Q(w^\nu)-w^\nu)^2\right] & =\|\mathbf{w}\|^2\left[\varpi^2\left(1-\frac{\varpi}{s}\right)+(s-\varpi)^2 \frac{\varpi}{s}\right]\\
			& =\|\mathbf{w}\|^2\varpi(s-\varpi).
		\end{aligned}
	\end{equation}
	
	Due to $0 \leq \varpi<s$, there is $\varpi(s-\varpi) \in\left[0, \frac{s^2}{4}\right]$. Assuming a bounded squared parameter norm $\|\mathbf{w}\|^2 \leq \sigma^2$, we have
	
	\begin{equation}
		\begin{aligned}
			& \mathbb{E}\left[\|Q(\mathbf{w})-\mathbf{w}\|^2\right]=\mathbb{E} \sum_{\nu=1}^d\left(Q\left(w^\nu\right)-w^\nu\right)^2 \\
			& \quad =\sum_{\nu=1}^d \mathbb{E}\left[(Q\left(w^\nu\right)-w^\nu)^2\right] =\sum_{\nu=1}^d \|\mathbf{w}\|^2\varpi(s-\varpi) \\
			& \quad \leq \|\mathbf{w}\|^2 \sum_{i=1}^d \frac{s^2}{4}=\frac{\sigma^2 d s^2}{4}.
		\end{aligned}
	\end{equation}
\end{proof}

To prove Theorem 2, we also need to have following lemmas. In the quantized DFedRW, the difference in parameters $Q(\mathbf{w}^{k+1}-\mathbf{w}^{k})$ transmitted between devices.

\begin{lemma}
	\label{lemma_9}
	Under Assumption \ref{assumption_3}, by using Lemmas \ref{lemma_2}, \ref{lemma_3} and Corollary \ref{corollary_1}, we have
	\begin{equation}
		\begin{aligned}
			& \eta^{k} \mathbb{E}\left[f(\mathbf{w}^{k})-f(\mathbf{w}^{k-\tau^{k}})\right] \\
			& \quad \leq \frac{2 n D^2}{\ln \left(1 / \lambda_p\right)} \sum_{k=\bar{K}}^{\infty} \ln k \cdot(\eta^{k})^2+n D \eta^{k} \tau^{k} \frac{\sigma \sqrt{d} s}{2}.
		\end{aligned}
	\end{equation}
\end{lemma}

\begin{proof}[Proof of Theorem 2]
	Using Lemma \ref{lemma_3} and Corollary \ref{corollary_1}, since $\mathbf{w}^{k+1}-\mathbf{w}^{k}\leftarrow Q(\mathbf{w}^{k+1}-\mathbf{w}^{k})$, we have 
	\begin{equation}
		\label{eq_34}
		\begin{aligned}
			& \eta^{k} \mathbb{E}\left[F_{i^{k}}(\mathbf{w}^{k-\tau^{k}})-F_{i^{k}}(\mathbf{w}^{k})\right] \\
			& \quad \leq D \eta^{k} \sum_{n=k-\tau^{k}}^{k-1} \mathbb{E}\left(\|Q(\mathbf{w}^{n+1}-\mathbf{w}^{n})\right. \\
			& \qquad \left. +(\mathbf{w}^{n+1}-\mathbf{w}^{n})-(\mathbf{w}^{n+1}-\mathbf{w}^{n})\|\right) \\
			& \quad \leq D \eta^{k} \sum_{n=k-\tau^{k}}^{k-1} \mathbb{E}\left(\|\mathbf{w}^{n+1}-\mathbf{w}^{n}\| \right. \\
			& \qquad \left. +\|Q(\mathbf{w}^{n+1}-\mathbf{w}^{n})-(\mathbf{w}^{n+1}-\mathbf{w}^{n})\|\right) \\
			& \quad \leq D \eta^{k} \sum_{n=k-\tau^{k}}^{k-1} \left(\mathbb{E}\|\mathbf{w}^{n+1}-\mathbf{w}^{n}\| \right. \\
			& \qquad \left. +\mathbb{E}\|Q(\mathbf{w}^{n+1}-\mathbf{w}^{n})-(\mathbf{w}^{n+1}-\mathbf{w}^{n})\|\right).
		\end{aligned}
	\end{equation}
	
	Using the Cauchy Schwarz inequality $(\mathbb{E}[X Y])^2 \leq$ $\mathbb{E}[X^2] \cdot \mathbb{E}[Y^2]$ and Lemma \ref{lemma_2}. Let $X=\left\|Q(\mathbf{w}^{n+1}-\mathbf{w}^{n})-(\mathbf{w}^{n+1}-\mathbf{w}^{n})\right\|$ and $Y=1$. Then, we have
	\begin{equation}
		\label{eq_35}
		\begin{aligned}
			& \mathbb{E}\left\|Q(\mathbf{w}^{n+1}-\mathbf{w}^{n})-(\mathbf{w}^{n+1}-\mathbf{w}^{n})\right\| \\
			& \quad \leq \sqrt{\mathbb{E}\left\|Q(\mathbf{w}^{n+1}-\mathbf{w}^{n})-(\mathbf{w}^{n+1}-\mathbf{w}^{n})\right\|^2} \\
			& \quad =\sqrt{\frac{\sigma^2 d s^2}{4}}=\frac{\sigma \sqrt{d} s}{2}.
		\end{aligned}
	\end{equation}
	
	Then
	\begin{equation}
		\label{eq_36}
		\begin{aligned}
			& \eta^{k} \mathbb{E}\left[F_{i^{k}}(\mathbf{w}^{k-\tau^{k}})-F_{i^{k}}(\mathbf{w}^{k})\right] \\
			& \leq \frac{D^2}{2} \tau^{k}(\eta^{k})^2+\frac{D^2}{2} \sum_{n=k-\tau^{k}}^{k-1}\left(\eta^{n}\right)^2 +D \eta^{k} \tau^{k} \frac{\sigma \sqrt{d} s}{2} .
		\end{aligned}
	\end{equation}

	Summing the results of (\ref{eq_36}), we have
	
	\begin{equation}
		\label{eq_29}
		\begin{aligned}
			& \sum_{k} \eta^{k} \mathbb{E}\left[F_{i^{k}}(\mathbf{w}^{k-\tau^{k}})-F_{i^{k}}(\mathbf{w}^{k})\right] \\
			& \quad \leq \frac{D^2}{2} \sum_{k} \tau^{k}(\eta^{k})^2+\frac{D^2}{2} \sum_{k} \sum_{n=k-\tau^{k}}^{k-1}\left(\eta^{n}\right)^2 \\
			& \qquad + \frac{D \sigma \sqrt{d} s}{2} \sum_{k} \eta^{k} \tau^{k}.
		\end{aligned}
	\end{equation}
	
	Derived from (\ref{eq_29}) and (\ref{eq_12}), we further have
	\begin{equation}
		\label{eq_30}
		\begin{aligned}
			& \sum_{k} \eta^{k} \mathbb{E}\left[F_{i^{k}}(\mathbf{w}^{k-\tau^{k}})-F_{i^{k}}(\mathbf{w}^{*})\right] \\
			& \quad \leq \frac{D^2}{2} \sum_{k} \tau^{k}(\eta^{k})^2+\frac{D^2}{2} \sum_{k} \sum_{n=k-\tau^{k}}^{k-1}\left(\eta^{n}\right)^2 \\
			& \qquad +\frac{D \sigma \sqrt{d} s}{2} \sum_{k} \eta^{k} \tau^{k} +\frac{1}{2}\|\mathbf{w}^{0}-\mathbf{w}^*\|^2 \\
			& \qquad +\sum_{k} \frac{(\eta^{k})^2}{4}\|\nabla f_{i^{k}}(\mathbf{w}^{k})\|^2 + \sum_{k} \frac{(\eta^{k})^2}{4} \hat{\gamma}^2 D^2.
		\end{aligned}
	\end{equation}
	
	It follows from (\ref{eq_30}), (\ref{eq_16}) and lemma \ref{lemma_6} that
	\begin{equation}
		\label{eq_31}
		\begin{aligned}
			& \sum_{k} \eta^{k} \mathbb{E}\left[f(\mathbf{w}^{k-\tau^{k}})-f\left(\mathbf{w}^*\right)\right] \\
			& \quad \leq \frac{2 D^2}{\ln \left(1 / \lambda_P\right)} \sum_{k=\bar{K}}^{\infty} \ln k \cdot (\eta^{k})^2+\frac{D \sigma \sqrt{d} s}{2} \sum_{k} \eta^{k} \tau^{k}  \\
			& \qquad+\frac{1}{2}\|\mathbf{w}^{0}-\mathbf{w}^*\|^2+\sum_{k} \frac{(\eta^{k})^2}{4}\|\nabla f_{i^{k}}(\mathbf{w}^{k})\|^2 \\
			& \qquad+\sum_{k} \frac{(\eta^{k})^2}{4} \hat{\gamma}^2 D^2+\sum_{k} \frac{n \eta^{k}}{2 k}.
		\end{aligned}
	\end{equation}
	
	Finally, using Lemmas \ref{lemma_5}, \ref{lemma_9}, and (\ref{eq_31}), we know that
	\begin{equation}
		\label{eq_32}
		\begin{aligned}
			& \mathbb{E} {\left[f(\overline{\mathbf{w}}^{k})-f\left(\mathbf{w}^*\right)\right] } \\
			& \quad\leq \frac{(1+n) \cdot 2 D^2 \sum_{k=\bar{K}}^{\infty} \ln k \cdot(\eta^{k})^2}{\ln \left(1 / \lambda_P\right) \sum_{k} \eta^{k}}+\frac{\|\mathbf{w}^{0}-\mathbf{w}^*\|^2}{2 \sum_{k} \eta^{k}} \\
			& \qquad +\frac{(\delta^2+\hat{\gamma}^2) D^2}{4 \sum_{k} \eta^{k}} \sum_{k}(\eta^{k})^2+\frac{\sum_{k} \frac{N \eta^{k}}{2 k}}{\sum_{k} \eta^{k}}\\
			& \qquad +\frac{(1+n) \frac{D \sigma \sqrt{d} s}{2} \sum_{k} \eta^{k} \tau^{k}}{\sum_{k} \eta^{k}}.
		\end{aligned}
	\end{equation}
	
	We have completed the proof of Theorem \ref{theorem_2}.
\end{proof}

\begin{proof}[Proof of Lemma \ref{lemma_9}] 
	Using Lemma \ref{lemma_3}, Corollary \ref{corollary_1}, (\ref{eq_16}) and (\ref{eq_35}), we have
	\begin{equation}
		\label{eq_37}
		\begin{aligned}
			& \eta^{k} \mathbb{E}\left[f(\mathbf{w}^{k})-f(\mathbf{w}^{k-\tau^{k}})\right] \\
			& \quad \leq n D \eta^{k} \sum_{n=k-\tau^{k}}^{k-1}\left(\mathbb{E}\|\mathbf{w}^{n+1}-\mathbf{w}^{n}\|\right. \\
			& \qquad \left.+\mathbb{E}\|Q(\mathbf{w}^{n+1}-\mathbf{w}^{n})-(\mathbf{w}^{n+1}-\mathbf{w}^{n})\|\right) \\
			& \quad \leq \frac{2 n D^2}{\ln (1 / \lambda_p)} \sum_{k=\bar{K}}^{\infty} \ln k \cdot(\eta^{k})^2+n D \eta^{k} \tau^{k} \frac{\sigma \sqrt{d} s}{2}.
		\end{aligned}
	\end{equation}
\end{proof}

\section{Proof of Proposition 1: Sufficient Conditions for Reducing Communication}
\begin{proof}
	Comparing Theorems \ref{theorem_1} and \ref{theorem_2}, we need to ensure that the additional term $\frac{\psi(d, s)}{k^{1-q}}$ in Theorem 2 does not significantly affect the convergence bound under error $\epsilon$, as given by
	\begin{equation}
		\label{eq_38}
		\frac{(1+n) \frac{D \sigma \sqrt{d} s}{2} \sum_{k} \eta^{k} \tau^{k}}{k^{1-q}}<\epsilon.
	\end{equation}
	
	Considering the decreasing learning rate $\eta^{k}=\mathcal{O}(\frac{1}{k^q})$, $\frac{1}{2}<q<1$, and the convergence bound $T_q=\varrho T_{n q}$, $\varrho>1$, we redefine $\psi(d, s)$ as 
	\begin{equation}
		\label{eq_39}
		\psi(d, s) = \Upsilon \sum_{k=1}^{T_q} \frac{1}{k^q} \cdot \min \left\{k, \max \left\{ \left\lceil \frac{\ln (2 \zeta k)}{\ln (1 / \lambda_P)} \right\rceil , K_{\mathbf{P}} \right\} \right\},
	\end{equation}	
	where $\Upsilon =(1+n) \frac{D \sigma \sqrt{d} s}{2}$. We can decompose (\ref{eq_39}) into two parts:
	\begin{equation}
		\label{eq_40}
		\begin{aligned}
			\psi(d, s) \approx & \Upsilon\left(\sum_{k=1}^{K_\mathbf{P}} k^{1-q}+\sum_{k=K_\mathbf{P}+1}^{T^{\prime}} k^{1-q} \right. \\
			&  \left. \quad +\sum_{k=T^{\prime}+1}^{T_q} \frac{\ln (2 \zeta k)}{k^q \ln (1 / \lambda_P)}\right),
		\end{aligned}
	\end{equation}
	where $T'$ is the intersection point with $k = \frac{\ln(2 \zeta k)}{\ln (1 / \lambda_P)}$. This means that for $k \leq K_\mathbf{P}$, $k$ directly influences the selection of the minimum value. For $K_\mathbf{P}<k \leq T^{\prime}$, the maximum value is constrained by $\frac{\ln (2 \zeta k)}{\ln (1 / \lambda_P)}$. For $k>T^{\prime}$, $\frac{\ln(2 \zeta k)}{\ln (1 / \lambda_P)}$ governs the minimum.
	
	Since the growth rate of $k$ is faster than $\frac{\ln(2 \zeta k)}{\ln (1 / \lambda_P)}$, assuming $k = T' = a \frac{\ln(2 \zeta T')}{\ln (1 / \lambda_P)}$, we can approximate the solution as $T' \approx \frac{\ln (2 \zeta T')}{\ln (1 / \lambda_P)}$. If $k \leq T'$,
	\begin{equation}
		\label{eq_41}
		\sum_{k=1}^{T^{\prime}} \frac{k}{k^q}=\sum_{k=1}^{T^{\prime}} k^{1-q} \approx \int_1^{T^{\prime}} x^{1-q} d x \approx\frac{(T^{\prime})^{2-q}-1}{2-q}.
	\end{equation}
	
	When $k > T'$,
	\begin{equation}
		\label{eq_42}
		\begin{aligned}
			\sum_{k=T^{\prime}+1}^{T_q} \frac{\ln (2 \zeta k)}{k^q \cdot \ln (1 / \lambda_P)} & \approx  \frac{\ln (2 \zeta T)}{\ln (1 / \lambda_P)} \sum_{k=T^{\prime}+1}^{T_q} \frac{1}{k^q} \\
			& \approx \frac{\ln (2 \zeta T)}{\ln (1 / \lambda_P)} T_q^{1-q}.
		\end{aligned}
	\end{equation}
	
	From (\ref{eq_38})$-$(\ref{eq_42}), we further have
	\begin{equation}
		\label{eq_43}
		\begin{aligned}
			& (1+n) \frac{D \sigma \sqrt{d} s}{2} \left(\frac{(T^{\prime})^{2-q}-1}{(2-q) T_q^{1-q}}\right. \\
			& \quad \left.+\frac{T_q^{1-q}\left[(1-q) \ln \left(2 \zeta T_q\right)-1\right]+1}{(1-q)^2 \ln (1 / \lambda_P) T_q^{1-q}}\right)<\epsilon
		\end{aligned}
	\end{equation}
	
	When $T_q$ is large and $q \rightarrow 1$, assuming $\ln (2 \zeta T_q) \approx \ln (T_q)$ and ignoring the logarithmic growth term of $T'$. We apply the AM-GM inequality, and then (\ref{eq_43}) is simplified to
	\begin{equation}
		\label{eq_44}
		\begin{aligned}
			\Upsilon \left(\frac{\ln (T_q)}{T_q^{1-q}}+\frac{1}{\ln (1 / \lambda_P)}\right)=2 \Upsilon \sqrt{\frac{1}{\ln (1 / \lambda_P)}}<\epsilon.
		\end{aligned}
	\end{equation}
	
	To achieve the desired error $\epsilon$, the communication costs of DFedRW and its quantized version are $32 d(|\mathcal{N}_c(i)|+K-1) T_{n q}$ bits and $(64+b d)(|\mathcal{N}_c(i)|+K-1) T_q$ bits, respectively. Therefore, if $32 d>\varrho(64+b d)$, the quantization reduces the communication overhead of DFedRW.
\end{proof}
\end{appendices}

\bibliographystyle{IEEEtran}
\bibliography{IEEEabrv,ref}

\begin{thebibliography}{10}
\providecommand{\url}[1]{#1}
\csname url@samestyle\endcsname
\providecommand{\newblock}{\relax}
\providecommand{\bibinfo}[2]{#2}
\providecommand{\BIBentrySTDinterwordspacing}{\spaceskip=0pt\relax}
\providecommand{\BIBentryALTinterwordstretchfactor}{4}
\providecommand{\BIBentryALTinterwordspacing}{\spaceskip=\fontdimen2\font plus
\BIBentryALTinterwordstretchfactor\fontdimen3\font minus
  \fontdimen4\font\relax}
\providecommand{\BIBforeignlanguage}[2]{{%
\expandafter\ifx\csname l@#1\endcsname\relax
\typeout{** WARNING: IEEEtran.bst: No hyphenation pattern has been}%
\typeout{** loaded for the language `#1'. Using the pattern for}%
\typeout{** the default language instead.}%
\else
\language=\csname l@#1\endcsname
\fi
#2}}
\providecommand{\BIBdecl}{\relax}
\BIBdecl

\bibitem{1Peter}
P.~Kairouz, H.~B. McMahan, B.~Avent, A.~Bellet, M.~Bennis, A.~N. Bhagoji
  \emph{et~al.}, ``Advances and open problems in federated learning,''
  \emph{Found. Trends Mach. Learn.}, vol.~14, no. 1–2, pp. 1--210, 2021.

\bibitem{2Rieke}
N.~Rieke, J.~Hancox, W.~Li, F.~Milletari, H.~R. Roth, S.~Albarqouni
  \emph{et~al.}, ``The future of digital health with federated learning,''
  \emph{npj Digit. Med.}, vol.~3, p. 119, 2020.

\bibitem{3Wu}
Q.~Wu, X.~Chen, Z.~Zhou, and J.~Zhang, ``Fedhome: Cloud-edge based personalized
  federated learning for in-home health monitoring,'' \emph{IEEE Trans. Mob.
  Comput.}, vol.~21, no.~8, pp. 2818--2832, 2022.

\bibitem{4Zeng}
T.~Zeng, O.~Semiari, W.~Saad, and M.~Bennis, ``Wireless-enabled asynchronous
  federated fourier neural network for turbulence prediction in urban air
  mobility ({UAM}),'' \emph{IEEE Trans. Wirel. Commun.}, vol.~22, no.~11, pp.
  7902--7916, 2023.

\bibitem{5Zheng}
S.~Zheng, Y.~Cao, M.~Yoshikawa, H.~Li, and Q.~Yan, ``Fl-market: Trading private
  models in federated learning,'' in \emph{Proc. 2022 IEEE Int. Conf. Big Data
  (Big Data)}, 2022.

\bibitem{6Li}
T.~Li, A.~K. Sahu, A.~Talwalkar, and V.~Smith, ``Federated learning:
  Challenges, methods, and future directions,'' \emph{IEEE Signal Process.
  Mag.}, vol.~37, no.~3, pp. 50--60, 2020.

\bibitem{7Korkmaz}
C.~Korkmaz, H.~E. Kocas, A.~Uysal, A.~Masry, O.~Ozkasap, and B.~Akgun, ``Chain
  {FL}: Decentralized federated machine learning via blockchain,'' in
  \emph{Proc. 2nd Int. Conf. Blockchain Comput. Appl. (BCCA)}, 2020, pp.
  140--146.

\bibitem{7Imteaj}
A.~Imteaj, U.~Thakker, S.~Wang, J.~Li, and M.~H. Amini, ``A survey on federated
  learning for resource-constrained {I}o{T} devices,'' \emph{IEEE Internet
  Things J.}, vol.~9, no.~1, pp. 1--24, 2022.

\bibitem{7Jianyu}
J.~Wang, Q.~Liu, H.~Liang, G.~Joshi, and H.~V. Poor, ``Tackling the objective
  inconsistency problem in heterogeneous federated optimization,'' in
  \emph{Proc. Adv. Neural Inf. Process. Syst. (NeurIPS)}, 2020.

\bibitem{8Mao}
X.~Mao, K.~Yuan, Y.~Hu, Y.~Gu, A.~H. Sayed, and W.~Yin, ``Walkman: A
  communication-efficient random-walk algorithm for decentralized
  optimization,'' \emph{IEEE Trans. Signal Process.}, vol.~68, pp. 2513--2528,
  2020.

\bibitem{9Sun}
T.~Sun, D.~Li, and B.~Wang, ``On the decentralized stochastic gradient descent
  with markov chain sampling,'' \emph{IEEE Trans. Signal Process.}, vol.~71,
  pp. 2895--2909, 2023.

\bibitem{30Liao}
Y.~Liao, Y.~Xu, H.~Xu, L.~Wang, C.~Qian, and C.~Qiao, ``Decentralized federated
  learning with adaptive configuration for heterogeneous participants,''
  \emph{IEEE Trans. Mob. Comput.}, vol.~23, no.~6, pp. 7453--7469, 2024.

\bibitem{15Zhou}
X.~Zhou, L.~Chang, and J.~Cao, ``Communication-efficient nonconvex federated
  learning with error feedback for uplink and downlink,'' \emph{IEEE Trans.
  Neural Netw. Learn. Syst.}, pp. 1--12, 2023.

\bibitem{16Tang}
Z.~Tang, S.~Shi, B.~Li, and X.~Chu, ``Gossip{FL}: A decentralized federated
  learning framework with sparsified and adaptive communication,'' \emph{IEEE
  Trans. Parallel Distrib. Syst.}, vol.~34, no.~3, pp. 909--922, 2023.

\bibitem{14Sun}
T.~Sun, D.~Li, and B.~Wang, ``Decentralized federated averaging,'' \emph{IEEE
  Trans. Pattern Anal. Mach. Intell.}, vol.~45, no.~4, pp. 4289--4301, 2023.

\bibitem{17Chen}
L.~Chen, W.~Liu, Y.~Chen, and W.~Wang, ``Communication-efficient design for
  quantized decentralized federated learning,'' \emph{IEEE Trans. Signal
  Process.}, vol.~72, pp. 1175--1188, 2024.

\bibitem{34Sun}
T.~Sun, D.~Li, and B.~Wang, ``Adaptive random walk gradient descent for
  decentralized optimization,'' in \emph{Proc. 39th Int. Conf. Mach. Learn.
  (ICML)}, 2022.

\bibitem{19Koloskova}
A.~Koloskova, N.~Loizou, S.~Boreiri, M.~Jaggi, and S.~Stich, ``A unified theory
  of decentralized {SGD} with changing topology and local updates,'' in
  \emph{Proc. 37th Int. Conf. Mach. Learn. (ICML)}, 2020.

\bibitem{20Bellet}
A.~Bellet, A.-M. Kermarrec, and E.~Lavoie, ``D-cliques: Compensating for data
  heterogeneity with topology in decentralized federated learning,'' in
  \emph{Proc. 41st Int. Symp. Reliable Distrib. Syst. (SRDS)}, 2022.

\bibitem{25Ayache}
G.~Ayache, V.~Dassari, and S.~E. Rouayheb, ``Walk for learning: A random walk
  approach for federated learning from heterogeneous data,'' \emph{IEEE J. Sel.
  Areas Commun.}, vol.~41, no.~4, pp. 929--940, 2023.

\bibitem{32Sun}
T.~Sun, Y.~Sun, and W.~Yin, ``On markov chain gradient descent,'' in
  \emph{Proc. Adv. Neural Inf. Process. Syst. (NeurIPS)}, 2018.

\bibitem{33Ayache}
G.~Ayache and S.~E. Rouayheb, ``Private weighted random walk stochastic
  gradient descent,'' \emph{IEEE J. Sel. Areas Inf. Theory}, vol.~2, no.~1, pp.
  452--463, 2021.

\bibitem{27Mendieta}
M.~Mendieta, T.~Yang, P.~Wang, M.~Lee, Z.~Ding, and C.~Chen, ``Local learning
  matters: Rethinking data heterogeneity in federated learning,'' in
  \emph{Proc. 2022 IEEE/CVF Conf. Comput. Vis. Pattern Recognit. (CVPR)}, 2022.

\bibitem{29Shi}
Y.~Shi, L.~Shen, K.~Wei, Y.~Sun, B.~Yuan, X.~Wang \emph{et~al.}, ``Improving
  the model consistency of decentralized federated learning,'' in \emph{Proc.
  40th Int. Conf. Mach. Learn. (ICML)}, 2023.

\bibitem{11Hashemi}
A.~Hashemi, A.~Acharya, R.~Das, H.~Vikalo, S.~Sanghavi, and I.~Dhillon, ``On
  the benefits of multiple gossip steps in communication-constrained
  decentralized federated learning,'' \emph{IEEE Trans. Parallel Distrib.
  Syst.}, vol.~33, no.~11, pp. 2727--2739, 2022.

\bibitem{12Heged}
I.~Heged{\H{u}}s, G.~Danner, and M.~Jelasity, ``Gossip learning as a
  decentralized alternative to federated learning,'' \emph{Distrib. Appl.
  Interoperable Syst.}, pp. 74--90, 2019.

\bibitem{13Lian}
X.~Lian, C.~Zhang, H.~Zhang, C.-J. Hsieh, W.~Zhang, and J.~Liu, ``Can
  decentralized algorithms outperform centralized algorithms? {A} case study
  for decentralized parallel stochastic gradient descent,'' in \emph{Proc. Adv.
  Neural Inf. Process. Syst. (NeurIPS)}, 2017.

\bibitem{18Vogels}
T.~Vogels, H.~Hendrikx, and M.~Jaggi, ``Beyond spectral gap: the role of the
  topology in decentralized learning,'' in \emph{Proc. Adv. Neural Inf.
  Process. Syst. (NeurIPS)}, 2022.

\bibitem{21Edoardo}
E.~Gabrielli, G.~Pica, and G.~Tolomei, ``A survey on decentralized federated
  learning,'' \emph{arXiv:2308.04604}, 2023.

\bibitem{22Li}
L.~Li, Y.~Fan, M.~Tse, and K.-Y. Lin, ``A review of applications in federated
  learning,'' \emph{Comput. Ind. Eng.}, vol. 149, p. 106854, 2020.

\bibitem{23Martínez}
E.~T.~M. Beltr{\'a}n, M.~Q. P{\'e}rez, P.~M.~S. S{\'a}nchez, S.~L. Bernal,
  G.~Bovet, M.~G. P{\'e}rez \emph{et~al.}, ``Decentralized federated learning:
  Fundamentals, state of the art, frameworks, trends, and challenges,''
  \emph{IEEE Commun. Surv. Tutor.}, vol.~25, no.~4, pp. 2983--3013, 2023.

\bibitem{24Qu}
Y.~Qu, H.~Dai, Y.~Zhuang, J.~Chen, C.~Dong, F.~Wu \emph{et~al.},
  ``Decentralized federated learning for {UAV} networks: Architecture,
  challenges, and opportunities,'' \emph{IEEE Netw.}, vol.~35, no.~6, pp.
  156--162, 2021.

\bibitem{28McMahan}
B.~McMahan, E.~Moore, D.~Ramage, S.~Hampson, and B.~A. y~Arcas,
  ``Communication-efficient learning of deep networks from decentralized
  data,'' in \emph{Proc. 20th Int. Conf. Artif. Intell. Stat. (AISTATS)}, 2017.

\bibitem{26Li}
T.~Li, A.~K. Sahu, M.~Zaheer, M.~Sanjabi, A.~Talwalkar, and V.~Smith,
  ``Federated optimization in heterogeneous networks,'' in \emph{Proc. Mach.
  Learn. Syst.}, 2020.

\bibitem{35Shah}
S.~M. {Shah} and K.~E. {Avrachenkov}, ``{Linearly Convergent Asynchronous
  Distributed ADMM via Markov Sampling},'' \emph{arXiv:1810.05067}, 2018.

\bibitem{36Ye}
Y.~Ye, H.~Chen, Z.~Ma, and M.~Xiao, ``Decentralized consensus optimization
  based on parallel random walk,'' \emph{IEEE Commun. Lett.}, vol.~24, no.~2,
  pp. 391--395, 2020.

\bibitem{37Karimireddy}
S.~P. Karimireddy, S.~Kale, M.~Mohri, S.~Reddi, S.~Stich, and A.~T. Suresh,
  ``{SCAFFOLD}: Stochastic controlled averaging for federated learning,'' in
  \emph{Proc. 37th Int. Conf. Mach. Learn. (ICML)}, 2020.

\bibitem{38Levin}
D.~A. Levin and Y.~Peres, \emph{Markov chains and mixing times}.\hskip 1em plus
  0.5em minus 0.4em\relax Amer. Math. Soc., 2017, vol. 107.

\bibitem{39Wang}
Z.~Wang, J.~Zhang, T.-H. Chang, J.~Li, and Z.-Q. Luo, ``Distributed stochastic
  consensus optimization with momentum for nonconvex nonsmooth problems,''
  \emph{IEEE Trans. Signal Process.}, vol.~69, pp. 4486--4501, 2021.

\bibitem{40Mota}
J.~F. Mota, J.~M. Xavier, P.~M. Aguiar, and M.~P{\"u}schel, ``D-admm: A
  communication-efficient distributed algorithm for separable optimization,''
  \emph{IEEE Trans. Signal Process.}, vol.~61, no.~10, pp. 2718--2723, 2013.

\bibitem{41Alistarh}
D.~Alistarh, D.~Grubic, J.~Li, R.~Tomioka, and M.~Vojnovic, ``Qsgd:
  Communication-efficient sgd via gradient quantization and encoding,'' in
  \emph{Proc. Adv. Neural Inf. Process. Syst. (NeurIPS)}, 2017.

\bibitem{42Angel}
D.~Angel, R.~M.~J. Jothi, R.~Revathi, and A.~Raja, ``Expander graphs – a
  study,'' \emph{J. Phys. Conf. Ser.}, vol. 1770, no.~1, p. 012078, mar 2021.

\bibitem{43McMahan}
H.~B. McMahan, D.~Ramage, K.~Talwar, and L.~Zhang, ``Learning differentially
  private recurrent language models,'' in \emph{Proc. 6th Int. Conf. Learn.
  Represent. (ICLR)}, 2018.

\end{thebibliography}

\end{document}